\documentclass[a4paper,10pt]{article}
\usepackage[a4paper, total={17cm, 24cm}]{geometry}
\pdfoutput=1

\usepackage[utf8]{inputenc}
\usepackage[english]{babel}
\usepackage[T1]{fontenc}
\usepackage{amsmath}
\usepackage{amsfonts}
\usepackage{hyperref}
\usepackage{tikz}
\usepackage{physics}
\usepackage{comment}
\usepackage{amsthm}
\usepackage{tabularx}
\usepackage{booktabs}
\usepackage{pgfplots}
\usepackage{multirow}
\usepackage{xfrac}
\usepackage[para,flushleft]{threeparttable}
\usepackage{empheq}
\usepackage{cite}
\usepackage{enumitem}
\usepackage{soul}
\usepackage{authblk}

\pgfplotsset{compat=1.7}
\usetikzlibrary{shapes.geometric} 

\makeatletter
\def\TPT@doparanotes{\par
    \prevdepth\z@ \TPT@hsize
    \TPTnoteSettings
    \parindent\z@ \pretolerance 8
    \linepenalty 100000
    \renewcommand\item[1][]{\relax\ifhmode \begingroup
        \hskip .6em plus .2em minus .6em
        \endgroup\fi
        \tnote{##1}\,\ignorespaces}%
    \let\TPToverlap\relax
    \def\endtablenotes{\par}%
}

\usetikzlibrary{quantikz}

\theoremstyle{plain}
\newtheorem{prop}{Proposition}[section]
\newtheorem{lemma}[prop]{Lemma}
\newtheorem{thm}[prop]{Theorem}

\newtheorem{assumption}[prop]{Assumption}

\newtheorem{remark}[prop]{Remark}
\newtheorem{example}[prop]{Example}

\DeclarePairedDelimiter\ceil{\lceil}{\rceil}

\definecolor{green}{RGB}{0,128,0}

\newcommand\mtl[2][none]{%
    \tikz[overlay,remember picture] 
        \node (marker-#2-a) at (0,0.9ex) {};%
}

\newcommand\mbr[2][none]{%
    \tikz[overlay,remember picture] 
        \node (marker-#2-b) at (0.2ex,0.3ex) {};%
    \tikz[overlay,remember picture,inner sep=5pt]
        \node[blend mode=darken,
                rectangle,
                rounded corners,
                fill=yellow,
                opacity=.3,
                fit=(marker-#2-a.center) (marker-#2-b.center)] {};%
    \tikz[overlay,remember picture] 
    \path 
        let
            \p{mtl} = (marker-#2-a), \p{mbr} = (marker-#2-b)
        in
            [x=\x{mbr}, y=\y{mtl}, xshift=1em, yshift=0.5em]
            node[rounded corners,inner sep=2pt,fill=white,opacity=0.8,text opacity=1] at (1, 1) {};%
}

\usetikzlibrary{external}

\begin{document}
\title{Conditions for a quadratic quantum speedup in nonlinear transforms with applications to energy contract pricing}

\author[1]{Gabriele Agliardi}

\author[2]{Corey O'Meara}

\author[3]{Kavitha Yogaraj}

\author[2]{Kumar Ghosh}

\author[4,5]{Piergiacomo Sabino}

\author[2]{Marina Fernández-Campoamor}

\author[2]{Giorgio Cortiana}

\author[6]{Juan Bernabé-Moreno}

\author[7]{Francesco Tacchino}

\author[8]{Antonio Mezzacapo}

\author[8]{Omar Shehab}

\affil{IBM Quantum, IBM Research, Italy}
\affil{E.ON Digital Technology GmbH, Essen, Germany}
\affil{IBM Quantum, IBM Research, India}
\affil{E.ON SE, Essen, Germany}
\affil{University of Helsinki, Finland, Department of Mathematics and Statistics}
\affil{IBM Research Europe, Dublin, Ireland}
\affil{IBM Quantum, IBM Research Europe, Zurich, Switzerland}
\affil{IBM Quantum, IBM Thomas J Watson Research Center, Yorktown Heights, NY, USA}

\maketitle

\begin{abstract}

Computing nonlinear functions over multilinear forms is a general problem with applications in risk analysis.
For instance in the domain of energy economics, accurate and timely risk management demands for efficient simulation of millions of scenarios, largely benefiting from computational speedups.
We develop a novel hybrid quantum-classical algorithm based on polynomial approximation of nonlinear functions, computed through Quantum Hadamard Products, and we rigorously assess the conditions for its end-to-end speedup for different implementation variants against classical algorithms. In our setting, a quadratic quantum speedup, up to polylogarithmic factors, can be proven only when forms are bilinear and approximating polynomials have second degree, if efficient loading unitaries are available for the input data sets.
We also enhance the bidirectional encoding, that allows tuning the balance between circuit depth and width, proposing an improved version that can be exploited for the calculation of inner products.
Lastly, we exploit the dynamic circuit capabilities, recently introduced on IBM Quantum devices, to reduce the average depth of the Quantum Hadamard Product circuit.
A proof of principle is implemented and validated on IBM Quantum systems.
\end{abstract}

\twocolumn

\section{Introduction}

\begin{figure*}[t]
    \centering
    \footnotesize
    \newcommand{\diagw}{26mm}
    \newcommand{\diagh}{12mm}
    \newcommand{\diaggapv}{1mm}
    \newcommand{\ellipsew}{40mm}
    \newcommand{\ellipseh}{13mm}
    \newcommand{\ellipsegapv}{1mm}
    \newcommand{\boxwmid}{30mm}
    \newcommand{\boxh}{16.5mm}
    \newcommand{\boxgapv}{1mm}
    \definecolor{steelblue31119180}{RGB}{31,119,180}
    \includegraphics{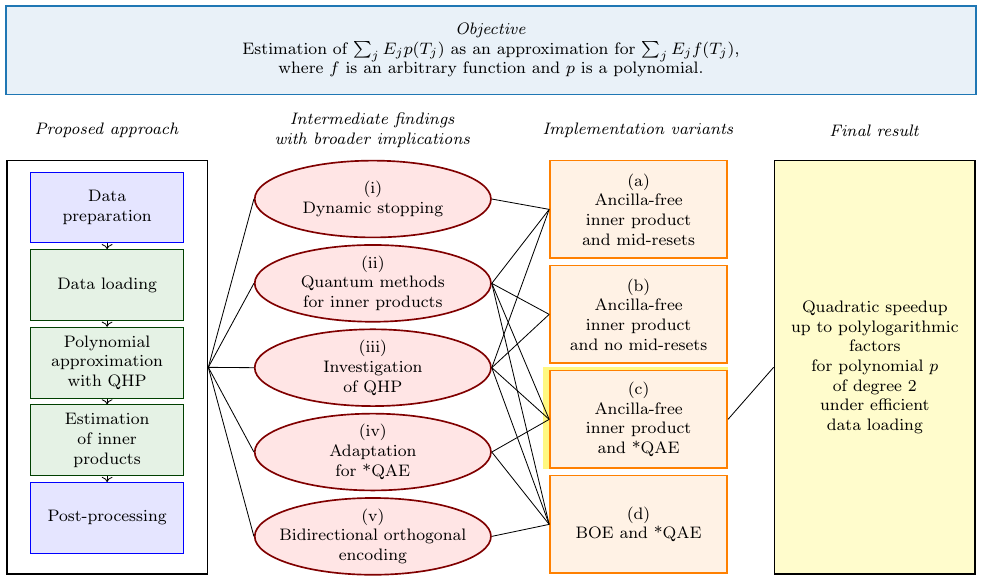}
    \caption{A conceptual representation of the paper, detailed in Subsec.~\ref{subsec:tech-overview}. We propose a hybrid \textit{approach} for the problem resolution. While tuning and implementing the algorithm, we derive five \textit{intermediate findings}, which contribute to different \textit{implementation variants}. Variant~(c) is finally selected for providing the desired speedup, which constitutes our \textit{final result}. The following terms are used in the picture: QHP for Quantum Hadamard Product (see Subsec.~\ref{subsec:qhp}), *QAE for Quantum Amplitude Estimation or equivalent alternatives, BOE for Bidirectional Orthogonal Encoding (see Subsec.~\ref{subsec:data-encoding}).}\label{fig:contrib}
\end{figure*}

Any speedup in the calculation of contract values has a paramount importance in the risk management of the energy industry, enabling real-time planning, finer risk diversification, and faster re-negotiation of hedging contracts~\cite{ghosh_energy_2024}. Gaining a speedup with quantum computing though poses two challenges: on one side, the native linear behavior of quantum operators makes nonlinearities in the contract functions less simple to represent through quantum circuits, and on the other side, classical algorithms with linear scaling in the number of data points are available.
In this work, we design a quantum approach for contract valuation based on polynomial approximations, and provide a comprehensive study of alternate implementations based on tuned combinations of quantum subroutines, comparing their performance. The different variants share the core idea of using Quantum Hadamard Products, which appears as a promising path~\cite{holmes_nonlinear_2021} for nonlinear transformations of data stored in state amplitudes. Finally, we select an optimized implementation with quadratic speedup up to polylogarithmic terms, for second-degree polynomials.

The idea of applying quantum computing techniques to risk management problems was widely explored in the context of finance~\cite{stamatopoulos_option_2020, woerner_quantum_2019, herman_survey_2022} and recently exported to the energy field by the authors in Ref.~\cite{ghosh_energy_2024}. These works have the main objective of accelerating the execution of Monte Carlo methods used for the assessment of risk measures of contract portfolios, exploiting the Quantum Amplitude Estimation (QAE) algorithm that provides a quadratic speedup to classical Monte Carlo simulations~\cite{brassard_quantum_2002, rall_amplitude_2022}.

In this manuscript, we rather focus on the more fundamental task of efficiently calculating the \textit{value} of a contract, and by extension of portfolios, thus providing the data for subsequent risk measure calculation. The contract value is defined\footnote{We introduce here primed notations for non-normalized vectors, consistently with the entire paper.} as
$$v = \sum_{j=0}^{N-1} E'_j \, f(T'_j),$$
where $[E'_j]$ and $[T'_j]$ are energy prices and temperatures respectively, while $f$ is a nonlinear transform.
The conceptual structure of the work is summarized in Fig.~\ref{fig:contrib} and described in Subsec.~\ref{subsec:tech-overview}. Technically, the problem is reduced to the evaluation of the inner product
$$\sum_{j=0}^{N-1} E_j \, p(T_j),$$
where $[E_j]$ and $[T_j]$ are arbitrary normalized vectors, and $p$ is a polynomial approximating the function $f$. This formulation is highly general and allows to address diverse forms of energy contracts, beyond the exemplary one adopted here, as well as other use cases across industries, e.g. finance \cite{bouland_prospects_2020}, climate science \cite{berger_quantum_2021}, etc.

The remainder of the section provides a technical overview of our approach and findings, a collection of prior art on nonlinear transforms, as well as the context of risk analysis. Section~\ref{sec:quantum-approach} describes the proposed approach, the main building blocks, and the different implementation variants, with a focus on the most performing one, namely (c). Section~\ref{sec:complexity} studies the algorithmic complexity from a theoretical stand point, while Section~\ref{sec:experiments} gives experimental results. The final Section~\ref{sec:conclusion} contains conclusions and outlooks. The paper is complemented with a rich set of Appendices, containing multiple implementation variants with the technical details for the underlying theory.

\subsection{Technical overview}\label{subsec:tech-overview}

\begin{figure*}[t]
    \centering
    \footnotesize
    \newcommand{\boxh}{17mm}
    \newcommand{\boxhsmall}{5mm}
    \newcommand{\boxw}{16mm}
    \newcommand{\boxwfat}{27mm}
    \newcommand{\boxwmid}{25mm}
    \newcommand{\gapv}{1mm}
    \newcommand{\gaphbig}{1mm}
    \newcommand{\gaph}{1mm}
    \includegraphics{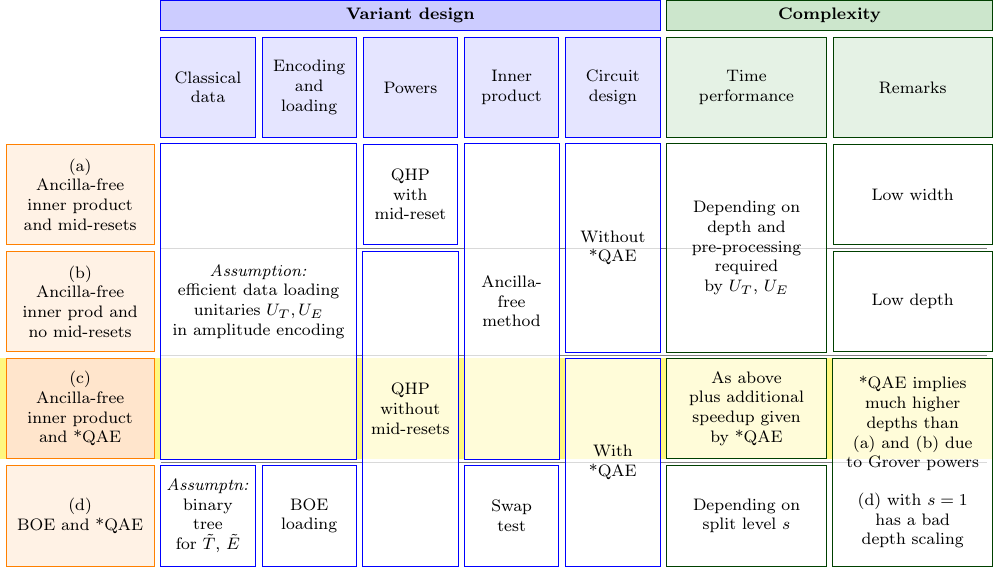}
    \caption{Characteristics of the four proposed implementation variants. Each row is a variant. See Table~\ref{tab:algo-compare-simplified} for quantitative time scaling analysis. Variant~(c) is highlighted as it provides the desired quantum speedup. Terms used in the table are defined in Section~\ref{sec:quantum-approach}: see Subsec.~\ref{subsec:data-encoding} for encoding ($U_T$ and $U_E$ are the loading unitaries for temperatures and prices, BOE stands for Bidirectional Orthogonal Encoding), Subsec.~\ref{subsec:qhp} for powers (QHP is the Quantum Hadamard Product), Subsec.~\ref{subsec:inner} for inner product, Subsec.~\ref{subsec:qae} for the circuit design (*QAE is the Quantum Amplitude Estimation, or any equivalent technique). Refer to Appendix~\ref{appendix:complexity} for a detailed description of the variants and for the extensive complexity analysis comparison.}
    \label{fig:summary}
\end{figure*}

\tikzexternaldisable
\begin{table*}
    \newcolumntype{Y}{>{\centering\arraybackslash}X}
    \centering
    \small
    \begin{tabularx}{\textwidth}{cm{1mm}Yc}
    \toprule
	& &
    Algorithm variant &
    Time complexity
    \\
 
	\midrule
    \multirow{3}{*}{\textit{Classical}} &
    &
    Classical exact &
    $O(N)$
    \\

    & &
    Classical polynomial approx &
    $O_K(N)$
    \\

    & &
    Sampling-based polynomial approx &
    $O_{\beta, K}(\epsilon^{-2} \lg N)$
    \\

	\midrule
    \multirow{2}{*}{\textit{QAE-free}} &
    (a) &
    Ancilla-free inner prod, mid-resets &
    $O_{\beta, K, b} \left(C_{\mathrm{c}, \mathrm{load}}(N) + C_{\mathrm{d}, \mathrm{load}}(N) N^{K-1} \epsilon^{-2} \right)$
    \\

    &
    (b) &
    Ancilla-free inner prod, no mid-res &
    $O_{\beta, K, b} \left(C_{\mathrm{c}, \mathrm{load}}(N) + C_{\mathrm{d}, \mathrm{load}}(N) N^{K-1} \epsilon^{-2} \right)$
    \\
    
	\midrule
    \multirow{2}{*}{\textit{QAE-based}} &
    \mtl[yellow]{c1}(c) &
    Ancilla-free inner prod, *QAE &
    $O_{\beta, K, b} \left( C_{\mathrm{c}, \mathrm{load}}(N) + \left[ C_{\mathrm{d}, \mathrm{load}}(N) + \lg N \right] N^{\frac{K-1}{2}} \epsilon^{-1} \right)$ \mbr{c1}
    \\

    &
    (d) &
    BOE and *QAE &
    $O_{\beta, K, b} \left( N^K \left(2^s +\lg^2 N -s^2 \right) \epsilon^{-1} \right)$
    \\
	\bottomrule
    \end{tabularx}
    \caption{
    Time complexity of the proposed algorithm variants, introduced in Fig.~\ref{fig:summary}, in comparison with the classical benchmarks. Variant~(c) is highlighted as it provides the desired speedup. An extended version is Table~\ref{tab:algo-compare} of Appendix~\ref{appendix:complexity}, that contains other complexity metrics as well as references to the Propositions justifying the results.
    Notations are summarized in Table~\ref{tab:notations}.
    The parameters in the analysis are the data set size $N$, the polynomial degree $K$ ($K \geq 1$), the target precision $\epsilon$ and confidence level $\beta$ such that $\mathbb{P}\left( \abs{ V - v^*} \leq \epsilon \abs{v} \right) \geq \beta$, the coefficients $b_k(\eta)$ in Eq.~\eqref{eq:contractvalue-approx}, and the split level $s=s(N) \in \{1, ..., \lg N\}$ for the BOE.
    Remarkably, the error $\epsilon$ is measured against the exact polynomial evaluation, under the assumption that the polynomial itself is a good approximation of the target volume function. Asymptotic estimations are provided for $\epsilon \to 0$ or $N \to \infty$. Constants affecting asymptotic estimates are marked in the subscript of the big $O$ notation: for instance the notation $O_\beta(\epsilon^{-2})$ is intended for $\epsilon \to 0$ uniformly in $N$, with factors depending only on $\beta$. For readability in the subscripts we use $b$ for $b_k(\eta)$.
    The time scaling is calculated under the additional assumption that norms scale as $\sqrt{N}$, as justified in Subsection~\ref{subsec:assemble-k}. We use the star notation in front of QAE (*QAE) to emphasize that any QAE technique is applicable, such as Iterative QAE (IQAE)~\cite{grinko_iterative_2021} Dynamic QAE~\cite{ghosh_energy_2024}, as long as it shares the same substantial time-scaling with the usual QAE.
    }
    \label{tab:algo-compare-simplified}
\end{table*}
\tikzexternalenable

Our overall approach is sketched in the first column of Fig.~\ref{fig:contrib}. Specifically, we employ polynomial approximation of nonlinear functions and then compute polynomials by resorting to Quantum Hadamard Products (QHP) \cite{holmes_nonlinear_2021,lubasch_variational_2020} for powers. Note that in a classical workflow, a polynomial approximation would also be a step, hidden inside the low-level calculation of the nonlinear function.

In the process of implementing and optimizing the algorithm, we derive the following intermediate findings that are interesting also beyond our specific use case (see second column in Fig.~\ref{fig:contrib}):
\begin{enumerate}[label=(\roman*)]
    \item \textit{Dynamic stopping:} we show that QHP can be implemented via dynamic circuits, introducing the \textit{dynamic stopping}, namely the early abortion of the circuit execution to reduce the average circuit depth, based on mid-circuit measurements which are a relatively new capability in commercial quantum processors,
    \item \textit{Quantum methods for inner products:} we show that the sampling complexity of the swap test for the calculation of an inner product $p$ is unbounded when $p$ tends to $0$, making other techniques (e.g. what we call ancilla-free) more convenient, whenever applicable,
    \item \textit{Investigation of QHP:} we provide a simplified formalization of the QHP, that describes it as a unitary providing a desired output state under a success rate, and we clarify the impact of normalization on the performance of the QHP algorithm,
    \item \textit{Adaptation for QAE:} we show how to adapt our approach based on QHP, in order to make use of the QAE algorithm, appropriately modifying the circuit structure to postpone measurements, and
    \item \textit{Bidirectional Orthogonal Encoding:} we highlight that the encoding strategy proposed in Ref.~\cite{araujo_configurable_2022} is not compatible with swap tests (used for computing the inner product of data vectors), but it can be modified giving rise to the new \textit{Bidirectional Orthogonal Encoding} (BOE), which is suitable for such a task.
\end{enumerate}

By combining the considerations above, we generate multiple implementation variants for our approach, that specifically differ for the encoding and loading subroutines, for the presence of mid-circuit measurements and resets in the circuits, for the quantum subroutine applied to calculate inner products, as well as for the presence of the Quantum Amplitude Estimation (QAE) as a performance booster.
Not all of these combinations are compatible with each other, resulting in four selected variants (see third column in Fig.~\ref{fig:contrib} and details in Fig.~\ref{fig:summary}), namely: (a) ancilla-free inner product and mid-reset, (b) ancilla-free inner product and no mid-reset, (c) ancilla-free inner product and *QAE, (d) BOE and *QAE. The first three variants use the amplitude encoding and the ancilla-free method for inner products, while the last one employs the Bidirectional Orthogonal Encoding (BOE) and the swap test. The notation *QAE remarks that QAE can be replaced by any equivalent alternative, such as the Iterative QAE (IQAE) \cite{grinko_iterative_2021}, the Chebyshev QAE (ChebQAE) \cite{rall_amplitude_2022} or the Dynamic QAE~\cite{ghosh_energy_2024}.

We prove that variant (c) outperforms the others and achieves a quadratic speedup up to polylogarithmic factors when the approximating polynomial has degree~2, if an efficient data loading unitary is available in amplitude encoding. Table~\ref{tab:assumptions}, discussed in Sec.~\ref{sec:complexity}, formalizes the conditions for correctness and speedup, as well as the setting under which the complexity analysis was conducted.
Since the cost of contract valuation is already linear in the number of data points on classical computers, the possibility to improve on that result via quantum algorithms is strictly connected to the ability of loading data efficiently, because in general the data loading procedure has a linear cost~\cite{shende_synthesis_2006, araujo_configurable_2022}. We deeply discuss the impact of data encoding and loading on the performance of the methods.

Our approach is validated on IBM Quantum devices, for small problem instances ($N=4$ data points). In this setting, overall errors are in line with the theory, but the effect of noise is already clearly observable when drilling down to the errors associated to high powers in the approximating polynomial.

\newpage
\subsection{Prior art on nonlinear transformations}
While calculating nonlinear function of data stored in the basis encoding can be done relatively easily by mimicking classical arithmetic~\cite{rieffel_quantum_2011, Nielsen} (see also Ref.~\cite[Appendix~C]{rebentrost_quantum_2018}), quantum circuits that produce nonlinear behaviors of amplitudes are not straightforwardly available, due to the native linearity of quantum operators \cite{terno_nonlinear_1999, holmes_nonlinear_2021, horodecki_limits_2003, schuld_quest_2014}.

Methods for treating nonlinearities either exploit black-box approaches for usage in trainable circuits~\cite{cong_quantum_2019,beer_training_2020}, or represent specific types of nonlinear functions~\cite{leyton_quantum_2008, 10.5555/2011586.2011592}. Ref.~\cite{terno_nonlinear_1999} discusses the role of selective operations and the probability of failure for nonlinear circuits. The work~\cite{maronese_quantum_2022} pursues an objective that is complementary to ours: it estimates an arbitrary (activation) function that takes an inner product as the input, while we approximate the inner product of two vectors, one of which is the output of an arbitrary function.
Ref.~\cite{guo_nonlinear_2021} discusses nonlinear transformations of amplitudes loaded by an oracle through polynomial approximations, and the method was very recently improved in Ref.~\cite{rattew_non-linear_2023} through importance sampling. The paper also shows applications in machine learning and optimization. Such methods though are not immediately applicable to our case, where we want to scale the number of data points: indeed the cited works do not discuss the effect of input renormalization, and conversely focus on data points that increasingly concentrate around $0$ when the problem size grows.

Remarkably, while prior works~\cite{rebentrost_quantum_2018, egger_credit_2021} for Monte Carlo estimations based on QAE typically encode the \textit{probabilities} of each outcome of a discrete distribution, we only assume availability of \textit{samples} from potentially continuous processes of temperatures and prices, and directly encode them. Our method is then agnostic on whether input data are generated by simple independent shots of a single variable, or sophisticated autoregressive models. Nonetheless, for simplicity of analysis, the complexity is discussed here under a set of additional assumptions, slightly more general than independence and identical distribution of the underlying variables (see Table~\ref{tab:assumptions}).

\subsection{Background on portfolio risk analysis in the energy industry}

The gas demand of households or heating can be well described by a deterministic dependency
of gas volumes and weather variables, typically the temperature. 
Standard contracts for private or industrial customers normally entail full-supply gas delivery, without volume constraints. Namely, the customer pays
a fixed price  for the individual consumption whereas the supplier takes
the risk of volume deviations from the projected load profile of the customer. For
example, on cold days the gas demand is likely to be higher than expected, therefore in order to
meet the demand, a supplier has to buy the extra gas needed on the day-ahead market,
typically at prices that are higher than those contracted with the customer. 

In contrast, excess volumes need to be sold by the supplier for lower prices in order to balance the
economic position when temperatures are higher than expected. This leads to costs and risks for the gas supplier which can be managed with either purely temperature-based weather derivatives or with cross-commodity temperature derivative contracts, often called \emph{quantos}. 

Accordingly, risk managers perform risk analysis and compute the fair value and some statistics of the entire weather-related portfolio and of the contracts it consists of. To this end, one defines upfront a joint stochastic model for the gas prices and temperatures and relies on extensive and time-consuming Monte Carlo simulations to estimate the fair value, which is seen as the sample mean, and some risk measures related to the sample statistics~\cite{BenthBenth05, BenthBenth07, cucu_et_al16, cs20_2}. %

Let us consider a simplified weather-related portfolio which depends on gas and temperature. %
For simplicity, we consider one market and one temperature station. On the other hand, a real portfolio would consider multiple markets as well as several weather stations.
Suppose we have a one-year time horizon
from 1-Jan-2022 to 31-Dec-2022, with daily granularity.

We assume that the gas prices and the temperatures, denoted $[E'_j]$ and $[T'_j]$, $j=0, 1, ..., 364$, respectively, are given. Typically they are generated by a two-factor Markov model each, namely the value at time $j$ is updated with $2$ random variables, hence the total number of random draws is $365 \times 2=730$ both for the gas and for the temperature, times the number of markets, leading to thousands. These random variables are usually assumed to be normally distributed and mutually correlated even across gas and temperature. If then one considers $R$ Monte Carlo repetitions one has to generate a random sample linearly scaling with $R$.

As mentioned, the portfolio we focus on consists of fully supply contracts based on which the customer can nominate gas volumes at an agreed sales price, denoted $asp$. These contracts are then implicitly temperature dependent indeed,  their volume can be described by a function $f$ of the temperature which is supposed to model the customers' demand.  Such a function is often a so-called\footnote{Despite the name being widespread in the energy industry, the function is not a sigmoid according to the usual mathematical definition.} \emph{sigmoid} function, namely
\begin{equation} \label{eq:volume}
    f(T') = \frac{A}{1 + \left( \frac{B}{T' - T_0}\right)^C} + D,
\end{equation}
with $T'\le T_0$. Of course, the higher the temperature, the lower is the volume demand and vice-versa.
The parameters $A\ge 0$, $B<0$, $C>1$, $D>0$ are given, in the sense that they are either part of the contract or are estimated from the historical demand of the cluster a specific customer belongs to (eg. households, medium-size enterprises, etc.), and normally one takes $T_0=40$ degree Celsius. %

\begin{table}[t!]
    \newcolumntype{Y}{>{\centering\arraybackslash}X}
    \centering
    \small
    \begin{tabularx}{\columnwidth}{cY}
    \toprule
	Symbol & Meaning \\
	\midrule
	$j=0,...,N-1$ & Index over time\\
	$[T'_j], [E'_j]$ & Temperature and price series\\
	$[T_j], [E_j]$ & Normalized temperature and price series, see Eqs.~\eqref{eq:defn-rho-T}, \eqref{eq:defn-rho-E}\\
	$[\tilde T_j], [\tilde E_j]$ & Normalized square-rooted temperature and price series, see Eq.~\eqref{eq:defn-TandE-bidir}\\
	$\rho_T, \rho_E, \tilde \rho_T, \tilde \rho_E$ & Normalization factors, see Eqs.~\eqref{eq:defn-rho-T}, \eqref{eq:defn-rho-E}, \eqref{eq:defn-rho-TandE-bidir}\\
	$\eta$ & Translation term for temperatures, see Eq.~\eqref{eq:affinity}\\
	$f(T')$ & Volume function, see Eq.~\eqref{eq:volume} \\
	$v$ & Contract value, see Eq.~\eqref{eq:contractvalue} \\
	$v^*$ & Polynomial approximate contract value, see Eq.~\eqref{eq:finalsum} \\
	$k=0,...,K$& Index over monomials in the approximating polynomial, see Eq.~\eqref{eq:poly}\\
	$b_k(\eta)$ & Coefficients in the approximating polynomial, see Eq.\eqref{eq:poly}\\
    $y_k$ & Normalized inner product defined in Eq.~\eqref{eq:defn-yk} \\
    $y'_k$ & Inner product defined in Eq.~\eqref{eq:defn-yprimek} \\
	$\ket{\psi_T}, \ket{\psi^{(k)}_T}$ & Amplitude encoding for $[T_j]$, $[T_j^k]$, see Eqs.~\eqref{eq:psiT} and \eqref{eq:psiTpower}\\
	$\ket{\tilde \psi_{\tilde T}}, \ket{\tilde \psi^{(k)}_{\tilde T}}$ & BOE for $[T_j]$, $[T_j^k]$, see Eqs.~\eqref{eq:psiTandE-bidir} and~\eqref{eq:psiTpower-bidir}\\
	$a_k, \tilde a_k$ & Normalization factor for $[T_j^k]$ and $[\tilde T_j^k]$, see Eqs.~\eqref{eq:defn_ak} and~\eqref{eq:defn_ak-bidir}\\
    $\odot$ & Quantum Hadamard Product, see Eq.~\eqref{eq:qhp}\\
    $s$ & Split level in the BOE \\
    $S$ & Number of samples \\
    $C_\cdot$ & Complexity measures, see Subsec.~\ref{subsec:complex-measures}\\
    $\epsilon, \alpha$ & Acceptable error threshold and associated confidence level for estimators\\
    $\lg$ & Base-2 logarithm\\
    $\norm{\cdot}$ & 2-norm of a vector (Euclidean norm) \\
	\bottomrule
    \end{tabularx}
    \caption{Summary of notations employed in the paper and appendices. Primed letters are used for the non-normalized version of variables, while tildes indicate normalization of according to square root of the input data.}
    \label{tab:notations}
\end{table}

Finally, we are given a vector $\tau'_j$ named \emph{season-normal} temperature which describes our daily expectation of the temperature station.

In our study we focus on the calculation of the change of gross margin, defined as the (unknown random) difference between the net random sales less the random costs at a certain future time $j$ and the planned, therefore known, sales minus cost at the same future time $j$.
Formally, the change (delta) gross margin ($\Delta GM$) of a contract, depending on a certain gas market and a certain temperature can be written as
\begin{equation}
\Delta GM = \sum_{j=0}^{N-1}(f(T'_j) - f(\tau'_j)) (asp - E'_j)).
\label{eq:deltaGM}
\end{equation}
In this study, we present multiple quantum approaches to evaluate $\Delta GM$ given the temperature and energy price vectors, and discuss their potential advantages over classical counterparts under different conditions. More specifically, we provide methods to efficiently compute contract value functions of the form 
\begin{equation}\label{eq:contractvalue}
    v = \sum_{j=0}^{N-1} f(T'_j) E'_j,
\end{equation}
which can be used to reconstruct the expression in Eq.~\eqref{eq:deltaGM}.

Notations adopted across the paper are collected in Table~\ref{tab:notations}.

\begin{figure*}[t!]
    \centering
    \includegraphics{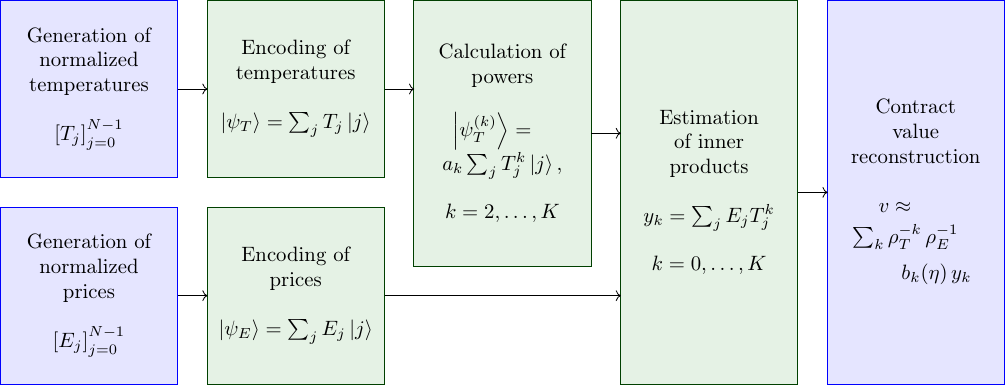}
    \caption{Algorithm for the estimation of the contract value through the polynomial approximation of $f$. The figure instantiates the approach represented in the first column of Fig.~\ref{fig:contrib}. Blue boxes represent classical pre- and post-processing, while green boxes are the core quantum processing.}
    \label{fig:diagram}
\end{figure*}

\section{A hybrid quantum-classical approach}\label{sec:quantum-approach}

Consider a single realization of the time series representing temperatures and prices, namely two real vectors $[T'_j]_{j=0}^{N-1}$ and $[E'_j]_{j=0}^{N-1}$ indexed over time. Suppose they are classically generated, and transformed into the normalized versions $[T_j]_{j=0}^{N-1}$ and $[E_j]_{j=0}^{N-1}$ through the affinities:
\begin{equation}\label{eq:affinity}
    \begin{cases}
    T_j = \rho_T (T'_j - \eta),\\
    E_j = \rho_E E'_j,
    \end{cases}
\end{equation}
where $\eta$ is an appropriate translation to guarantee $T_j>0$ for all $j$ at least in probability (refer to Assumption~\ref{ass:innerprod-positive} below), and $\rho_T$ and $\rho_E$ are suitable scale factors:
\begin{eqnarray}
    \label{eq:defn-rho-T}
    \rho_T^{-2} = \sum_{j=0}^{N-1} (T'_j-\eta)^2\\
    \label{eq:defn-rho-E}
    \rho_E^{-2} = \sum_{j=0}^{N-1} (E'_j)^2.
\end{eqnarray}
Our proposed quantum algorithm, outlined in Fig.~\ref{fig:diagram}, approximates the volume function $f(T')$ by means of a polynomial of degree $K$. %
WLOG, we can write the polynomial in the form
\begin{equation}\label{eq:poly}
    f(T'_j) \approx p(T'_j) := \sum_{k=0}^K b_k(\eta) \; (T'_j - \eta)^k,
\end{equation}
where $b_k(\eta)$ are real coefficients. Consequently the contract value of Eq.~\eqref{eq:contractvalue} writes
\begin{equation}\label{eq:contractvalue-approx}
\begin{split}
    v &= \sum_{j=0}^{N-1} E'_j \; f(T'_j) \\
    & \approx
    v^* := 
    \sum_{j=0}^{N-1} E'_j \; p(T'_j) =
    \sum_{j,k} b_k(\eta) \; E'_j \; (T'_j-\eta)^k.
\end{split}
\end{equation}
We require the following:
\begin{assumption}\label{ass:vapprox-good}
    The polynomial $p$ is a good approximation of the function $f$ under the probability distributions of prices and temperatures, in the following sense:
    $\mathbb{P}\left(\abs{ v^* - v} \leq \epsilon \abs{v} \right) \geq \beta$.
\end{assumption}
A sufficient condition is provided in Prop.~\ref{prop:sufficient-vapprox} and commented in the subsequent remarks of Subsec.~\ref{subsec:assemble-k}.

The algorithm starts by loading normalized temperature and price vectors into a quantum register. A sequence of non-linear transformations is then exploited to calculate the powers $T_j^k$ for all $j=0,\dots, N-1$ and $k=0,\dots, K$. Finally, the inner product of the processed vectors is efficiently evaluated, thus returning 
\begin{equation}\label{eq:defn-yk}
    y_k := \sum_{j=0}^{N-1} E_j \, T_j^k
\end{equation}
for all $k$. The estimation of $y_k$ is boosted by the QAE algorithm, that provides the quadratic speedup (up to polylogarithmic factors).

Finally, the contract value in Eq.~\eqref{eq:contractvalue-approx}
is reconstructed as the summation
\begin{equation}\label{eq:finalsum}
    v^* = \sum_{k=0}^{K} \rho_T^{-k} \, \rho_E^{-1} \; b_k(\eta) \; y_k,
\end{equation}
where $b_k(\eta)$ are known classically. It will be useful to write $v^* = \sum_k b_k(\eta) \; y'_k$ where
\begin{equation}\label{eq:defn-yprimek}
y'_k:= \rho_T^{-k} \, \rho_E^{-1} \; y_k.
\end{equation}
In the remainder of this Section, we discuss in detail the key parts of the quantum algorithm, namely data encoding, power calculation, and inner product. We focus on the main implementation variant (c), while highlighting the key differences where relevant. For additional clarity, we summarize the main variant at the end of the Section, namely in Subsection~\ref{subsec:variants-selected}.

\subsection{Data encoding and loading}\label{subsec:data-encoding}
There are multiple ways to represent the same data set in quantum registers \cite{weigold_encoding_2021,barenco1995elementary, Kumar, Cortes, Plesch, Miszczak, Heinosaari, Bengtsson}, such as the \textit{basis} (aka \textit{digital}, or \textit{equally-weighted}) \textit{encoding}, the \textit{amplitude} (aka \textit{analog}) \textit{encoding}, the \textit{angle encoding}, etc. The subsequent quantum processing techniques are highly dependent on the data encoding protocol. Here, we focus on the amplitude encoding. Additionally, we introduce the \textit{Bidirectional Orthogonal Encoding} (BOE), which can be seen as a variant of the amplitude encoding designed to balance circuit width and depth. %

In the \textit{amplitude encoding}, a normalized classical vector $[D_j]$ of linear size $N$ is represented as
\begin{equation}\label{eq:psiD-ampl}
     \ket{\psi_D} := \sum_{j=0}^{N-1} D_j \ket{j},
\end{equation}
where $\{\ket{j}\}$ are computational basis states of $\lg N$ qubits. While amplitude encoding offers an effective use of memory resources, the exact preparation of an arbitrary state $\ket{\psi_D}$ of the form shown in Eq.~\eqref{eq:psiD-ampl} requires $O\left(N\right)$ operations in the worst case~\cite{shende_synthesis_2006}, thus jeopardizing the benefits of many quantum algorithms. In practice, amplitude encoding remains attractive in combination with approximate data loading techniques (e.g., qGANs \cite{zoufal2019quantum, agliardi_optimal_2022}), in the presence of specific data structures \cite{grover_creating_2002} or under standard quantum memory assumptions \cite{giovannetti_architectures_2008}. In our complexity analysis, the overall computational time is studied as a function of the classical and quantum costs of data loading, without constraining our work to a specific choice of loading technique, in consideration of the constant evolution of such a critical and universal subject. Notice that, for our purposes, loading data from a classical memory or generating data through a quantum simulation are equivalent processes: our only requirement is a unitary which encodes the samples in the amplitudes.

Appendix~\ref{appendix:bidir} details a second version of our protocol that exploits an alternative scheme, based on a newly introduced data encoding that we name \textit{Bidirectional Orthogonal Encoding} (BOE). Even though in the current application the BOE does not achieve the same performance as the amplitude encoding, it may turn useful beyond our use case, as it is designed to balance memory and time resources (i.e., circuit width and depth). On one side, it builds upon the bidirectional encoding %
\cite{araujo_configurable_2022}, and on the other side, it guarantees that side registers are orthogonal \cite{araujo2021divide}, thus enabling subsequent processing through swap tests. The BOE is therefore characterised by states of the form
\begin{equation}
    \ket{\tilde\psi_{D}} := \sum_{j=0}^{N-1} D_j \ket{j} \ket{\phi^D_j}, \label{eq:psiD-bidir}
\end{equation}
where the second register is auxiliary, and contains states entangled to the main register, with the orthonormal property $\braket{\phi^D_i}{\phi^D_j} = \delta_{ij}$. Similar to the bidirectional encoding, the BOE requires classical data to be organized in a binary tree structure (Fig.~\ref{fig:binary-tree}). The split level steers the balance between circuit depth and width: for $s=1$, the width is $O(N)$ and the depth $O(\lg^2 N)$ (depth efficient), while for $s=\lg N$ the width is $O(\lg N)$ and the depth $O(N)$ (memory efficient, akin to amplitude encoding).
The technique is less favourable than the amplitude encoding in our context as Eq.~\eqref{eq:finalsum} gets replaced by Eq.~\eqref{eq:finalsum-bidir}, that has a quadratic dependence on the norms instead of a linear one, implying a worse scaling with $N$. Refer to Subsection~\ref{subsec:variants-selected} for more information on the error scaling.
The data loading procedure is detailed in Subsection~\ref{subsec:bidir}.

\subsection{Non-linear transformation with QHP}\label{subsec:qhp}
\begin{figure}[t]
    \centering
    \tikzsetnextfilename{qhp}
    \begin{quantikz}
    \lstick{$\ket{\psi_0}$} & \qwbundle{n} & \ctrl{1} \gategroup[wires=2,steps=2,style={dashed, rounded corners, fill=black!10}, background]{} & \qw{} & \rstick{$\ket{\psi_0 \odot \psi_1}$} \qw{}\\
    \lstick{$\ket{\psi_1}$} & \qwbundle{n} & \targ{} & \meterD{\ketbra{0}{0}}
    \end{quantikz}
    \caption{Circuit implementation of the Quantum Hadamard Product producing the state $\ket{\psi_0 \odot \psi_1}$ when the bottom register is measured in the state $\ket{0}^{\otimes n}$. The success probability is $a^{-2}$, where $a$ is the normalization constant appearing in Eq.~\eqref{eq:qhp}. The CNOT notation on registers indicates pairwise application on belonging qubits (namely, a CNOT between the first qubits of both register, a CNOT between the second qubits of both registers, etc.)}
    \label{fig:qhp}
\end{figure}
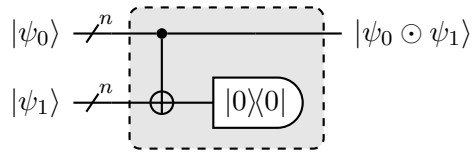
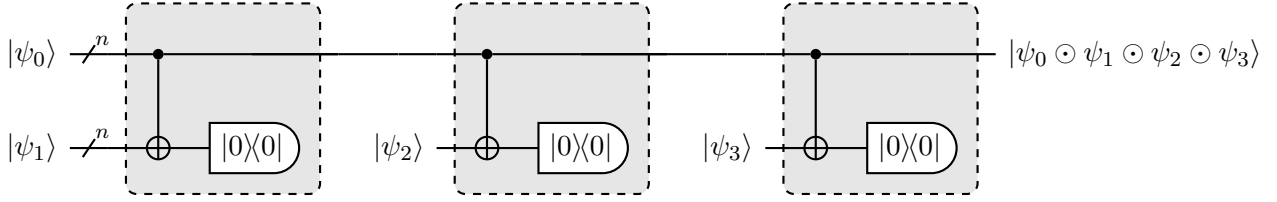
\begin{figure*}[t!]
    \centering
    \tikzsetnextfilename{qhp-midreset}
    \begin{quantikz}
    \lstick{$\ket{\psi_0}$} & \qwbundle{n} & \ctrl{1}\gategroup[wires=2,steps=2,style={dashed, rounded corners, fill=black!10}, background]{} & \qw{} & \hphantomgate{wide}\qw{} & \qw{} & \ctrl{1}\gategroup[wires=2,steps=2,style={dashed, rounded corners, fill=black!10}, background]{} & \qw{} & \hphantomgate{wide}\qw{} & \qw{} & \ctrl{1}\gategroup[wires=2,steps=2,style={dashed, rounded corners, fill=black!10}, background]{} & \qw{} & \rstick{$\ket{\psi_0 \odot \psi_1 \odot \psi_2 \odot \psi_3}$} \qw{} \\
    \lstick{$\ket{\psi_1}$} & \qwbundle{n} & \targ{} & \meterD{\ketbra{0}{0}} & & \lstick{$\ket{\psi_2}$} & \targ{} & \meterD{\ketbra{0}{0}} & & \lstick{$\ket{\psi_3}$}  & \targ{} & \meterD{\ketbra{0}{0}}
    \end{quantikz}
    \caption{QHP of $k=4$ vectors with mid-measurements and mid-resets, requiring $k-1=3$ iterations.}
    \label{fig:qhp-midreset}
\end{figure*}
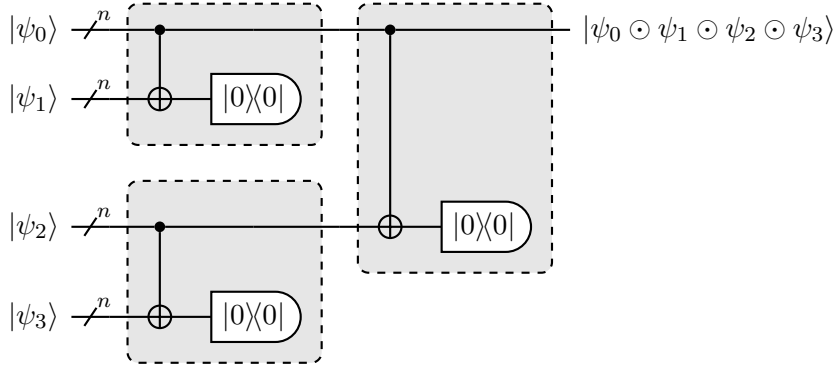
\begin{figure*}[t!]
    \centering
    \tikzsetnextfilename{qhp-nomidreset}
    \begin{quantikz}
    \lstick{$\ket{\psi_0}$} & \qwbundle{n} & \ctrl{1} \gategroup[wires=2,steps=2,style={dashed, rounded corners, fill=black!10}, background]{} & \qw{} & \qw{} & \ctrl{3} \gategroup[wires=4,steps=2,style={dashed, rounded corners, fill=black!10}, background]{} & \qw{} & \rstick{$\ket{\psi_0 \odot \psi_1 \odot \psi_2 \odot \psi_3}$} \qw{}\\
    \lstick{$\ket{\psi_1}$} & \qwbundle{n} & \targ{} & \meterD{\ketbra{0}{0}} & & & \\
    &&&&&&&\\
    \lstick{$\ket{\psi_2}$} & \qwbundle{n} & \ctrl{1} \gategroup[wires=2,steps=2,style={dashed, rounded corners, fill=black!10}, background]{} & \qw{} & \qw{} & \targ{} & \meterD{\ketbra{0}{0}} & & \\
    \lstick{$\ket{\psi_3}$} & \qwbundle{n} & \targ{} & \meterD{\ketbra{0}{0}} & & &
    \end{quantikz}
    \caption{QHP of $k=4$ vectors with no mid-measurements and no mid-resets, requiring $\lg k = 2$ iterations (note that all measurements can be postponed to the end). This implementation is adopted in \cite{lubasch_variational_2020}, without being named QHP.}
    \label{fig:qhp-nomidreset}
\end{figure*}
Let us now discuss the non-linear transformation, namely the calculation of monomial powers. Assume temperatures can be encoded in the amplitudes of a quantum state. More precisely, let $[T_j]_{j=0}^{N-1}$ be a normalized time series of temperatures, namely a vector of real numbers such that $\sum_j T^2_j = 1$. We can assume $[T_j]_{j=0}^{N-1}$ is obtained by simulated temperatures through a translation and re-scaling as in Eq.~\eqref{eq:affinity}. We are able to produce the state
\begin{equation}\label{eq:psiT}
    \ket{\psi_T} := \sum_{j=0}^{N-1} T_j \ket{j}.
\end{equation}
In order to evaluate the approximating polynomial, we need the powers $T_j^k$, in the form of the state
\begin{equation}\label{eq:psiTpower}
    \ket{\psi^{(k)}_T} := a_k \sum_{j=0}^{N-1} T_j^k \ket{j},
\end{equation}
for $k=1,\dots,K$, where $a_k$ is the appropriate scale factor, 
\begin{equation}\label{eq:defn_ak}
    a_k:=\left(\sum_{j=0}^{N-1} T_j^{2k} \right)^{-\sfrac{1}{2}}.
\end{equation}
Note that $a_k$ is an increasing finite sequence, when $k$ grows. In particular, $a_k \geq 1$ for all $k$ since $a_1=1$.

For the calculation of the powers, we resort to the non-linear transformation known as Quantum Hadamard Product (QHP), recently introduced by Holmes \textit{et al.}~\cite{holmes_nonlinear_2021}.
Given two states $\ket{\psi_0}$ and $\ket{\psi_1}$, their QHP is the state
\begin{equation}\label{eq:qhp}
    \ket{\psi_0 \odot \psi_1} := a \sum_{j=0}^{N-1} \braket{\psi_0}{j} \braket{\psi_1}{j} \ket{j},
\end{equation}
which, in the circuit of Fig.~\ref{fig:qhp}, is obtained in the first register as the result of a post-selection conditioned on the second register being in $\ket{0}^{\otimes n}$. The probability to measure $\ket{0}$ in the second register, i.e. the success rate for the calculation of the QHP, equals $a^{-2}$. It is important to mention that the QHP was originally defined~\cite{holmes_nonlinear_2021} as the (not necessarily normalized) weighted state arising from the application of a rank-1 measurement operator $\ketbra{0}{0}^{\otimes n}$, i.e., as the output of the circuit in Fig.~\ref{fig:qhp} without post-selection. Our simpler formulation in Eq.~\eqref{eq:qhp} is, nevertheless, fully equivalent and well suited for our purposes.

The QHP can easily be iterated to compute the Hadamard product of multiple vectors, hence producing higher-order powers. In particular, states of the form $\ket{\psi^{(k)}_T}$ as defined in Eq.~\eqref{eq:psiTpower} are obtained by loading $k$ copies of $\ket{ \psi_T }$ as input states and calculating their QHPs:
\begin{equation}\label{eq:qhp-k}
    \ket{\psi_T^{(k)}} := \ket{\bigodot^k \psi_T}.
\end{equation}
Figures~\ref{fig:qhp-midreset} and~\ref{fig:qhp-nomidreset} show two implementations for the QHP of four states. Note that the former requires mid-circuit measurements, and has higher depth but lower width than the latter. The former is also suited for an additional improvement, that we call \textit{dynamic stopping}: thanks to the feature of `dynamic circuits' recently made available on commercial hardware \cite{pressrelease_dynamic_2021, dynamic_circuits}, the execution of the circuit can be aborted right after a mid-circuit measurements if the measurement does not output a $0$, as this corresponds to a failure in the QHP. This way, dynamic stopping allows reducing the average circuit depth.

In the main algorithm variant (c), highlighted in Fig.~\ref{fig:summary} and in Table~\ref{tab:algo-compare-simplified}, we want to exploit QAE and related techniques for improved performance, and therefore we need a unitary circuit, thus forcing us to the adoption of the implementation without mid-resets. Variants that include mid-resets and dynamic stopping are discussed in the Appendices~\ref{appendix:inner-prod} and~\ref{appendix:complexity}.

\subsection{Computing inner products}\label{subsec:inner}

The third quantum step of the algorithm (Fig.~\ref{fig:diagram}) is the calculation of the inner products
\begin{equation}
    y_k := \sum_{j=0}^{N-1} E_j T_j^k,
\end{equation}
for $k = 1,\dots,K$.

\begin{figure}[t]
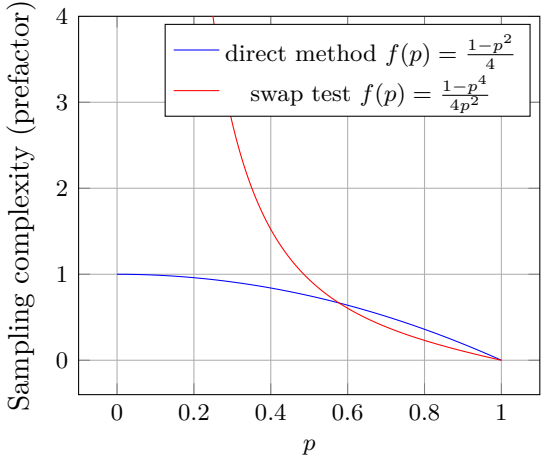

    \centering
    \include{Figures/innerproducts0}
    \caption{The dependence of the sampling complexity on the target inner product $p$ between two statevectors. For the swap test, the complexity is not upper bounded when $p \to 0$. Refer to Prop.~\ref{prop:errorbound-swap} for the details. A numerical evidence can be found in Fig.~\ref{fig:inner-compare}. %
    }
    \label{fig:inner-p}
\end{figure}

Again, let us start assuming that we have two real vectors mapped to quantum states of the form $\ket{\psi_\alpha}=\sum_j \psi^{(\alpha)}_j\ket{j} $ via amplitude encoding. Multiple quantum techniques for the computation of the absolute inner product $\left| \braket{\psi_0}{\psi_1} \right| = \left| \sum_j \psi^{(0)}_j \psi^{(1)}_j \right| := p$ are known. Specifically, in Appendix~\ref{appendix:inner-prod}, we describe and compare the so-called \textit{swap test} and the \textit{ancilla-free method}. The latter is simply grounded on the observation that, given the two loading unitaries $U_T$ and $U_E$, the expression $\matrixel{0}{U_E^\dag U_T}{0}$ calculates the inner product of the two vectors.

Now, the swap test provides an estimation of $ \frac{1}{2}+\frac{1}{2} p^2$, while the ancilla-free method directly outputs $p^2$. As an effect, one obtains that with the ancilla-free method, the sampling complexity is independent of $p$, while with the swap test it is unbounded for $p \to 0$, as shown numerically in Fig.~\ref{fig:inner-compare}. In other words, it is impossible to estimate the number of required samples to achieve a given precision, without knowing (an estimation of) the result itself, see Fig.~\ref{fig:inner-p}. For this reason we choose, for all results discussed in the main text, the ancilla-free method in association with the amplitude encoding. Refer to the Appendix~\ref{appendix:inner-prod} for more information on both techniques, and in particular to Table~\ref{tab:inner-compare} for an in-depth comparison.

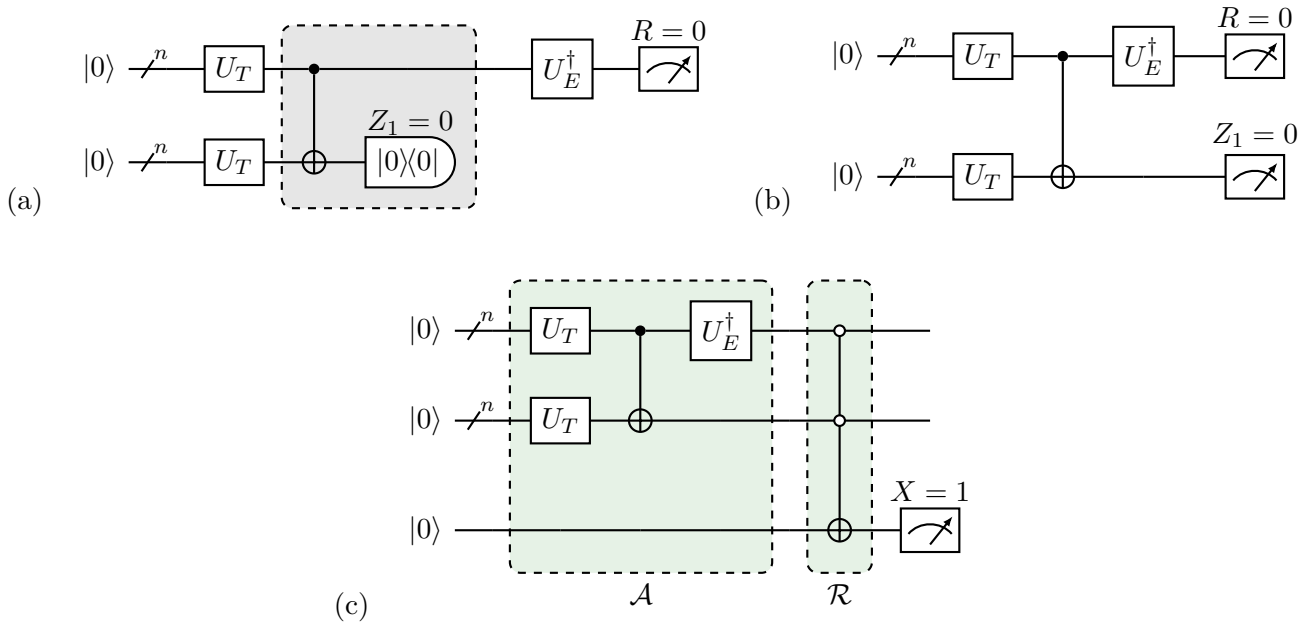
\begin{figure*}[t]
    \centering
    (a)
    \tikzsetnextfilename{qhp-direct1}
    \begin{quantikz}
    \lstick{$\ket{0}$} & \qwbundle{n} & \gate{U_T} & \ctrl{1} \gategroup[wires=2,steps=2,style={dashed, rounded corners, fill=black!10}, background]{} & \qw{} & \qw{} & \gate{U_E^\dag} & \meter{$R=0$} \\
    \lstick{$\ket{0}$} & \qwbundle{n} & \gate{U_T} & \targ{} & \meterD{\ketbra{0}{0}} \gategroup[wires=1,steps=1,style={draw=none}, label style={label position=above, anchor=north}]{$Z_1=0$}
    \end{quantikz}
    \hfill
    (b)
    \tikzsetnextfilename{qhp-direct2}
    \begin{quantikz}
    \lstick{$\ket{0}$} & \qwbundle{n} & \gate{U_T} & \ctrl{1} & \gate{U_E^\dag} & \meter{$R=0$} \\
    \lstick{$\ket{0}$} & \qwbundle{n} & \gate{U_T} & \targ{} & \qw{} & \meter{$Z_1=0$}
    \end{quantikz}\\[8mm]
    (c)
    \tikzsetnextfilename{qhp-direct3}
    \begin{quantikz}
    \lstick{$\ket{0}$} & \qwbundle{n} & \gate{U_T} \gategroup[wires=3,steps=3,style={dashed, rounded corners, fill=green!10}, background, label style={label position=below, anchor=north, yshift=-0.2cm}]{$\mathcal{A}$} & \ctrl{1} & \gate{U_E^\dag} & \qw{} & \octrl{1} \gategroup[wires=3,steps=1,style={dashed, rounded corners, fill=green!10}, background, label style={label position=below, anchor=north, yshift=-0.2cm}]{$\mathcal{R}$}   & \qw{} \\
    \lstick{$\ket{0}$} & \qwbundle{n} & \gate{U_T} & \targ{} & \qw{} & \qw{} & \octrl{1} & \qw{} \\
    \lstick{$\ket{0}$} & \qw{} & \qw{} & \qw{} & \qw{} & \qw{} & \targ{} & \meter{$X=1$}
    \end{quantikz}
    \caption{An exemplary demonstration of the algorithm based on QHP and the ancilla-free method, for the calculation of the inner product between temperatures to the power $k$ and prices, when $k=2$. $U_T$ and $U_E$ are the loading unitaries for $T$ and $E$ respectively. In (a), the algorithm in its original formulation, where the gray box highlights the QHP. In (b), the version without mid-measurements. All measurements are deferred to the end. In (c), through a multi-controlled NOT, a single qubit needs measurement. This suggests the definition of unitaries $\mathcal{A}$ and $\mathcal{R}$, marked in green, that can be fed to a QAE algorithm for efficient estimation, as described in Prop.~\ref{prop:errorbound-qae-qhp+direct}.}
    \label{fig:qhp+direct}
\end{figure*}

Assembling all the algorithm components introduced so far, the inner product $\braket{\psi_T^{(k)}}{\psi_E}$ can be estimated, provided that the following holds
\begin{assumption}\label{ass:innerprod-positive}
For all powers $k=0, \dots, K$, we assume that the inner product of the normalized vectors $\sum_j T_j^k E_j$ is not null, and that its sign can be determined a priori.
\end{assumption}
For simplicity, we rely on the following stronger version:
\begin{assumption}\label{ass:innerprod-positiveterms}
For all  $j = 0, \dots, N-1$, we assume $T_j> 0$ and $ E_j > 0$, and  $\sum_j T_j^k E_j$ is not null (and therefore positive).
\end{assumption}

As a remark, we highlight that Assumptions~\ref{ass:innerprod-positive} and~\ref{ass:innerprod-positiveterms} may be overcome in variant~(c) by resorting to sign-sensitive QAE methods~\cite{manzano_real_2023}. For ease of treatise and comparison, and in consideration that input data would still require a normalization through an affitinity in the form of Eq.~\eqref{eq:affinity}, we nevertheless enforce Assumption~\ref{ass:innerprod-positiveterms} across the manuscript.

\subsection{Quantum Amplitude Estimation techniques}\label{subsec:qae}
The final building block for our main implementation variant (c) is the Quantum Amplitude Estimation (QAE), that is typically described as the quantum alternative to classical Monte Carlo methods, providing a quadratic speedup against the classical version, by leveraging Montanaro's approach~\cite{montanaro_quantum_2015}. In this case, we are rather interested in the fact that it also provides a quadratic speedup in terms of precision $\epsilon$ when estimating the amplitude of a state, against the straight-forward averaging of measuring repeated shots of a quantum circuit (refer again to Ref.~\cite{rebentrost_quantum_2018} for a synthetic explanation). We often use the star notation in front of QAE (*QAE) to emphasize that any QAE technique is applicable, such as Iterative QAE (IQAE)~\cite{grinko_iterative_2021}  or Dynamic QAE~\cite{ghosh_energy_2024}, as long as it shares the same substantial time-scaling with the usual QAE.

\subsection{Assembling the main variant and adapting the circuit for QAE}\label{subsec:variants-selected}

We assume data is available in the amplitude encoding. The circuit is fed with $k$ copies of the temperature state $\ket{\psi_T}$ in Eq.~\eqref{eq:psiT}, so that the power Eq.~\eqref{eq:psiTpower} of the temperature vector is calculated trough the QHP. Moreover, we resort to the QHP implementation without mid-circuit resets, as depicted in Fig.~\ref{fig:qhp-nomidreset}. Afterwards, under the assumption that a loading unitary $U_E$ for $\ket{\psi_E}$ is known, the inner products $y_k$ can be calculated with the ancilla-free method. The circuit described so far can be slightly modified as shown in Fig.~\ref{fig:qhp+direct}(c), to be fed to a *QAE technique. Finally the volume $v$ is reconstructed classically via Eq.~\eqref{eq:finalsum}.

It should be noted now that various monomials in Eq.~\eqref{eq:poly} affect the overall error differently, as a result of the factors $\rho_T^{-k} \rho_E^{-1} b_k(\eta)$ in Eq.~\eqref{eq:finalsum}. In order to achieve a relative error $\leq \epsilon$ on the final result $v$, the absolute error of the monomial of index $k$ must be controlled by
\begin{equation}\label{eq:eps-k}
    \epsilon_k := \epsilon \; \rho_T^{k-1} \; K^{-1} \; \abs{b_k(\eta)}^{-1}.
\end{equation}
The reason for such scaling is that the monomial of power $k$ decays as $\rho_T^k$ when $N$ grows, meaning that, when mapping back $y_k$ to $y'_k$, the absolute error $\epsilon_k$ associated to $y_k$ becomes $\epsilon'_k = \rho_T^{-k} \epsilon_k$ in terms of $y'_k$. On the other side, the final result $v$ only scales as $\rho_T^{-1}$ due to bilinearity. The relative error, defined as the absolute error divided by the target value, then scales as $\rho_T^{1-k} \epsilon_k$. %
The argument is formalized and discussed in Subsection~\ref{subsec:assemble-k}.
Consequently, $\epsilon_k$ is the threshold used for the *QAE.

Theoretical details for the implementation outlined above are presented in the Appendix~\ref{appendix:qae-ampl}, building on the version without *QAE introduced in Appendix~\ref{appendix:inner-prod} (specifically in Subsec.~\ref{subsec:qhp+inner}).

\newpage
\section{Quantum complexity analysis}\label{sec:complexity}
In this section we discuss the space and time complexity of the algorithm.

The time complexity of the quantum algorithm is summarized by the sum of the classical runtime and quantum runtime. The latter is defined in turn as the sum over $k=0, ..., K$ ($k$ being the index of the monomials forming the approximating polynomial $p$) of the circuit depths (number of layers in the quantum circuit) multiplied by the respective sampling complexities (number circuit executions). 

\begin{table}[t]
    \fbox{\parbox{\dimexpr\columnwidth-2\fboxsep-2\fboxrule\relax}{%
    \small%
    \vspace{-4mm}
    \paragraph{Conditions for method correctness.} All method variants provide an estimate of the contract value with relative error smaller than a desired threshold with high probability if:
    \begin{itemize}
        \item The polynomial $p$ is a good approximation of $f$, see Assumption~\ref{ass:vapprox-good}.
        \item Normalized temperatures and prices are available, see Eq.~\eqref{eq:affinity} and Assumption~\ref{ass:innerprod-positiveterms}.
        \item The algorithm is executed for a sufficiently large number of shots (see details in Table~\ref{tab:algo-compare}, raw `Samples $C_S(k)$').
    \end{itemize}
    \paragraph{Assumptions in complexity analysis.} The complexity analysis is performed in the following simplified setting:
    \begin{itemize}
        \item The 2-norms of the input vectors scale as $\sqrt{N}$ when $N$ grows, and the target contract value scales as $N$, see Assumptions~\ref{ass:norm-scaling} and~\ref{ass:v-scaling}.
    \end{itemize}
    \paragraph{Conditions for asymptotic speedup.} Variant (c) of the algorithm provides an asymptotic speedup under the following hypotheses:
    \begin{itemize}
        \item $p$ has degree $K \leq 2$.
        \item There exists an efficient loading method in the amplitude encoding, namely $C_{\mathrm{c}, \mathrm{load}}(N) \ll N$ and quantum depth $C_{\mathrm{d}, \mathrm{load}}(N) \ll N^{1/2}$ (see Subsec.~\ref{subsec:complexity-results}).
    \end{itemize}
    }}
    \caption{Summary of conditions and assumptions.}
    \label{tab:assumptions}
\end{table}

\subsection{Setting}\label{subsec:complexity-assumptions}
Time complexity is studied in the limit of the horizon length $N \to \infty$, for simplicity, since an asymptotic speedup implies at least a break-even point where quantum provides an advantage. The characterization of the break-even point is left for future work. Additionally, to perform the asymptotic analysis, we need to focus on the relative error and to introduce further hypotheses on how the involved norms scale with $N$. Assumptions are collected in Table~\ref{tab:assumptions} and commented below. It is worth emphasizing though that the analysis in the Appendices is conducted in full generality, before applying the Assumptions of this Subsection: in particular, the term $r_k$ defined in Eq.~\eqref{eq:defn-rk} and used in the full detailed complexity analysis of Table~\ref{tab:algo-compare}, accounts for the ratio between the result norm and the input norms.

\begin{assumption}\label{ass:norm-scaling}
The Euclidean norms $\rho_T^{-1}$ and $\rho_E^{-1}$ of the input vectors scale as $\sqrt{N}$, when $N$ grows.
\end{assumption}
For instance, the previous Assumption holds when all temperatures (prices) are independently sampled from the same random variable $T'$ ($E'$), whatever its distribution, as long as it is $L^2$ (see Remark~\ref{rem:norm-scaling} and Example~\ref{ex:error-scaling} in Subsection~\ref{subsec:assemble-k}). The independence of sampled temperatures and prices along time, though, is a strong assumption that is hardly verified in practice. We introduced it here as an intuitive justification for an assumption that holds in broader contexts. Sufficient conditions for the former assumption, related to stationary processes, as well as an asymptotic study under different assumptions, are sparks for future developments of this work, that may not be limited to the energy economics field.

Similarly, we enforce the following:
\begin{assumption}\label{ass:v-scaling}
The target contract value $v$ scales linearly with $N$, when $N$ grows.
\end{assumption}
The previous Assumption is verified, for example, when temperatures are independently sampled from a same random variables, as well as prices, under the regularity condition $f(T') E' \in L^1$, see Remark~\ref{rem:v-scaling}.

As shown in Table~\ref{tab:algo-compare-simplified}, quantum algorithms are benchmarked against three classical alternatives (see also Subsection~\ref{subsec:benchmarks}), increasingly closer to the quantum approach. The first is constituted by the exact evaluation of the function and the subsequent naive application of the inner products. The second involves the replacement of $f$ by the approximating polynomial $p$, providing advantages when the same inputs are used in multiple functions, as we discuss at the end of this section. The third is a sampling-based approach and compares against the quantum variant (d) as they both require a classical binary tree structure.

\subsection{Results}\label{subsec:complexity-results}

If an efficient loading procedure in amplitude encoding is known, having classical cost $C_{\mathrm{c}, \mathrm{load}}(N) \ll N$ and quantum depth $C_{\mathrm{d}, \mathrm{load}}(N) \ll N^{1/2}$, then the proposed main variant (c) has a quantum speedup against the classical case when $K=2$. If instead $C_{\mathrm{c}, \mathrm{load}}(N) = O(1)$ and $C_{\mathrm{d}, \mathrm{load}}(N) = O(1)$, for $K=3$ the quantum complexity is comparable to the classical one, up to logarithmic factors.

More in general, for any $K$, the time complexity of the main variant (c) is
$$O \left( C_{\mathrm{c}, \mathrm{load}}(N) + \epsilon^{-1} \left[ C_{\mathrm{d}, \mathrm{load}}(N) + \lg N \right] N^{\frac{K-1}{2}} \right),$$
where $\epsilon$ is the acceptable relative error threshold. This cost compares with classical benchmarks of $O(N)$, implying an advantage for $K \leq 2$ and efficient data loading. Table~\ref{tab:algo-compare-simplified} (and the extended Table~\ref{tab:algo-compare}) show that the other selected variants have higher costs. The analyzed QAE-free methods, though, have the same asymptotic time complexity of the classical techniques when $N$ grows, if $K = 2$ and data loading is performed in $O(1)$ depth, thus nullifying any advantage.

As a final remark on time, consider what happens if the same temperature series and the same price series are used for multiple contracts, that differ for the definition of the volume function $f$ (for instance, constants $A$, $B$, $C$, $D$ in Eq.~\eqref{eq:volume} are varied). Then the polynomial approximation approach, both quantumly or classically implemented, is advantageous as $y_k$ in Eq.~\eqref{eq:finalsum} can be computed just once for the different versions of $f$. Only the polynomial coefficients $b_k(\eta)$ and the sum in Eq.~\eqref{eq:finalsum} need to be recomputed for each $f$. Specifically for the quantum case, the circuits are executed once, independently of how many $f$ functions are being evaluated.

The space complexity is represented by the \textit{circuit width} $C_{\mathrm{w}}$, namely the total number of qubits. For our main variant (c), it is $O(k \lg N)$.

The derivation of the theoretical bounds and their validation on simulators is contained in the Appendices, that consider one variant at a time. The detailed comparison of the variants, under multiple complexity metrics, can be found in Appendix~\ref{appendix:complexity} and specifically in Table~\ref{tab:algo-compare}.

\section{Experimental demonstration}\label{sec:experiments}

\begin{table*}
    \centering
    \small
    \newcolumntype{Y}{>{\centering\arraybackslash}X}
    \begin{tabularx}{\textwidth}{YcYYYYcYYYY}
    \toprule
    \multirow{2}{*}{\begin{tabular}{@{}c@{}}Overall \\relative \\error\end{tabular}} && \multicolumn{4}{c}{\vspace{2mm} Relative error for $y'_k$} && \multicolumn{4}{c}{\vspace{2mm} Error for the contribution of $y'_k$} \\
    \cline{3-6} \cline{8-11}
    && $k=0$ & $k=1$ & $k=2$ & $k=3$ && $k=0$ & $k=1$ & $k=2$ & $k=3$ \\
    \midrule
    1.36\% && 1.71\% & 0.05\% & 16.52\% & 89.86\% && 2.23\% & 0.01\% & 0.83\% & 0.05\% \\
    1.02\% && 1.71\% & 0.37\% & 21.25\% & 76.94\% && 2.23\% & 0.09\% & 1.07\% & 0.04\% \\
    4.97\% && 1.67\% & 0.34\% & 57.79\% & 82.79\% && 2.18\% & 0.09\% & 2.92\% & 0.04\% \\
    4.76\% && 2.94\% & 0.07\% & 17.05\% & 78.73\% && 3.84\% & 0.02\% & 0.86\% & 0.04\% \\
    0.81\% && 0.20\% & 0.05\% & 20.57\% & 90.48\% && 0.26\% & 0.01\% & 1.04\% & 0.05\% \\
    \midrule
    2.58\% && 1.65\% & 0.18\% & 26.64\% & 83.76\% && 2.15\% & 0.04\% & 1.35\% & 0.04\% \\
    (0.0209) && (0.0097) & (0.0016) & (0.1754) & (0.0622) && (0.0126) & (0.0004) & (0.0080) & (0.0000) \\
    \bottomrule
    \end{tabularx}
    \caption{Errors arising in the estimation, in five independent algorithm runs of the algorithm with $N=4$ and $K=3$. The algorithm is modified to set $\epsilon_k = 0.04$ as described in Sec.~\ref{sec:experiments}. The first column represents the overall relative error $\abs{\frac{V-v^*}{v^*}}$. The second group of columns contains the relative error for the estimators of $y'_k$, namely $\abs{\frac{Y'_k-y'_k}{y'_k}}$. The last group contains the same errors as the second, but rescaled according to the contributions they give to the overall estimation, namely $\abs{\frac{b_k(\eta) \; (Y'_k-y'_k)}{v^*}}$. On the bottom two rows, the average and standard deviation of the values above. In the table, the coefficients $b_k(\eta)$ are those for the Taylor expansion centered in $0$, see Sec.~\ref{sec:experiments}. Underlying values are plotted in Fig.~\ref{fig:results}, gray area.}
    \label{tab:experiment-results-taylor}
\end{table*}

\begin{table*}
    \centering
    \small
    \newcolumntype{Y}{>{\centering\arraybackslash}X}
    \begin{tabularx}{\textwidth}{YcYYYYcYYYY}
    \toprule
    \multirow{2}{*}{\begin{tabular}{@{}c@{}}Overall \\relative \\error\end{tabular}} && \multicolumn{4}{c}{\vspace{2mm} Relative error for $y'_k$} && \multicolumn{4}{c}{\vspace{2mm} Error for the contribution of $y'_k$} \\
    \cline{3-6} \cline{8-11}
    && $k=0$ & $k=1$ & $k=2$ & $k=3$ && $k=0$ & $k=1$ & $k=2$ & $k=3$ \\
    \midrule
    0.03\% && 1.71\% & 0.05\% & 16.52\% & 89.86\% && 2.23\% & 0.01\% & 0.76\% & 1.46\% \\
    0.10\% && 1.71\% & 0.37\% & 21.25\% & 76.94\% && 2.23\% & 0.10\% & 0.98\% & 1.25\% \\
    3.41\% && 1.67\% & 0.34\% & 57.79\% & 82.79\% && 2.19\% & 0.10\% & 2.66\% & 1.35\% \\
    5.93\% && 2.94\% & 0.07\% & 17.05\% & 78.73\% && 3.85\% & 0.02\% & 0.78\% & 1.28\% \\
    2.14\% && 0.20\% & 0.05\% & 20.57\% & 90.48\% && 0.27\% & 0.01\% & 0.95\% & 1.47\% \\
    \midrule
    2.32\% && 1.65\% & 0.18\% & 26.64\% & 83.76\% && 2.16\% & 0.05\% & 1.23\% & 1.36\% \\
    (0.0247) && (0.0097) & (0.0016) & (0.1754) & (0.0622) && (0.0127) & (0.0004) & (0.0080) & (0.0010)\\
    \bottomrule
    \end{tabularx}
    \caption{Same as Table~\ref{tab:experiment-results-taylor}. Here, the coefficients $b_k(\eta)$ are those for the best-fit polynomial, see Sec.~\ref{sec:experiments}. Estimation is based on the same shots as in the previous Table, only the post-processing is modified to account for the modified coefficients. Underlying values are plotted in Fig.~\ref{fig:results}, blue area.}
    \label{tab:experiment-results-bestfit}
\end{table*}

\begin{table*}
    \centering
    \small
    \newcolumntype{Y}{>{\centering\arraybackslash}X}
    \begin{tabularx}{\textwidth}{YYYYYYYYY}
    \toprule
    $k$ & Grover power & Width & Depth & RZ count & SX count & X \, count & CNOT count & Meas count\\
    \midrule
    $0$	&$2^0$	&$2$	&$11$	&$9$	&$5$	&$1$	&$2$	&$2$\\
    $0$	&$2^0$	&$2$	&$11$	&$9$	&$5$	&$1$	&$2$	&$2$\\
    $0$	&$2^2$	&$2$	&$118$	&$73$	&$37$	&$6$	&$52$	&$2$\\
    $1$	&$2^0$	&$2$	&$12$	&$8$	&$6$	&$1$	&$2$	&$2$\\
    $1$	&$2^0$	&$2$	&$12$	&$8$	&$6$	&$1$	&$2$	&$2$\\
    $1$	&$2^2$	&$2$	&$14$	&$9$	&$8$	&$0$	&$2$	&$2$\\
    $2$	&$2^0$	&$4$	&$22$	&$16$	&$18$	&$2$	&$12$	&$4$\\
    $2$	&$2^1$	&$4$	&$167$	&$101$	&$73$	&$7$	&$112$	&$4$\\
    $2$	&$2^4$	&$4$	&$596$	&$364$	&$243$	&$27$	&$416$	&$4$\\
    $3$	&$2^0$	&$6$	&$39$	&$30$	&$29$	&$3$	&$23$	&$6$\\
    $3$	&$2^1$	&$6$	&$652$	&$365$	&$162$	&$23$	&$470$	&$6$\\
    \bottomrule
    \end{tabularx}
    \caption{Collection of the IQAE iterations needed for the estimation underlying Tables~\ref{tab:experiment-results-taylor} and~\ref{tab:experiment-results-bestfit}. Each row represents an iteration, and contains the metrics of the circuit transpiled for IBM Jakarta. For high $k$, some circuits have depths that lie beyond the possibilities of current hardware, even under error mitigation.}
    \label{tab:experiment-results-iterations}
\end{table*}

\begin{figure*}[t]
    \centering \small
    \include{Figures/results}
    \caption{
    Empirical outcomes of the algorithm, as described in Sec.~\ref{sec:experiments}. Errors for these experiments are provided in Tables~\ref{tab:experiment-results-taylor} and~\ref{tab:experiment-results-bestfit}. Marks represent five independent experiments with $N=4$ and $K=3$, compared to the exact solution in bars. Each bar represents the contribution of a power $k=0, ..., 3$ to the overall result. In the gray case, the coefficients $b_k(\eta)$ are those for the Taylor expansion centered in 0. In the blue case, same experiments as before, but coefficients this time are for the best-fit polynomial. 
    }
    \label{fig:results}
\end{figure*}

We focus our experiments on an instance as small as $N=4$. We take realistic temperatures simulated for a weather station, and the volume function in Eq.~\eqref{eq:volume} with parameters $T_0 = 40$, $A = 20000$, $B = -35$, $C = 3$, $D = 6000$. These constants specify the shape of the sigmoid function and are generated by fitting the curve to historic weather temperatures and gas volumes. Since all generated temperatures are positive, we take $\eta=0$. For energy prices, we use a four-dimensional random vector.

For demonstration purposes, we consider here $K=3$, even though our algorithm is asymptotically advantageous in $N$ only for $K=2$. Note that asymptotically in $N$ implies infinitely in the future temperature and weather forecast data, which clearly does not occur in any real application as described here.
Polynomial coefficients are $[b_k]_k = [17976,\allowbreak -360,\allowbreak -7.17,\allowbreak 0.0072]$ if we consider the Taylor expansion in $0$, or $[b_k]_k = [17957,\allowbreak -393,\allowbreak -6.50,\allowbreak 0.225]$ if we consider the best-fit polynomial. %
In both cases, given $b_k$ are fast decreasing, the only relevant terms for with a $1\%$ error threshold, are $b_0$ and $b_1$, according to Eq.~\eqref{eq:eps-k}. Let us emphasize again that this is due to small $N$, while for increasing $N$ the term $\rho_T^K$ always becomes dominant, independently of the coefficients.

To provide a more insightful example, we run the algorithm with $\epsilon_k = 0.04$ for all $k$, violating the prescription in Eq.~\eqref{eq:eps-k}. We apply IQAE with $100$ shots per iteration. Circuits are executed on the IBM Jakarta device, having 7 qubits, a Quantum Volume~\cite{cross_validating_2019} of 16, and 2.4K CLOPS~\cite{wack_quality_2021}.

Data loading in amplitude encoding is performed through general-purpose initialization methods available in Qiskit \cite{Qiskit}, given the unstructured nature of input data. For circuit optimization and error reduction we use three techniques, namely: `Mapomatic' for finding the best qubit-layout identifying the low noise subgraph \cite{mapomatic}, `dynamical decoupling' (using single $X$-gate configuration) for mitigating decoherence in the ideal qubits \cite{dynamicdecoupling} and `M3 error mitigation' for measurement error mitigation \cite{m3errormitigation}.

Fig.~\ref{fig:results}, as well as Tables~\ref{tab:experiment-results-taylor} and~\ref{tab:experiment-results-bestfit}, show the outputs. As expected, the error is very high for $k=2$ and $3$, but contributes little to the overall result given the low associated coefficients $b_k(\eta)$. The error level for large $k$ is not only a consequence of the choice of $\epsilon_k$, but mostly an effect of noise: indeed, some iterations have a depth as high as 596 or 652, with more than 400 CNOT gates, as shown in Table~\ref{tab:experiment-results-iterations}.

\section{Conclusion}\label{sec:conclusion}
We analyzed a newly proposed approach for the calculation of the inner product $\sum_j f(T'_j) E'_j$ where $[T'_j]_j$ and $[E'_j]_j$ are two input vectors, and $f$ is a function well approximated by a polynomial $p$.
The approach allows to break the workload into multiple parts, where the bottleneck becomes the calculation of the inner products $\sum_j T^k_j E_j$, where $[T_j]_j$ and $[E_j]_j$ are suitably normalized vectors. Hence, we explored the application of quantum computing to accelerate the summations $\sum_j T^k_j E_j$ for all relevant $k$.

In our approach, powers of the input vectors were calculated through QHPs. We introduced a variant of QHP with dynamic stopping thereby leveraging dynamic quantum circuits, achieving an improvement in the leading constants of time scaling, compared to the naive implementation. This variant though is incompatible with QAE and is therefore excluded from the core algorithm.

By applying additional QHPs, the method extends to the calculation of values resulting from multilinear functions $\sum_j X_{j,1} \cdots X_{j,v}$, where $X_{j,i}$ are the inputs, thus generalizing the bilinearity in our methodology. It also adapts to the case where $X_{j,i}$ result from the elementwise application of functions $f_i$, as long as the $f_i$ can be polynomially approximated by $p_i$ for all $i$. Said extensions, though, do not preserve the asymptotic analysis expressed in the manuscript.

To encode the input vectors, we adopted the amplitude encoding technique, after having evaluated the novel Bidirectional Orthogonal Encoding (BOE). The latter was proposed here to overcome the limitations of the bidirectional encoding for our purposes, and specifically for the calculation of inner products through the swap test.

We compared the ancilla-free method against the swap test approach for the calculation of inner products, discussing their asymptotic width and time complexity as functions of the data set size and the error, and highlighting that the ancilla-free method is preferable.

We merged these building blocks and adapted them to apply *QAE, thus getting an asymptotic improvement against the classical performance under the assumption of efficient data loading, for degree $K = 2$ of the approximating polynomial. The optimal implementation was selected after evaluating multiple variants, collected and discussed in deeper detail in accompanying Appendices.

We provided a rigorous analysis and evaluation of four variants of the algorithm, differing with respect to data loading, inner product computations, and sequential algorithmic steps.
Whilst most of prior literature studies the complexity of quantum subroutines, our work focused on the performance of an overall quantum workload. The extensive and rigorous analysis of the asymptotic complexity of each variant showed how subroutines of an application quantum workload interact and condition the overall performance. We highlighted that subtle details of the encodings can make dramatic effect in terms of performance or applicability of methods.
A potential area of future investigation is the scope of improvement if one adopts the Szegedy walk \cite{szegedy_quantum_2004} to create the input.
Other areas of prospect work may include the validation and extension of the hypotheses under which the asymptotic analysis was conducted.

Experiments on real quantum hardware were conducted for the main variant, and showed errors in line with theoretical expectations for small problem instances. The effect of noise becomes relevant for high power orders of polynomial expansion, and it is known from the theory that error gets amplified when dealing with lengthier input vectors. Extensive experiments on quantum simulators are included in the Appendices and validate the overall theoretical framework across the different variants. It is currently impossible to run experiments in the asymptotically advantageous regime, due to the depth of the circuits. The estimation of the break-even point and of the resources needed to achieve the asymptotic regime is a substantial work, left as a future development (see Ref.~\cite{campbell_applying_2019} as an example of resource estimation).

\section*{Acknowledgements}
G.A., K.Y., and O.S. acknowledge Travis L. Scholten, Raja Hebbar, and Morgan Delk for helping with the business case analysis; Francois Varchon, Winona Murphy, and Matthew Stypulkoski from the IBM Quantum Support team to help executing the experiments; Kristan Temme, Daniel Egger, and Stefan Woerner for their feedbacks on the manuscript; Jay Gambetta, Thomas Alexander and Sarah Sheldon for allocating compute time on advanced hardware; Maria Cristina Ferri, Jeannette M. Garcia, Gianmarco Quarti Trevano, Katie Pizzolato, Jae-Eun Park, Heather Higgins, and Saif Rayyan for their support in cross-team collaborations.

\bibliographystyle{quantum}
\bibliography{eon_bibliography} 

\begin{thebibliography}{10}

\bibitem{ghosh_energy_2024}
Kumar Ghosh, Kavitha Yogaraj, Gabriele Agliardi, Piergiacomo Sabino, Marina Fernández-Campoamor, Juan Bernabé-Moreno, Giorgio Cortiana, Omar Shehab, and Corey O'Meara.
\newblock ``Energy risk analysis with dynamic amplitude estimation and piecewise approximate quantum compiling''.
\newblock \href{https://dx.doi.org/10.1109/TQE.2024.3425969}{{IEEE} Transactions on Quantum EngineeringPages 1--17}~(2024).

\bibitem{holmes_nonlinear_2021}
Zoë Holmes, Nolan Coble, Andrew~T. Sornborger, and Yiğit Subaşı.
\newblock ``On nonlinear transformations in quantum computation''.
\newblock {arXiv}:2112.12307 [quant-ph]~(2021).
\newblock  \href{http://arxiv.org/abs/2112.12307}{arXiv:2112.12307}.

\bibitem{stamatopoulos_option_2020}
Nikitas Stamatopoulos, Daniel~J. Egger, Yue Sun, Christa Zoufal, Raban Iten, Ning Shen, and Stefan Woerner.
\newblock ``Option {Pricing} using {Quantum} {Computers}''.
\newblock \href{https://dx.doi.org/10.22331/q-2020-07-06-291}{Quantum {\bf 4}, 291}~(2020).

\bibitem{woerner_quantum_2019}
Stefan Woerner and Daniel~J. Egger.
\newblock ``Quantum {Risk} {Analysis}''.
\newblock \href{https://dx.doi.org/10.1038/s41534-019-0130-6}{npj Quantum Information {\bf 5}, 15}~(2019).

\bibitem{herman_survey_2022}
Dylan Herman, Cody Googin, Xiaoyuan Liu, Alexey Galda, Ilya Safro, Yue Sun, Marco Pistoia, and Yuri Alexeev.
\newblock ``A {Survey} of {Quantum} {Computing} for {Finance}''~(2022).
\newblock  \href{http://arxiv.org/abs/2201.02773}{arXiv:2201.02773}.

\bibitem{brassard_quantum_2002}
Gilles Brassard, Peter Høyer, Michele Mosca, and Alain Tapp.
\newblock ``Quantum amplitude amplification and estimation''.
\newblock In Samuel~J. Lomonaco and Howard~E. Brandt, editors, Contemporary Mathematics.
\newblock \href{https://dx.doi.org/10.1090/conm/305/05215}{Volume 305, pages 53--74}.
\newblock American Mathematical Society~(2002).

\bibitem{rall_amplitude_2022}
Patrick Rall and Bryce Fuller.
\newblock ``Amplitude {Estimation} from {Quantum} {Signal} {Processing}''~(2022) \href{http://arxiv.org/abs/2207.08628}{arXiv:2207.08628}.
\newblock arXiv:2207.08628 [quant-ph].

\bibitem{bouland_prospects_2020}
Adam Bouland, Wim van Dam, Hamed Joorati, Iordanis Kerenidis, and Anupam Prakash.
\newblock ``Prospects and challenges of quantum finance''~(2020) \href{http://arxiv.org/abs/2011.06492 [quant-ph, q-fin]}{arXiv:2011.06492 [quant-ph, q-fin]}.

\bibitem{berger_quantum_2021}
Casey Berger, Agustin Di~Paolo, Tracey Forrest, Stuart Hadfield, Nicolas Sawaya, Michał Stęchły, and Karl Thibault.
\newblock ``Quantum technologies for climate change: Preliminary assessment''~(2021) \href{http://arxiv.org/abs/2107.05362 [quant-ph]}{arXiv:2107.05362 [quant-ph]}.

\bibitem{grinko_iterative_2021}
Dmitry Grinko, Julien Gacon, Christa Zoufal, and Stefan Woerner.
\newblock ``Iterative quantum amplitude estimation''.
\newblock \href{https://dx.doi.org/10.1038/s41534-021-00379-1}{npj Quantum Information {\bf 7}, 52}~(2021).
\newblock  \href{http://arxiv.org/abs/1912.05559}{arXiv:1912.05559}.

\bibitem{lubasch_variational_2020}
Michael Lubasch, Jaewoo Joo, Pierre Moinier, Martin Kiffner, and Dieter Jaksch.
\newblock ``Variational quantum algorithms for nonlinear problems''.
\newblock \href{https://dx.doi.org/10.1103/PhysRevA.101.010301}{Physical Review A {\bf 101}, 010301}~(2020).

\bibitem{araujo_configurable_2022}
Israel~F. Araujo, Daniel~K. Park, Teresa~B. Ludermir, Wilson~R. Oliveira, Francesco Petruccione, and Adenilton~J. da~Silva.
\newblock ``Configurable sublinear circuits for quantum state preparation''.
\newblock {arXiv}:2108.10182 [quant-ph]~(2022).
\newblock  \href{http://arxiv.org/abs/2108.10182}{arXiv:2108.10182}.

\bibitem{shende_synthesis_2006}
V.~V. Shende, S.~S. Bullock, and I.~L. Markov.
\newblock ``Synthesis of quantum-logic circuits''.
\newblock \href{https://dx.doi.org/10.1109/TCAD.2005.855930}{{IEEE} Transactions on Computer-Aided Design of Integrated Circuits and Systems {\bf 25}, 1000--1010}~(2006).

\bibitem{rieffel_quantum_2011}
E.G. Rieffel and W.H. Polak.
\newblock ``Quantum computing: A gentle introduction''.
\newblock Scientific and Engineering Computation. {MIT} Press. ~(2011).
\newblock  url:~\href{https://books.google.it/books?id=9Cs3AgAAQBAJ}{books.google.it/books?id=9Cs3AgAAQBAJ}.

\bibitem{Nielsen}
M.~A. Nielsen and Chuang~I. L.
\newblock ``Quantum computation and quantum information''.
\newblock \href{https://dx.doi.org/10.1017/CBO9780511976667}{Cambridge University Press}. Cambridge~(2001).

\bibitem{rebentrost_quantum_2018}
Patrick Rebentrost, Brajesh Gupt, and Thomas~R. Bromley.
\newblock ``Quantum computational finance: Monte carlo pricing of financial derivatives''.
\newblock \href{https://dx.doi.org/10.1103/PhysRevA.98.022321}{Physical Review A {\bf 98}, 022321}~(2018).
\newblock  \href{http://arxiv.org/abs/1805.00109}{arXiv:1805.00109}.

\bibitem{terno_nonlinear_1999}
Daniel~R. Terno.
\newblock ``Nonlinear operations in quantum-information theory''.
\newblock \href{https://dx.doi.org/10.1103/PhysRevA.59.3320}{Physical Review A {\bf 59}, 3320--3324}~(1999).

\bibitem{horodecki_limits_2003}
Paweł Horodecki.
\newblock ``From limits of quantum operations to multicopy entanglement witnesses and state-spectrum estimation''.
\newblock \href{https://dx.doi.org/10.1103/PhysRevA.68.052101}{Physical Review A {\bf 68}, 052101}~(2003).

\bibitem{schuld_quest_2014}
Maria Schuld, Ilya Sinayskiy, and Francesco Petruccione.
\newblock ``The quest for a quantum neural network''.
\newblock \href{https://dx.doi.org/10.1007/s11128-014-0809-8}{Quantum Information Processing {\bf 13}, 2567--2586}~(2014).

\bibitem{cong_quantum_2019}
Iris Cong, Soonwon Choi, and Mikhail~D. Lukin.
\newblock ``Quantum convolutional neural networks''.
\newblock \href{https://dx.doi.org/10.1038/s41567-019-0648-8}{Nature Physics {\bf 15}, 1273--1278}~(2019).

\bibitem{beer_training_2020}
Kerstin Beer, Dmytro Bondarenko, Terry Farrelly, Tobias~J. Osborne, Robert Salzmann, Daniel Scheiermann, and Ramona Wolf.
\newblock ``Training deep quantum neural networks''.
\newblock \href{https://dx.doi.org/10.1038/s41467-020-14454-2}{Nature Communications {\bf 11}, 808}~(2020).

\bibitem{leyton_quantum_2008}
Sarah~K. Leyton and Tobias~J. Osborne.
\newblock ``A quantum algorithm to solve nonlinear differential equations''~(2008) \href{http://arxiv.org/abs/0812.4423 [quant-ph]}{arXiv:0812.4423 [quant-ph]}.

\bibitem{10.5555/2011586.2011592}
Todd~A. Bruni.
\newblock ``Measurimg polynomial functions of states''.
\newblock Quantum Info. Comput. {\bf 4}, 401–408~(2004).

\bibitem{maronese_quantum_2022}
Marco Maronese, Claudio Destri, and Enrico Prati.
\newblock ``Quantum activation functions for quantum neural networks''.
\newblock \href{https://dx.doi.org/10.1007/s11128-022-03466-0}{Quantum Information Processing {\bf 21}, 128}~(2022).

\bibitem{guo_nonlinear_2021}
Naixu Guo, Kosuke Mitarai, and Keisuke Fujii.
\newblock ``Nonlinear transformation of complex amplitudes via quantum singular value transformation''~(2021) \href{http://arxiv.org/abs/2107.10764 [quant-ph]}{arXiv:2107.10764 [quant-ph]}.

\bibitem{rattew_non-linear_2023}
Arthur~G. Rattew and Patrick Rebentrost.
\newblock ``Non-linear transformations of quantum amplitudes: Exponential improvement, generalization, and applications''~(2023) \href{http://arxiv.org/abs/2309.09839 [quant-ph]}{arXiv:2309.09839 [quant-ph]}.

\bibitem{egger_credit_2021}
Daniel~J. Egger, Ricardo Garcia~Gutierrez, Jordi~Cahue Mestre, and Stefan Woerner.
\newblock ``Credit risk analysis using quantum computers''.
\newblock \href{https://dx.doi.org/10.1109/TC.2020.3038063}{{IEEE} Transactions on Computers {\bf 70}, 2136--2145}~(2021).

\bibitem{BenthBenth05}
F.E. Benth and J.~Saltyte-Benth.
\newblock ``Stochastic modelling of temperature variations with a view towards weather derivatives''.
\newblock \href{https://dx.doi.org/10.1080/1350486042000271638}{Applied Mathematical Finance {\bf 12}, 53--85}~(2005).

\bibitem{BenthBenth07}
F.E. Benth and J.~Saltyte-Benth.
\newblock ``The volatility of temperature and pricing of weather derivatives''.
\newblock \href{https://dx.doi.org/10.1080/14697680601155334}{Quantitative Finance {\bf 7}, 553--561}~(2007).

\bibitem{cucu_et_al16}
L.~Cucu, R.~D\"ottling, P.~Heider, and S.~Maina.
\newblock ``Managing temperature-driven volume risks''.
\newblock \href{https://dx.doi.org/10.21314/JEM.2016.145}{Journal of energy markets {\bf 9}, 95--110}~(2016).

\bibitem{cs20_2}
P.~Sabino and N.~Cufaro Petroni.
\newblock ``Fast {P}ricing of {E}nergy {D}erivatives with {M}ean-{R}everting {J}ump-diffusion {P}rocesses''.
\newblock \href{https://dx.doi.org/10.1080/1350486X.2021.1909488}{Applied Mathematical Finance {\bf 0}, 1--22}~(2021).

\bibitem{weigold_encoding_2021}
Manuela Weigold, Johanna Barzen, Frank Leymann, and Marie Salm.
\newblock ``Encoding patterns for quantum algorithms''.
\newblock \href{https://dx.doi.org/10.1049/qtc2.12032}{IET Quantum Communication {\bf 2}, 141--152}~(2021).

\bibitem{barenco1995elementary}
Adriano Barenco, Charles~H Bennett, Richard Cleve, David~P DiVincenzo, Norman Margolus, Peter Shor, Tycho Sleator, John~A Smolin, and Harald Weinfurter.
\newblock ``Elementary gates for quantum computation''.
\newblock \href{https://dx.doi.org/10.1103/PhysRevA.52.3457}{Physical review A {\bf 52}, 3457}~(1995).

\bibitem{Kumar}
P.~Kumar.
\newblock ``Direct implementation of an n-qubit controlled-unitary gate in a single step''.
\newblock \href{https://dx.doi.org/10.1007/s11128-012-0465-9}{Quantum information processing {\bf 12}, 1201--1223}~(2013).

\bibitem{Cortes}
J.~A. Cortese and T.~M. Braje.
\newblock ``Loading classical data into a quantum computer''~(2018).
\newblock  \href{http://arxiv.org/abs/1803.01958}{arXiv:1803.01958}.

\bibitem{Plesch}
M.~Plesch and {\v{C}}aslav Brukner.
\newblock ``Quantum-state preparation with universal gate decompositions''.
\newblock \href{https://dx.doi.org/10.1103/PhysRevA.83.032302}{Phys. Rev. A {\bf 83}, 032302}~(2011).

\bibitem{Miszczak}
J.~A. Miszczak.
\newblock ``Singular value decomposition and matrix reorderings in quantum information theory''.
\newblock \href{https://dx.doi.org/10.1142/S0129183111016683}{International Journal of Modern Physics C {\bf 22}, 897--918}~(2011).

\bibitem{Heinosaari}
T.~Heinosaari and M.~Ziman.
\newblock ``Guide to mathematical concepts of quantum theory''.
\newblock \href{https://dx.doi.org/10.2478/v10155-010-0091-y}{AcPSl {\bf 58}, 487--674}~(2008).

\bibitem{Bengtsson}
I.~Bengtsson and K.~Zyczkowski.
\newblock ``Geometry of quantum states''~(2006).

\bibitem{zoufal2019quantum}
Christa Zoufal, Aur{\'e}lien Lucchi, and Stefan Woerner.
\newblock ``Quantum generative adversarial networks for learning and loading random distributions''.
\newblock \href{https://dx.doi.org/10.1038/s41534-019-0223-2}{npj Quantum Information {\bf 5}, 1--9}~(2019).

\bibitem{agliardi_optimal_2022}
Gabriele Agliardi and Enrico Prati.
\newblock ``Optimal {Tuning} of {Quantum} {Generative} {Adversarial} {Networks} for {Multivariate} {Distribution} {Loading}''.
\newblock \href{https://dx.doi.org/10.3390/quantum4010006}{Quantum Reports {\bf 4}, 75--105}~(2022).

\bibitem{grover_creating_2002}
Lov Grover and Terry Rudolph.
\newblock ``Creating superpositions that correspond to efficiently integrable probability distributions''~(2002).
\newblock  \href{http://arxiv.org/abs/quant-ph/0208112}{arXiv:quant-ph/0208112}.

\bibitem{giovannetti_architectures_2008}
Vittorio Giovannetti, Seth Lloyd, and Lorenzo Maccone.
\newblock ``Architectures for a quantum random access memory''.
\newblock \href{https://dx.doi.org/10.1103/PhysRevA.78.052310}{Physical Review A {\bf 78}, 052310}~(2008).

\bibitem{araujo2021divide}
Israel~F Araujo, Daniel~K Park, Francesco Petruccione, and Adenilton~J da~Silva.
\newblock ``A divide-and-conquer algorithm for quantum state preparation''.
\newblock \href{https://dx.doi.org/10.1038/s41598-021-85474-1}{Scientific Reports {\bf 11}, 1--12}~(2021).

\bibitem{pressrelease_dynamic_2021}
``Quantum circuits get a dynamic upgrade with the help of concurrent classical computation''.
\newblock \textsc{url:}~\url{https://www.ibm.com/blogs/research/2021/02/quantum-phase-estimation/}~(2021).

\bibitem{dynamic_circuits}
A.~D. C\'orcoles, Maika Takita, Ken Inoue, Scott Lekuch, Zlatko~K. Minev, Jerry~M. Chow, and Jay~M. Gambetta.
\newblock ``Exploiting dynamic quantum circuits in a quantum algorithm with superconducting qubits''.
\newblock \href{https://dx.doi.org/10.1103/PhysRevLett.127.100501}{Phys. Rev. Lett. {\bf 127}, 100501}~(2021).

\bibitem{manzano_real_2023}
Alberto Manzano, Daniele Musso, and Álvaro Leitao.
\newblock ``Real quantum amplitude estimation''.
\newblock \href{https://dx.doi.org/10.1140/epjqt/s40507-023-00159-0}{{EPJ} Quantum Technology {\bf 10}, 2}~(2023).

\bibitem{montanaro_quantum_2015}
Ashley Montanaro.
\newblock ``Quantum speedup of monte carlo methods''.
\newblock \href{https://dx.doi.org/10.1098/rspa.2015.0301}{Proceedings of the Royal Society A: Mathematical, Physical and Engineering Sciences {\bf 471}, 20150301}~(2015).

\bibitem{cross_validating_2019}
Andrew~W. Cross, Lev~S. Bishop, Sarah Sheldon, Paul~D. Nation, and Jay~M. Gambetta.
\newblock ``Validating quantum computers using randomized model circuits''.
\newblock \href{https://dx.doi.org/10.1103/PhysRevA.100.032328}{Physical Review A {\bf 100}, 032328}~(2019).

\bibitem{wack_quality_2021}
Andrew Wack, Hanhee Paik, Ali Javadi-Abhari, Petar Jurcevic, Ismael Faro, Jay~M. Gambetta, and Blake~R. Johnson.
\newblock ``Quality, speed, and scale: three key attributes to measure the performance of near-term quantum computers''~(2021) \href{http://arxiv.org/abs/2110.14108 [quant-ph]}{arXiv:2110.14108 [quant-ph]}.

\bibitem{Qiskit}
{Qiskit contributors}.
\newblock ``Qiskit: An open-source framework for quantum computing''~(2023).

\bibitem{mapomatic}
Matthew Treinish et~al.
\newblock ``mapomatic: Automatic mapping of compiled circuits to low-noise sub-graphs''.
\newblock \textsc{url:}~\url{https://github.com/Qiskit-Partners/mapomatic}~(2022).

\bibitem{dynamicdecoupling}
{Qiskit software developers}.
\newblock ``Dynamical decoupling insertion pass''.
\newblock \textsc{url:}~\url{https://qiskit.org/documentation/stubs/qiskit.transpiler.passes.DynamicalDecoupling.html}~(2022).

\bibitem{m3errormitigation}
Paul~D. Nation, Hwajung Kang, Neereja Sundaresan, and Jay~M. Gambetta.
\newblock ``Scalable mitigation of measurement errors on quantum computers''.
\newblock \href{https://dx.doi.org/10.1103/PRXQuantum.2.040326}{PRX Quantum {\bf 2}, 040326}~(2021).

\bibitem{szegedy_quantum_2004}
M.~Szegedy.
\newblock ``Quantum {Speed}-{Up} of {Markov} {Chain} {Based} {Algorithms}''.
\newblock In 45th {Annual} {IEEE} {Symposium} on {Foundations} of {Computer} {Science}.
\newblock \href{https://dx.doi.org/10.1109/FOCS.2004.53}{Pages 32--41}.
\newblock Rome, Italy~(2004). IEEE.

\bibitem{campbell_applying_2019}
Earl Campbell, Ankur Khurana, and Ashley Montanaro.
\newblock ``Applying quantum algorithms to constraint satisfaction problems''.
\newblock \href{https://dx.doi.org/10.22331/q-2019-07-18-167}{Quantum {\bf 3}, 167}~(2019).
\newblock  \href{http://arxiv.org/abs/1810.05582 [quant-ph]}{arXiv:1810.05582 [quant-ph]}.

\bibitem{fanizza_beyond_2020}
M.~Fanizza, M.~Rosati, M.~Skotiniotis, J.~Calsamiglia, and V.~Giovannetti.
\newblock ``Beyond the swap test: Optimal estimation of quantum state overlap''.
\newblock \href{https://dx.doi.org/10.1103/PhysRevLett.124.060503}{Physical Review Letters {\bf 124}, 060503}~(2020).

\bibitem{markov_generalized_2022}
Vanio Markov, Charlee Stefanski, Abhijit Rao, and Constantin Gonciulea.
\newblock ``A generalized quantum inner product and applications to financial engineering''~(2022) \href{http://arxiv.org/abs/2201.09845}{arXiv:2201.09845}.

\bibitem{buhrman_quantum_2001}
Harry Buhrman, Richard Cleve, John Watrous, and Ronald de~Wolf.
\newblock ``Quantum fingerprinting''.
\newblock \href{https://dx.doi.org/10.1103/PhysRevLett.87.167902}{Physical Review Letters {\bf 87}, 167902}~(2001).

\bibitem{schuld_introduction_2015}
Maria Schuld, Ilya Sinayskiy, and Francesco Petruccione.
\newblock ``An introduction to quantum machine learning''.
\newblock \href{https://dx.doi.org/10.1080/00107514.2014.964942}{Contemporary Physics {\bf 56}, 172--185}~(2015).

\bibitem{havlicek_supervised_2019}
Vojtěch Havlíček, Antonio~D. Córcoles, Kristan Temme, Aram~W. Harrow, Abhinav Kandala, Jerry~M. Chow, and Jay~M. Gambetta.
\newblock ``Supervised learning with quantum-enhanced feature spaces''.
\newblock \href{https://dx.doi.org/10.1038/s41586-019-0980-2}{Nature {\bf 567}, 209--212}~(2019).

\bibitem{nakaji_faster_2020}
Kouhei Nakaji.
\newblock ``Faster amplitude estimation''.
\newblock \href{https://dx.doi.org/10.26421/QIC20.13-14-2}{Quantum Information and Computation {\bf 20}, 1109--1123}~(2020).
\newblock  \href{http://arxiv.org/abs/2003.02417}{arXiv:2003.02417}.

\bibitem{maslov_advantages_2016}
Dmitri Maslov.
\newblock ``Advantages of using relative-phase {Toffoli} gates with an application to multiple control {Toffoli} optimization''.
\newblock \href{https://dx.doi.org/10.1103/PhysRevA.93.022311}{Physical Review A {\bf 93}, 022311}~(2016).

\bibitem{tang_quantum-inspired_2019}
Ewin Tang.
\newblock ``A quantum-inspired classical algorithm for recommendation systems''.
\newblock In Proceedings of the 51st Annual {ACM} {SIGACT} Symposium on Theory of Computing.
\newblock \href{https://dx.doi.org/10.1145/3313276.3316310}{Pages 217--228}.
\newblock ~(2019).
\newblock  \href{http://arxiv.org/abs/1807.04271}{arXiv:1807.04271}.

\end{thebibliography}

\onecolumn

\appendix

\section{Alternate ways to compute the inner product in amplitude encoding}\label{appendix:inner-prod}
Consider two vectors $\ket{\psi_0}$ and $\ket{\psi_1}$, and suppose one wants to calculate their inner product $p$. Multiple quantum techniques for the computation of their inner product are known~\cite{fanizza_beyond_2020}.
In this Appendix we focus on two, namely the swap test and the ancilla-free method. We show in Appendix~\ref{appendix:qae-ampl} how these methods can be further enhanced recurring to Quantum Amplitude Estimation techniques as suggested in Ref.~\cite{markov_generalized_2022}.

\subsection{The swap test and the ancilla-free method: definition and sampling complexity}

The swap-test~\cite{buhrman_quantum_2001} is depicted in Fig.~\ref{fig:swaptest}. Being characterized by low gate depth, it is widely used in near-term applications including quantum machine learning~\cite{schuld_introduction_2015}.

To discuss its convergence, we need a Lemma:
\begin{lemma}\label{lemma:clt-sqrt}
Let $\bar X_S$ be the mean of $S$ i.i.d. random variables with mean $\mu >0$ and variance $\sigma^2$, and let $Y_S := \sqrt{\max\{a X _S + b, 0\}}$, for some real constants $a, b$ with $a \neq 0$ and $a\mu+b > 0$.
Then $\frac{ Y_S -\sqrt{a\mu+b} }{\abs{a} \sigma/\sqrt{4S(a\mu+b)}}$ is asymptotically a standard normal random variable when $S \to \infty$. Therefore, the error is controlled by $\mathbb{P}\left( \abs{ Y_S - \sqrt{a\mu+b}} < \epsilon \right) = \alpha$
once $S$ is chosen as
\begin{equation}\label{eq:clt-sqrt:S-scaling}
    S = \frac{a^2 \sigma^2}{4 (a\mu+b) \epsilon^2} \left[ \Phi^{-1} \!\left(\frac{1+\alpha}{2}\right) \right]^2
\end{equation}
asymptotically when $\epsilon \to 0$, where $\Phi$ is the CDF of the standard normal distribution.
\end{lemma}
\begin{proof}
By the Central Limit Theorem, for any real $\beta$,
$$\mathbb{P} \left( \frac{\bar X_S - \mu}{\sigma/ \sqrt{S}} < \beta \right) \to \Phi^{-1}(\beta)$$
when $S \to \infty$. Therefore, using the symmetry of $\Phi$ if $a<0$,
$$\mathbb{P} \left( \frac{ a \bar X_S + a \mu}{\abs{a} \sigma/\sqrt{S}} < \beta \right) \to \Phi^{-1}(\beta)$$
and 
$$\mathbb{P} \left( \frac{(a \bar X_S + b) - (a \mu +b)}{\abs{a} \sigma/\sqrt{S}} < \beta \right) \to \Phi^{-1}(\beta).$$
Then also
\begin{equation*}
\mathbb{P} \left( \frac{ (a \bar X_S + b) - (a \mu +b)}{\abs{a} \sigma/\sqrt{S}} < \beta + \beta^2 \frac{\abs{a} \sigma/\sqrt{S}}{4(a\mu+b)} \right) \to \Phi^{-1}(\beta)
\end{equation*}
by continuity of $\Phi^{-1}$, since the additional term is defined for $a\mu+b > 0$ and tends to 0. Rearranging the inequality:
\begin{equation*}
\mathbb{P} \left( \frac{a \bar X_S + b}{\abs{a} \sigma/\sqrt{S}} < \frac{a \mu +b}{\abs{a} \sigma/\sqrt{S}} + \beta + \beta^2 \frac{\abs{a} \sigma/\sqrt{S}}{4(a\mu+b)} \right) \to \Phi^{-1}(\beta)
\end{equation*}
namely
\begin{equation*}
\mathbb{P} \left( \frac{a \bar X_S + b}{\abs{a}\sigma/\sqrt{S}} < c^2 \right) \to \Phi^{-1}(\beta) \qquad \text{where } c := \sqrt{\frac{a \mu + b}{\abs{a} \sigma/\sqrt{S}}} + \beta \sqrt{  \frac{\abs{a} \sigma/\sqrt{S}}{4 (a\mu+b)}}.
\end{equation*}
At the same time, the probability can be decomposed as
\begin{equation*}
\mathbb{P} \left( \frac{a \bar X_S + b}{\abs{a}\sigma/\sqrt{S}} < c^2 \mid \bar X_S \geq 0 \right) \mathbb{P} \left( \bar X_S \geq 0 \right) + \mathbb{P} \left( \frac{a \bar X_S + b}{\abs{a}\sigma/\sqrt{S}} < c^2 \mid \bar X_S < 0 \right) \mathbb{P} \left( \bar X_S < 0 \right)
\end{equation*}
where the second term tends to $0$ for the strong law of large numbers, since $\mu > 0$. On the first term, $a \bar X_S + b$ equals $Y^2$, so that we obtained
\begin{equation*}
\mathbb{P} \left( \frac{Y_S^2}{\abs{a}\sigma/\sqrt{S}} < c^2 \right) \to \Phi^{-1}(\beta)
\end{equation*}
Taking the square roots:
\begin{equation*}\mathbb{P} \left( \frac{Y_S}{\sqrt{\abs{a} \sigma/\sqrt{S}}} < \sqrt{\frac{a\mu+b}{\abs{a} \sigma/\sqrt{S}}} + \beta \sqrt{  \frac{\abs{a} \sigma/\sqrt{S}}{4 (a\mu+b)}} \right) \to \Phi^{-1}(\beta)
\end{equation*}
which rewrites
\begin{equation*}
\mathbb{P} \left( \frac{ Y_S -\sqrt{a\mu+b} }{\abs{a} \sigma/\sqrt{4S(a\mu+b)}} < \beta \right) \to \Phi^{-1}(\beta).
\end{equation*}
This proves the asymptotic standard normality. 
Therefore $\mathbb{P}\left( \abs{ Y_S - \sqrt{a\mu+b}} < \epsilon \right) = \alpha$ is equivalent to
\begin{equation*}
\frac{\epsilon}{\abs{a} \sigma/\sqrt{4S(a\mu+b)}} = \Phi^{-1}\!\left(\frac{1+\alpha}{2}\right)
\end{equation*}
asymptotically, where the last equation turns to be~\eqref{eq:clt-sqrt:S-scaling}.
\end{proof}

\begin{figure}
    \centering
    \tikzsetnextfilename{swaptest}
    \begin{quantikz}
    \lstick{$\ket{\psi_0}$} & \qwbundle{n} & \qw{} \gategroup[wires=3,steps=3,style={dashed, rounded corners, fill=blue!10}, background]{} & \swap{1} & \qw{} & \qw{} \\
    \lstick{$\ket{\psi_1}$} & \qwbundle{n} & \qw{} & \targX{} & \qw{} & \qw{} \\
    \lstick{$\ket{0}$} & \qw{} & \gate{H} & \ctrl{-2} & \gate{H} & \meter{}
    \end{quantikz}
    \caption{The swap test. The probability to measure $\ket{0}$ in the result qubit is $ \frac{1}{2}+\frac{1}{2} \left| \braket{\psi_0}{\psi_1} \right|^2$.}
    \label{fig:swaptest}
\end{figure}

\begin{prop}[Sampling complexity of the swap test in amplitude encoding]\label{prop:errorbound-swap}
Let $\alpha \in (0,1)$. Let $X_i$, for $i=1,...,S$, be a r.v. representing the output of the swap-test measurement after the $i$-th shot of circuit in Fig.~\ref{fig:swaptest}. Call $\bar X_S$ the mean r.v. resulting from the $S$ independent shots. Then $Y_S := \sqrt{\max\{0,1-2 \bar X_S\}}$ is an estimator for $p=\braket{\psi_0}{\psi_1}$. Assuming $p > 0$, the error is controlled by
\begin{equation}\label{eq:errorbound-clt-swap}
    \mathbb{P}\left( \abs{ Y_S - p} < \epsilon \right) = \alpha,
\end{equation}
once $S$ is chosen as
\begin{equation}\label{eq:errorbound-clt-shots-swap}
    S = \frac{1-p^4}{4\epsilon^2 p^2} \left[ \Phi^{-1} \!\left(\frac{1+\alpha}{2}\right) \right]^2
\end{equation}
asymptotically when $\epsilon \to 0$, where $\Phi$ is the CDF of the standard normal distribution.
\end{prop}

\begin{proof}
The probability to measure $\ket{0}$ in the swap test qubit is $\frac{1}{2}+\frac{1}{2}p^2$. Therefore, $X_i$ are independent Bernoulli with mean
$\mu = \frac{1}{2} -\frac{1}{2} p^2$
and variance
$\sigma^2 = \mu(1-\mu) = \frac{1}{4}-\frac{1}{4} p^4$.
Then Apply Lemma~\ref{lemma:clt-sqrt}.
\end{proof}

\begin{remark}
Since the presence of $p$ in the denominator of Eq.~\eqref{eq:errorbound-clt-shots-swap} may come unexpected, let us shortly comment: it derives from the fact that our estimator is bound to the mean r.v. through a square root. Indeed, $\frac{\bar X_S - \mu}{\sigma/ \sqrt{S}}$ being asymptotically standard normal is equivalent to $\frac{Y_S^2 - p^2}{2 \sigma/ \sqrt{S}}$ being asymptotically standard normal. Now, we can write $Y_S^2 - p^2 = (Y_S - p)(Y_S + p)$ and informally observe that $Y_S + p \to 2p$. So, informally, $\frac{2p (Y_S - p)}{2 \sigma/ \sqrt{S}} = \frac{Y_S - p}{\sigma/ (p\sqrt{S})}$ is asymptotically standard normal, and $p$ appears at the denominator in the estimator variance.
\end{remark}

\begin{remark}
The previous proposition gives a sufficient condition when $p \neq 0$; we also have the following sufficient condition when $p=0$ (and more in general for $\abs{p} < \epsilon$):
$$
    S \geq \frac{1 - p^4}{\left( \epsilon^2 -p^2 \right)^2} \left[ \Phi^{-1} \!\left(\frac{1+\alpha}{2}\right) \right]^2
$$
via CLT applied to $Y_S^2 - p^2$.
\end{remark}

We call \textit{ancilla-free method} another, even simpler way~\cite{havlicek_supervised_2019, markov_generalized_2022} to calculate the inner product of two statevectors $\ket{\psi_0}$ and $\ket{\psi_1}$. Its application is possible once an (efficient) unitary for loading at least one of the two states, say $\ket{\psi_1}$, is known. Namely, an operator $U_1$ is given such that $U_1 \ket{0} = \ket{\psi_1}$. Indeed:
\begin{equation}\label{eq:U1_inverse}
    \matrixel{0}{U_1^\dag}{\psi_0} = \innerproduct{\psi_1}{\psi_0} = p,
\end{equation}
so that it is sufficient to build the $U_1^\dag\ket{\psi_0}$ circuit, and project on $\ket{0}$.

\begin{prop}[Sampling complexity of the ancilla-free method] \label{prop:errorbound-direct}
Let $\alpha \in (0,1)$. Suppose the ancilla-free method for the calculation of the inner product is implemented, the register is measured, and the execution is repeated $S$ times. Let $R_i \in \{0, ..., N-1\}$ be the measurement output for the $i$-th shot, for $i=1,...,S$, and let $X_i$ be a r.v. valued 1 if $R_i=0$, and valued 0 otherwise. Call $\bar X_S$ the mean r.v. resulting from the $S$ independent shots. Then $Y_S := \sqrt{\bar X_S}$ is an estimator for $p$, and the error is controlled by
\begin{equation}\label{eq:errorbound-clt-direct}
    \mathbb{P}\left( \abs{ Y_S - p} < \epsilon \right) = \alpha,
\end{equation}
once $S$ is chosen as
\begin{equation}\label{eq:errorbound-clt-shots-direct}
    S = \frac{1-p^2}{4 \epsilon^2}  \left[ \Phi^{-1} \!\left(\frac{1+\alpha}{2}\right) \right]^2
\end{equation}
asymptotically when $\epsilon \to 0$, where $\Phi$ is the CDF of the standard normal distribution.
\end{prop}
\begin{proof}
By Eq.~\eqref{eq:U1_inverse}, $X_i$ are independent Bernoulli with mean $\mu = p^2$ and variance $\sigma^2 = \mu(1-\mu) = p^2(1-p^2)$. The proof is an application of Lemma~\ref{lemma:clt-sqrt}.
\end{proof}

\begin{remark}
Lemma~\ref{lemma:clt-sqrt} and therefore the proof of Prop.~\ref{prop:errorbound-swap} leverage the fact that $Y_S$ is definitely positive. Nonetheless we comment in Fig.~\ref{fig:inner-compare} that when $p$ is small ($p=0.072$), even with $S$ as high as $10,000$, we empirically get the left-hand side to be negative with a probability of $38\%$. This implies that the estimator $Y_S$ is remarkably biased for the swap test. Fortunately this is not the case of the ancilla-free method since $Y_S$ is always positive.
\end{remark}

The key differences between the ancilla-free method and the swap test are summarized in Table~\ref{tab:inner-compare}. The sampling complexity is plotted in Fig.~\ref{fig:inner-p} as a function of the inner product $p$. It is easy to verify analytically that the sampling complexity of the swap test is unbounded for small $p$, what makes it impossible to choose the number of shots a priori. Fig.~\ref{fig:inner-compare} empirically shows the different behavior of the two methods when $p$ is large (left plot) rather than small (right plot).

\begin{table*}
    \centering
    \small
    \newcolumntype{Y}{>{\centering\arraybackslash}X}
    \begin{tabularx}{\textwidth}{YYY}
    \toprule
	& Ancilla-free method & Swap test \\
	\midrule
	\multirow{2}{*}{Number of measured qubits} &
	$\displaystyle \lg N$ & $\displaystyle 1$ \\[1mm]
	& [Higher] & [Lower] \\[3mm]
	{Sampling complexity for absolute error $\epsilon$ and confidence $\alpha$} &
	$\displaystyle \frac{1-p^2}{4 \epsilon^2}  \left[ \Phi^{-1} \!\left(\frac{1+\alpha}{2}\right) \right]^2$ & 
	$\displaystyle \frac{1-p^4}{4\epsilon^2 p^2} \left[ \Phi^{-1} \!\left(\frac{1+\alpha}{2}\right) \right]^2$ \\[4mm]
	& [Bounded in $p$] & [Unbounded for $p \to 0$] \\[3mm]
    Need for the loading unitary $U_1$ & Yes & No \\[3mm]
	\multirow{4}{*}{Depth} & $U_1^\dag$ is applied on the same register where $\ket{\psi_0}$ is loaded & $\ket{\psi_0}$ and $\ket{\psi_1}$ are loaded in parallel, and the test takes depth $O(\lg N)$ \\[1mm]
	& [Higher] & [Lower]\\
	\bottomrule
    \end{tabularx}
    \caption{Comparison of the ancilla-free method and the swap test for the calculation of the inner product $p$ between two statevectors $\ket{\psi_0}$ and $\ket{\psi_1}$.}
    \label{tab:inner-compare}
\end{table*}

\begin{figure}
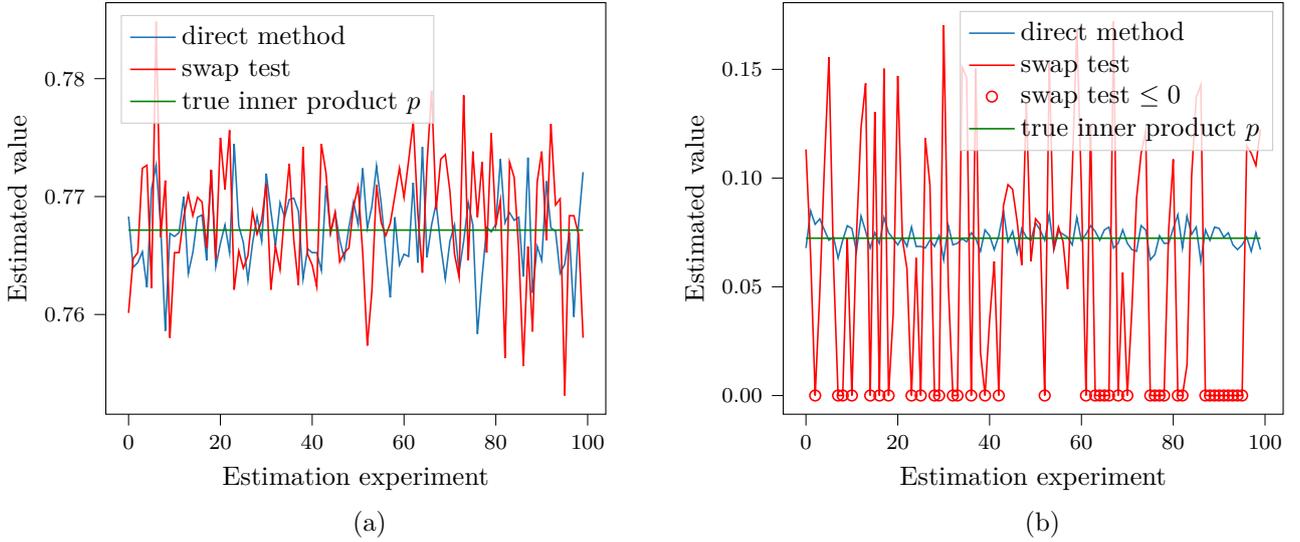

    \centering\small
    \include{Figures/innerproducts1}
    \hspace{1.5cm} (a) \hspace{.48\textwidth} (b)
    \caption{Each figure shows 100 independent estimations of the inner product between the same two statevectors. All estimations are obtained through 10,000 shots of the circuits on a noiseless simulator, so that the different outcomes are only an effect of random sampling. (a) The target inner product is big ($p=0.767$). In this case, the variance is $1.01 \cdot 10^{-5}$ for the ancilla-free method, and $2.77 \cdot 10^{-5}$ for the swap test. (b) The behavior of the swap test worsens when $p$ is small (in this case, $p=0.072$). The variance is $2.60 \cdot 10^{-5}$ for the ancilla-free method and $3.20 \cdot 10^{-3}$ for the swap test, which is significantly higher. Notice the presence of runs that provide a negative squared inner product, in 38 cases out of 100.}
    \label{fig:inner-compare}
\end{figure}

\subsection{Using the inner product after the Quantum Hadamard Product}\label{subsec:qhp+inner}
Two aspects must be taken into account to exploit the inner product techniques expressed so far into our algorithm based on the polynomial expansion: on one hand, the effect of the rescaling factors $\rho_T$ and $\rho_E$ on the sampling complexity, and on the other one, the consequences of the success rate $a_k^{-2}$ of the QHPs.
\begin{figure}
    \centering
    (a)
    \tikzsetnextfilename{qhp-swap1}
    \begin{quantikz}[row sep={0.8cm,between origins}]
    \lstick{$\ket{0}$} & \qwbundle{n} & \gate{U_T} & \ctrl{1} \gategroup[wires=2,steps=2,style={dashed, rounded corners, fill=black!10}, background]{} & \qw{} & \qw{} & \qw{} \gategroup[wires=4,steps=3,style={dashed, rounded corners, fill=blue!10}, background]{} & \swap{2} & \qw{} & \qw{} \\
    \lstick{$\ket{0}$} & \qwbundle{n} & \gate{U_T} & \targ{} & \meterD{\ketbra{0}{0}} \gategroup[wires=1,steps=1,style={draw=none}, label style={label position=above, anchor=north}]{$Z_1=0$} \\
    \lstick{$\ket{0}$} & \qwbundle{n} & \gate{U_E} & \qw{} & \qw{} & \qw{} & \qw{} & \swap{1} & \qw{} & \qw{} \\
    \lstick{$\ket{0}$} & \qw{} & \qw{} & \qw{} & \qw{} & \qw{} & \gate{H} & \ctrl{} & \gate{H} & \meter{$X=1$}
    \end{quantikz}
    \\[8mm]
    (b)
    \tikzsetnextfilename{qhp-swap2}
    \begin{quantikz}[row sep={0.8cm,between origins}]
    \lstick{$\ket{0}$} & \qwbundle{n} & \gate{U_T} & \ctrl{1} & \swap{2} & \qw{} & \qw{} \\
    \lstick{$\ket{0}$} & \qwbundle{n} & \gate{U_T} & \targ{} & \qw{} & \qw{} & \meter{$Z_1=0$}\\
    \lstick{$\ket{0}$} & \qwbundle{n} & \gate{U_E} & \qw{} & \swap{1} & \qw{} & \qw{} \\
    \lstick{$\ket{0}$} & \qw{} & \qw{} & \gate{H} & \ctrl{} & \gate{H} & \meter{$X=1$}
    \end{quantikz}
    \caption{An exemplary demonstration of the algorithm based on QHP and the swap test, for the calculation of the inner product between temperatures to the power $k$ and prices, when $k=2$. $U_T$ and $U_E$ are the loading unitaries for $T$ and $E$ respectively. In (a), the algorithm in its original formulation, where the gray box highlights the QHP, and the blue box the swap test. In (b), the version without mid-measurements. All measurements are deferred to the end.}
    \label{fig:qhp+swap}
\end{figure}
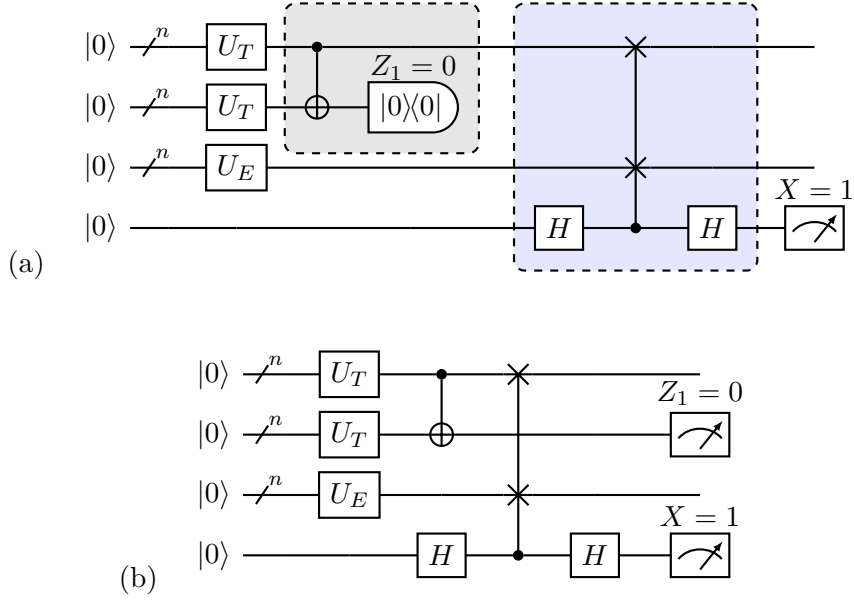
\begin{prop}[Algorithm QHP + swap test in amplitude encoding]\label{prop:errorbound-qhp+swap}
Let $k$ be a fixed power order. Consider a circuit that produces the state $\ket{\psi_T^{(k)}}$ defined in Eq.~\eqref{eq:psiTpower} through QHPs (with or without mid-measurements), then loads $\ket{\psi_E}$ and applies the swap test between $\ket{\psi_T^{(k)}}$ and $\ket{\psi_E}$, as depicted in Fig.~\ref{fig:qhp+swap}. Call $X \in \{0,1\}$ the output of the measurement of the control qubit in the swap test, and $\mathbf{Z} \in \{0, ..., N-1\}^{k-1}$ the outputs of all the $k-1$ measurements in the QHPs. Define $X_i \sim X$ and $\mathbf{Z}_i \sim \mathbf{Z}$, for $i=1,...,S$, as the outcomes of $S$ independent samples from the circuit. Let
\begin{equation*}
Y_S := \sqrt{ \frac{2 \# \{i : X_i = 0, \mathbf{Z}_i = \mathbf{0} \} - \# \{i: \mathbf{Z}_i = \mathbf{0}\}}{S}}; \qquad Y'_S := \rho_T^{-k} \rho_E^{-1} Y_S.
\end{equation*}
Then
\begin{enumerate}
    \item
$\mathbb{E}[Y_S] \to \sum_j E_j T_j^k =: y_k$ when $S\to\infty$ and $\mathbb{E}[Y'_S] \to \sum_j E'_j (T'_j-\eta)^k =: y'_k$ when $S\to\infty$;

    \item
assuming $y_k \neq 0$, the absolute error for $Y_S$ is controlled by $\mathbb{P}\left( \abs{ Y_S - y_k} < \epsilon \right) \leq \alpha$ once $S$ is chosen as
\begin{equation} \label{eq:errorbound-qhp+swap:claim2}
    S \geq \max \left\{ 4 y_k^2 (a_k^2-1), \ \frac{1-a_k^4 y_k^4}{a_k^2 y_k^2} \right\} \frac{1}{\epsilon^2} \left[ \Phi^{-1} \!\left(\frac{3+\alpha}{4}\right) \right]^2
\end{equation}
asymptotically when $\epsilon \to 0$, where $\Phi$ is the CDF of the standard normal distribution;

    \item
assuming again $y_k \neq 0$, $\mathbb{P}\left( \abs{ Y_S - y_k} < \epsilon \right) \leq \alpha$ is also guaranteed by the stronger condition
\begin{equation} \label{eq:errorbound-qhp+swap:claim3}
    S \geq \max \left\{ 4, \ a_k^{-2} y_k^{-2} \right\} \frac{1}{\epsilon^2} \left[ \Phi^{-1} \!\left(\frac{3+\alpha}{4}\right) \right]^2
\end{equation}
asymptotically when $\epsilon \to 0$;

    \item
assuming $y'_k \neq 0$, any of the conditions in Eqs.~\eqref{eq:errorbound-qhp+swap:claim2} or~\eqref{eq:errorbound-qhp+swap:claim3} is also sufficient to control the error of the originally scaled problem in the following sense: $\mathbb{P}\left( \abs{ Y'_S - y'_k} < \rho_T^{-k} \rho_E^{-1} \epsilon \right) \leq \alpha$.
\end{enumerate}
\end{prop}
\begin{proof}
$X$ is a Bernoulli r.v. and, by the swap test theory,
\begin{equation}\label{eq:prob-swap1}
    \mathbb{P}(X = 0 | \mathbf{Z} = \mathbf{0}) = \frac{1}{2}+\frac{1}{2} \left| \braket{\psi_E}{\psi_T^{(k)} } \right|^2
    = \frac{1}{2}+\frac{1}{2} a_k^2 \left( \sum_{j=0}^{N-1} E_j T_j^k \right)^2
\end{equation}
On the other hand,
\begin{equation}\label{eq:prob-swap2}
    \mathbb{P}(X = 0 | \mathbf{Z} = \mathbf{0}) = \frac{\mathbb{P}(X = 0, \mathbf{Z} = \mathbf{0})}{\mathbb{P}( \mathbf{Z} = \mathbf{0})}.
\end{equation}
Recalling that
\begin{equation}\label{eq:successrate}
    \mathbb{P}( \mathbf{Z} = \mathbf{0}) = a_k^{-2},
\end{equation}
by Eqs.\eqref{eq:prob-swap1} and~\eqref{eq:prob-swap2} we derive
$$\mathbb{P}(X = 0, \mathbf{Z} = \mathbf{0}) = \frac{1}{2 a_k^2}+\frac{1}{2} \left( \sum_{j=0}^{N-1} E_j T_j^k \right)^2$$
and therefore
\begin{equation}\label{eq:inner-from-probabilities}
    2 \mathbb{P}(X = 0, \mathbf{Z} = \mathbf{0}) - \mathbb{P}( \mathbf{Z} = \mathbf{0}) = \left( \sum_{j=0}^{N-1} E_j T_j^k \right)^2 .
\end{equation}
The first claim is an application of the law of large numbers. For the second part of this claim, consider that $y_k' = \rho_T^{-k} \rho_E^{-1} y_k$.

Moving to the second claim, observe that
$$Y_S = \sqrt{\abs{ 1-2 \frac{ \sum_{i: \mathbf{Z}_i = \mathbf{0}} X_i}{\# \{i: \mathbf{Z}_i = \mathbf{0}\}}}}
\sqrt{ \frac{\# \{i: \mathbf{Z}_i = \mathbf{0}\}}{S}}.$$

Therefore, let us name the two square roots $A_S$ and $B_S$ respectively. We know $B_S \to a_k^{-1}$ from Eq.~\eqref{eq:successrate}, then
\begin{equation*}
\begin{split}
    \mathbb{P}\left( \abs{ Y_S - y_k} < \epsilon \right)
    & = \mathbb{P}\left( \abs{ A_S B_S - A_S a_k^{-1} + A_S a_k^{-1} - y_k} < \epsilon \right) \\
    & \geq \mathbb{P}\left( \abs{ A_S B_S - A_S a_k^{-1}} < \epsilon /2 \text{ and } \abs{A_S a_k^{-1} - y_k} < \epsilon /2 \right)\\
    & = 1- \mathbb{P}\left( \abs{ A_S B_S - A_S a_k^{-1}} > \epsilon /2 \text{ or } \abs{A_S a_k^{-1} - y_k} > \epsilon /2 \right)\\
    & \geq 1- \mathbb{P}\left( \abs{ A_S B_S - A_S a_k^{-1}} > \epsilon /2 \right) - \mathbb{P}\left( \abs{A_S a_k^{-1} - y_k} > \epsilon /2 \right)\\
    & = \mathbb{P}\left( \abs{ A_S B_S - A_S a_k^{-1}} < \epsilon /2 \right) + \mathbb{P}\left( \abs{A_S a_k^{-1} - y_k} < \epsilon /2 \right) -1\\
    & \geq 2 \beta -1 = \alpha,
\end{split}
\end{equation*}
once we prove that the following inequalities hold for $\beta = \frac{1+\alpha}{2}$:
\begin{empheq}[left=\empheqlbrace]{align}
    & \mathbb{P}\left( \abs{ A_S B_S - A_S a_k^{-1}} < \epsilon /2 \right) \geq \beta \label{eq:errorbound-qhp+swap:proof2a} \\
    & \mathbb{P}\left( \abs{A_S - a_k y_k} < a_k \epsilon /2 \right) \geq \beta \label{eq:errorbound-qhp+swap:proof2b}
\end{empheq}

Let us start with the latter. Apply Prop.~\ref{prop:errorbound-swap} to $X_i$ conditioned to $\mathbf{Z}_i = \mathbf{0}$, taking $a_k y_k$ as the $p$ in Prop.~\ref{prop:errorbound-swap},  $\epsilon a_k / 2$ as the $\epsilon$ in Prop.~\ref{prop:errorbound-swap}, $\# \{i: \mathbf{Z}_i = \mathbf{0}\}$ as the $S$ in Prop.~\ref{prop:errorbound-swap}, and $\beta$ as the $\alpha$ in Prop.~\ref{prop:errorbound-swap}, thus getting
$$\# \{i: \mathbf{Z}_i = \mathbf{0}\} = \frac{1-a_k^4 y_k^4}{a_k^4 \epsilon^2 y_k^2} \left[ \Phi^{-1} \!\left(\frac{1+\beta}{2}\right) \right]^2. $$
Since $\# \{i: \mathbf{Z}_i = \mathbf{0}\}$ is asymptotic to $a_k^{-2} S$, Eq.~\eqref{eq:errorbound-qhp+swap:claim2} guarantees the last expression and therefore Eq.~\eqref{eq:errorbound-qhp+swap:proof2b}.

Let us consider Eq.~\eqref{eq:errorbound-qhp+swap:proof2a} now. Since $A_S$ tends to $a_k y_k$, it is definitely dominated by $2a_k y_k$. Eq.~\eqref{eq:errorbound-qhp+swap:proof2a} is then implied by
$$\abs{B_S - a_k^{-1}} < \frac{\epsilon}{4 a_k y_k}.$$
Now, $B_S^2$ is the mean of $S$ i.i.d. Bernoulli variables with $\mu=a_k^{-2}$ and $\sigma^2 = \mu (1-\mu) = a_k^{-2} (1-a_k^{-2})$. Applying Lemma~\ref{lemma:clt-sqrt}, an asymptotically sufficient condition for Eq.~\eqref{eq:errorbound-qhp+swap:proof2a} is
$$S \geq \frac{a_k^{-2} (1-a_k^{-2})}{4 a_k^{-2} (\epsilon^2 / 16 a_k^2 y_k^2)} \left[ \Phi^{-1} \!\left(\frac{1+\beta}{2}\right) \right]^2,$$
which is again implied by Eq.~\eqref{eq:errorbound-qhp+swap:claim2}. This way the second claim is proved.

The third claim derives from the second one, once we consider that $y_k = \sum_j E_j T_j^k \leq \norm{E} \norm{T^k} = a_k^{-1}$. The last claim is trivial.
\end{proof}
\begin{prop}[Algorithm QHP + ancilla-free method in amplitude encoding] \label{prop:errorbound-qhp+direct}
Let $k$ be a fixed power order. Implement a circuit that produces the state $\ket{\psi_T^{(k)}}$ defined in Eq.~\eqref{eq:psiTpower} through QHPs (with or without mid-measurements), then loads $\ket{\psi_E}$ and applies the ancilla-free method for the inner product between $\ket{\psi_T^{(k)}}$ and $\ket{\psi_E}$, and subsequently measures the target register, as depicted in Fig.~\ref{fig:qhp+direct}(a) and~(b). Let $R \in \{0, ..., N-1\}$ be the measurement output of the target register, let $\mathbf{Z} \in \{0, ..., N-1\}^{k-1}$ the outputs of all the $k-1$ measurements in the QHPs, and let $X$ be a r.v. valued 1 if $R=0$ and $\mathbf{Z}=\mathbf{0}$, and valued 0 otherwise. Consider $S$ independent shots, and let $X_i \sim X$ be their outcomes, for $i=1,...,S$. Finally define
\begin{equation*}
Y_S :=  \sqrt{\bar X_S}; \qquad Y'_S := \rho_T^{-k} \rho_E^{-1} Y_S.
\end{equation*}
Then
\begin{enumerate}
    \item
$\mathbb{E}[Y_S] \to \sum_j E_j T_j^k =: y_k$ when $S\to\infty$ and $\mathbb{E}[Y'_S] \to \sum_j E'_j (T'_j-\eta)^k =: y'_k$ when $S\to\infty$;

    \item
assuming $y_k \neq 0$, the absolute error for $Y_S$ is controlled by $\mathbb{P}\left( \abs{ Y_S - y_k} < \epsilon \right) \leq \alpha$ once $S$ is chosen as
\begin{equation}\label{eq:errorbound-qhp+direct:claim2}
    S = \frac{1 - y_k^2}{4 \epsilon^2}  \left[ \Phi^{-1} \!\left(\frac{1+\alpha}{2}\right) \right]^2
\end{equation}
asymptotically when $\epsilon \to 0$, where $\Phi$ is the CDF of the standard normal distribution;

    \item
assuming again $y_k \neq 0$,  $\mathbb{P}\left( \abs{ Y_S - y_k} < \epsilon \right) \leq \alpha$ is also guaranteed by the stronger condition
\begin{equation} \label{eq:errorbound-qhp+direct:claim3}
    S = \frac{1}{4 \epsilon^2}  \left[ \Phi^{-1} \!\left(\frac{1+\alpha}{2}\right) \right]^2
\end{equation}
asymptotically when $\epsilon \to 0$;

    \item
assuming $y'_k \neq 0$, any of the conditions in Eqs.~\eqref{eq:errorbound-qhp+direct:claim2} or~\eqref{eq:errorbound-qhp+direct:claim3} is also sufficient for the error of the originally scaled problem in the following sense: $\mathbb{P}\left( \abs{ Y'_S - y'_k} < \rho_T^{-k} \rho_E^{-1} \epsilon \right) \leq \alpha$.
\end{enumerate}
\end{prop}
\begin{proof}
Consider that
$$
\mathbb{E}[X] =
\mathbb{P}(R=0, \mathbf{Z}=\mathbf{0})=
\mathbb{P}(R=0 | \mathbf{Z}=\mathbf{0}) \, \mathbb{P}(\mathbf{Z}=\mathbf{0})=
a_k^2 y_k^2 a_k^{-2} = y_k^2
$$
Now the first claim is an application of the law of large numbers. For the second part of this claim, consider that $y_k' = \rho_T^{-k} \rho_E^{-1} y_k$. The second claim is an application of Lemma~\ref{lemma:clt-sqrt}. The third and fourth claims are obvious.
\end{proof}

\subsection{Circuit width and depth}

\begin{prop}\label{prop:width-depth-k}
    Let $k$ be a fixed power order. Let an amplitude-encoding state preparation routine be given for both $\ket{\psi_T}$ and $\ket{\psi_P}$. Suppose each state preparation works in depth $C_{\mathrm{d}, \mathrm{load}}(N)$. Then the algorithm described in Prop.~\ref{prop:errorbound-qhp+direct} (called QHP + ancilla-free method) and that in Prop.~\ref{prop:errorbound-qhp+swap} (QHP + swap test) have the following width and depth:
    \begin{equation}\label{eq:width-k}
        C_{\mathrm{w}}(k) = r(k) \lg N+ \delta^{(\mathrm{swap})},
    \end{equation}
    \begin{equation}\label{eq:depth-k}
        C_{\mathrm{d}}(k) \leq m(k) \, C_{\mathrm{d}, \mathrm{load}}(N) + k + (3\lg N + 1) \delta^{(\mathrm{swap})},
    \end{equation}
    where
    \begin{equation*}
        \delta^{(\mathrm{swap})}=\begin{cases}
              1 \quad &\text{if swap test is used,} \\
              0 \quad &\text{if ancilla-free method is used,} \\
         \end{cases}
    \end{equation*}
    \begin{equation*}
        r(k)=\begin{cases}
              2 \quad &\text{if mid-reset is used,} \\
              k+\delta^{(\mathrm{swap})} \quad &\text{otherwise,} \\
         \end{cases}
    \end{equation*}
    \begin{equation*}
        m(k)=\begin{cases}
              k + (1-\delta^{(\mathrm{swap})}) \quad &\text{if mid-reset is used and $\mu_{1}=\dots=\mu_{k-1}=0$,}\\
              t + (1-\delta^{(\mathrm{swap})}) \quad &\text{if mid-reset is used, $\mu_{1}=\dots=\mu_{t-1}=0$ and $\mu_t=1$, some $t \leq k-1$,}\\
              1 + (1-\delta^{(\mathrm{swap})}) \quad &\text{if no mid-reset is used,}
         \end{cases}
    \end{equation*}
    $\mu_t=0$ being a successful application of the QHP (namely $\mu_t$ being the output of measuring the $Z_t$ variable, defined in Fig.~\ref{fig:qhp+swap} and Fig.~\ref{fig:qhp+direct}).
    If mid-reset is used $m(k) \leq k$ and
    \begin{equation*}
        \mathbb E [m(k)] = 
              k a_k^{-2(k-1)} + (1-a_k^{-2}) \sum_{k'=1}^{k-1} k' a_k^{-2 (k'-1)}.
    \end{equation*}
\end{prop}
\begin{proof}
Let us start calculating the space complexity $C_{\mathrm{w}}$ of the quantum circuit, namely the circuit width. The width required to load a data set of size $N$ is $\lg N$. Let us now justify the prefactor $r(k)$. When calculating the width required to encode data for the calculation, two scenarios must be taken into account. If we do not resort to mid-measurements, $k$ copies of $\ket{ \psi_T }$ in different registers are needed, plus one copy of $\ket{ \psi_E }$ that lies in a different register only in the case of the swap test. If we conversely can apply mid-circuit resets, the number of registers can be reduced to $2$, regardless of $k$.
Finally, the swap test requires only one additional qubit, and the overall space cost is that in Eq.~\eqref{eq:width-k}.

Moving to depth, if mid-reset is not used, the data encoding of the $k$ copies $\ket{ \psi_T }$ is performed in parallel, as well as that of $\ket{ \psi_E }$. Vice versa, if mid-reset is used, data encoding is done in series, and in a given shot, an iteration of encoding is performed only if the measurement of the previous iteration was successful. This is called \textit{dynamic stopping}.
The inequality and the expectation of $m(k)$ are obvious.

To complete the derivation of Eq.~\eqref{eq:depth-k}, notice that once data is loaded, the additional depth for each QHP is 1, since all CNOTs can be performed in parallel. Consequently, the additional depth required to produce $\ket{ \psi^{(k)}_T }$, is $t$, and is therefore dominated by $k$. Finally, the swap test has a depth of $3n+1$, as all swaps are controlled by the same ancilla.
\end{proof}

\subsection{Putting powers together}\label{subsec:assemble-k}
So far, we constructed an algorithm providing an estimator $Y_S$, called $Y_k$ in this subsection, that is able to approximate $y_k$, for a single power $k$, up to a given error. Now, following Eq.~\eqref{eq:finalsum}, we define an estimator
\begin{equation}\label{eq:estimatorV}
    V := \sum_{k=0}^K \rho_T^{-k} \rho_E^{-1} b_k(\eta) Y_k
\end{equation}
and we want to verify that $V$ is a good approximation for $v$. As a part of out asymptotic analysis, we shall discuss the error scaling when $N \to \infty$. Since we can expect the contract value to be affected by the growth of $N$, the analysis must be conducted under \textit{relative} error.
By means of the next Proposition, we split the task of estimating the error between $V$ and $v$, into a first calculation of the error between $V$ and $v^*$, and a subsequent one of the error between $v^*$ and $v$.

\begin{prop}[Error triangulation with $v^*$]\label{prop:error-tria}
    Suppose that there exist a relative error $\epsilon > 0$ and a confidence level $\beta > 0$ such that
    \begin{enumerate}[label=(\roman*)]
        \item $V$ is a good estimator for $v^*$, namely
    \begin{equation}\label{eq:V-goodestimator}
        \mathbb{P}\left(\abs{V-v^*} \leq \epsilon \abs{v} \right) \geq \beta
    \end{equation}
    where $V$ is defined by Eq.~\eqref{eq:estimatorV} and $v^*$, introduced in Eq.~\eqref{eq:contractvalue-approx}, is
    $$v^* = \sum_{j=0}^{N-1} E'_j \; p(T'_j) =
    \sum_{j,k} b_k(\eta) \; E'_j \; (T'_j-\eta)^k = \sum_{k=0}^{K} b_k(\eta) \; y'_k = \sum_{k=0}^{K} \rho_T^{-k} \rho_E^{-1} b_k(\eta) \; y_k;$$
        \item $v^*$ is a good approximation for $v$, namely
    \begin{equation}\label{eq:vstar-goodproxy}
        \mathbb{P}\left(\abs{ v^* - v} \leq \epsilon \abs{v} \right) \geq \beta.
    \end{equation}
    \end{enumerate}
    Then, $V$ is also a good estimator for $v$:
    $$\mathbb{P}\left(\abs{V-v} \leq 2 \epsilon \abs{v} \right) \geq 2 \beta -1.$$
\end{prop}
\begin{proof}
Consider that
\begin{equation*}
\begin{split}
\mathbb{P}\left(\abs{V-v} \leq 2 \epsilon \abs{v} \right)
&\geq \mathbb{P}\left(\abs{V-v^*} \leq \epsilon \abs{v} \text{ and } \abs{v^*-v} \leq \epsilon \abs{v} \right) \\
&= 1- \mathbb{P}\left(\abs{V-v^*} \geq \epsilon \abs{v} \text{ or } \abs{v^*-v} \geq \epsilon \abs{v} \right) \\
&\geq 1- \mathbb{P}\left(\abs{V-v^*} \geq \epsilon \abs{v} \right) - \mathbb{P}\left(\abs{v^*-v} \geq \epsilon \abs{v} \right)\\
&= 1- \mathbb{P}\left(\abs{V-v^*} \geq \epsilon \abs{v} \right) + 1 - \mathbb{P}\left(\abs{v^*-v} \geq \epsilon \abs{v} \right) -1\\
&= \mathbb{P}\left(\abs{V-v^*} \leq \epsilon \abs{v} \right) + \mathbb{P}\left(\abs{v^*-v} \leq \epsilon \abs{v} \right) -1\\
&\geq 2\beta -1.
\end{split}
\end{equation*}
The proof is thus given.
\end{proof}

In condition (ii) of the previous Proposition, we introduced a probability of $v^*$ being close to~$v$, implying that the underlying temperatures and prices are sampled from a random process. This subsection is the only portion of the Appendices where such stochasticity is explicitly utilized. The effect of the shape of said distributions is further studied in the next Proposition, which develops sufficient conditions for (ii) of Prop.~\ref{prop:error-tria}. In the rest of the Appendices, on the contrary, condition (i) is examined, so that prices and temperatures (as well as $v$ and $v^*$) are assumed to be fixed, and whenever we talk about probabilities, the randomness derives from the quantum measurements only.

\begin{prop}[Sufficient condition for $v^* \approx v$]\label{prop:sufficient-vapprox}
    The following set of conditions is sufficient to satisfy Eq.~\eqref{eq:vstar-goodproxy}:
    \begin{enumerate}[label=(\roman*)]
        \item
    $\abs{v}$ scales linearly with $N$ when $N$ grows, namely $\abs{v} \sim c N$, some $c \in (0, +\infty)$;
        \item
    prices are bounded, namely $E'_j$ are valued in a bounded interval $[-B, B]$, with $B>0$;
        \item 
    the polynomial $p(T'_j) := \sum_k b_k(\eta) \; (T'_j - \eta)^k$, introduced in Eq.~\eqref{eq:poly}, is a good approximation for $f(T'_j)$ under the distribution of temperatures, in the following sense:
    $$\mathbb P \left( \frac{1}{N} \abs{\sum_{j=0}^{N-1} (f(T'_j)- p(T'_j))} \leq \frac{c}{B} \epsilon \right) \geq \beta$$ definitely when $N \to \infty.$
    \end{enumerate}
\end{prop}
\begin{proof}
Write
\begin{equation*}
\frac{\abs{v-v^*}}{\abs{v}}
\sim \frac{1}{c N} \abs{v-v^*}
\overset{\Delta}{=} \frac{1}{cN} \abs{\sum_{j=0}^{N-1} (f(T'_j)- p(T'_j)) E'_j}
\leq \frac{B}{c} \frac{1}{N} \abs{\sum_{j=0}^{N-1} (f(T'_j)- p(T'_j))}.
\end{equation*}
\end{proof}

\begin{remark}[Linear scaling of $v$]\label{rem:v-scaling}
    In the previous Proposition, we introduced the assumption that $v$ scales linearly with $N$: let us now discuss it. If the time horizon is doubled in length, it is reasonable to think that the overall contract value doubles, under stationary processes. More technically, this is the case for instance if $T'_j$ are i.i.d. samples drawn from a same random variable $T'$, and $E'_j$ from $E'$, such that $f(T') E' \in L^1$: indeed, by the law of large numbers,
    $$ \frac{v}{N} = \frac{1}{N} \sum_{j=0}^{N-1} f(T'_j) E'_j \to \mathbb E [f(T') E']$$
    that is a positive quantity since volumes and prices are.
    
    Our cost-scaling analysis of the relative error under the assumption that $v$ behaves linearly, is equivalent to the study of the cost-scaling for an error $\epsilon / N$, which is the approach followed by Ref.~\cite{guo_nonlinear_2021}.
\end{remark}

\begin{remark}
    In the simplified case where temperatures are i.i.d. samples from a same random variable $T'$, and assuming $f(T'), p(T') \in L^1$, then condition (iii) in the previous proposition is guaranteed by the Central Limit Theorem, definitely in $N$, when
    \begin{equation}\label{eq:p-goodapprox}
        \mathbb E [f(T') - p(T')] < \frac{c}{B}\epsilon.
    \end{equation}
    Eq.~\eqref{eq:p-goodapprox} exhibits that condition (iii) of the proposition is indeed a requirement that $p$ well approximates $f$ under the distribution of temperatures.
\end{remark}

With Prop.~\ref{prop:sufficient-vapprox} and its remarks, we have completed the discussion of condition (ii) in Prop.~\ref{prop:error-tria}. Let us now turn to (i), which is the core convergence discussion for the quantum algorithm. Below in this subsection we argument (i) in the context of QAE-free methods in amplitude encoding. We explore (i) for QAE-based methods and for the BOE encoding in the next Appendices.

\begin{prop}[Convergence rate in amplitude encoding]\label{prop:assemble-k}
    Let $w \in [0,1]^K$ such that $\sum_k w_k = 1$. Also, let $\alpha \in [0, 1)^K$ such that $\sum_k (1- \alpha_k) = 1-\beta$ for some $\beta \in (0,1)$. Finally, let $\epsilon > 0$. For instance, one may take $w_k = K^{-1}$ for all $k$ and $\alpha_k = \frac{K-1+\beta}{K}$ for all $k$.

    Then $V$ defined in Eq.~\eqref{eq:estimatorV} is an estimator for $v^*$ such that
    $$\mathbb{P}\left(\abs{V-v^*} \leq \epsilon \right) \geq \beta, \qquad \text{where } v^* := \sum_{k=0}^{K} b_k(\eta) \; y'_k = \sum_{k=0}^{K} \rho_T^{-k} \rho_E^{-1} b_k(\eta) \; y_k, $$
    provided that $\mathbb{P}(\rho_T^{-k} \rho_E^{-1} \abs{b_k(\eta)} \ \abs{Y_k-y_k} \leq w_k \epsilon) \geq \alpha_k$ for all $k \leq K$.

    Consequently, a similar estimation holds for the relative error:
    $$\mathbb{P}\left(\abs{V-v^*} \leq \epsilon \abs{v} \right) \geq \beta,$$ provided that
    $\mathbb{P}\left( \abs{Y_k-y_k} \leq r_k^{-1} w_k \; \abs{b_k(\eta)}^{-1} \; \epsilon \right)  \geq \alpha_k$ for all $k \leq K$, where
    \begin{equation}\label{eq:defn-rk}
        r_k^{-1} := \abs{v} \; \rho_T^{k} \; \rho_E.
    \end{equation}

    Finally, let $K$ be fixed. When $N$ grows, $r_K$ dominates the asymptotic behavior of $r_k$ for all other $k \leq K$.
\end{prop}
\begin{proof}
Consider that
\begin{equation*}
\begin{split}
    \mathbb{P}\left(\abs{V-\sum_{k=0}^{K} \rho_T^{-k} \rho_E^{-1} b_k(\eta) \; y_k} \leq \epsilon \right)
    & \geq \mathbb{P}\left( \sum_k \rho_T^{-k} \rho_E^{-1} \abs{b_k(\eta)} \ \abs{Y_k-y_k} \leq \epsilon \right) \\
    & \geq \mathbb{P}\left(\rho_T^{-k} \rho_E^{-1} \abs{b_k(\eta)} \ \abs{Y_k-y_k} \leq w_k \epsilon \ \text{for all } k \right) \\
    & = 1- \mathbb{P}\left(\rho_T^{-k} \rho_E^{-1} \abs{b_k(\eta)} \ \abs{Y_k-y_k} \geq w_k \epsilon \ \text{for some } k \right) \\
    & \geq 1- \sum_{k=0}^K \mathbb{P}\left(\rho_T^{-k} \rho_E^{-1} \abs{b_k(\eta)} \ \abs{Y_k-y_k} \geq w_k \epsilon \right) \\
    & \geq 1- \sum_{k=0}^K (1-\alpha_k) = \beta.
\end{split}
\end{equation*}
The second claim is an application of the first one. As for the third,
if $r_k^{-1} \to 0$ for some $k$ when $N\to \infty$, then $\rho_T^k \rho_E \to 0$. In such case, since $\rho_T \leq 1$, $r_K^{-1}$ goes to 0 at least as fast as $r_k^{-1}$.
\end{proof}

\begin{remark}
    $\alpha_k = \frac{K-1+\beta}{K}$ is very close to $1$ if $K$ grows, implying a very high sampling complexity from Eq.~\eqref{eq:errorbound-qhp+swap:claim3} or~\eqref{eq:errorbound-qhp+direct:claim3}. Therefore the technique is effective only if $K$ is low.
    Additionally, let us point out that one may leverage the knowledge of $b_k(\eta)$ to refine the definition of $\alpha_k$ and $w_k$.
\end{remark}

It is clear from the previous proposition that the scaling of the error when $N$ grows is bound to that of the norms $\rho_E^{-1}$ and $\rho_T^{-1}$, as well as to the powers $k$.

\begin{remark}[Norm scaling] \label{rem:norm-scaling}
    In general, it is reasonable to assume that the norms $\rho_E^{-1}, \rho_T^{-1}$ scale as $\sqrt{N}$: indeed if $T'_j$ are sampled from a same r.v. $T'$, $\norm{T'_j-\eta} / \sqrt{N}$ tends to the finite quantity $\sqrt{\mathbb{E}[(T'-\eta)^2]}$ if the distribution is $L^2$, and similar for $E'$.
\end{remark}

\begin{remark}\label{rem:error-scaling}
    Combining the linear scaling of $v$ (Assumption~\ref{ass:v-scaling} and Remark~\ref{rem:v-scaling}) with the square-root scaling of norms (Assumption~\ref{ass:norm-scaling} and Remark~\ref{rem:norm-scaling}), one gets that $r_K$ scales as $N^{K/2+1/2-1} = N^{K/2-1/2}$.
\end{remark}

Fig.~\ref{fig:error-scaling-k} confirms the previous remark. Additionally, it shows the effect of $b_k$ in agreement with Prop.~\ref{prop:assemble-k}.
\begin{figure}[p]
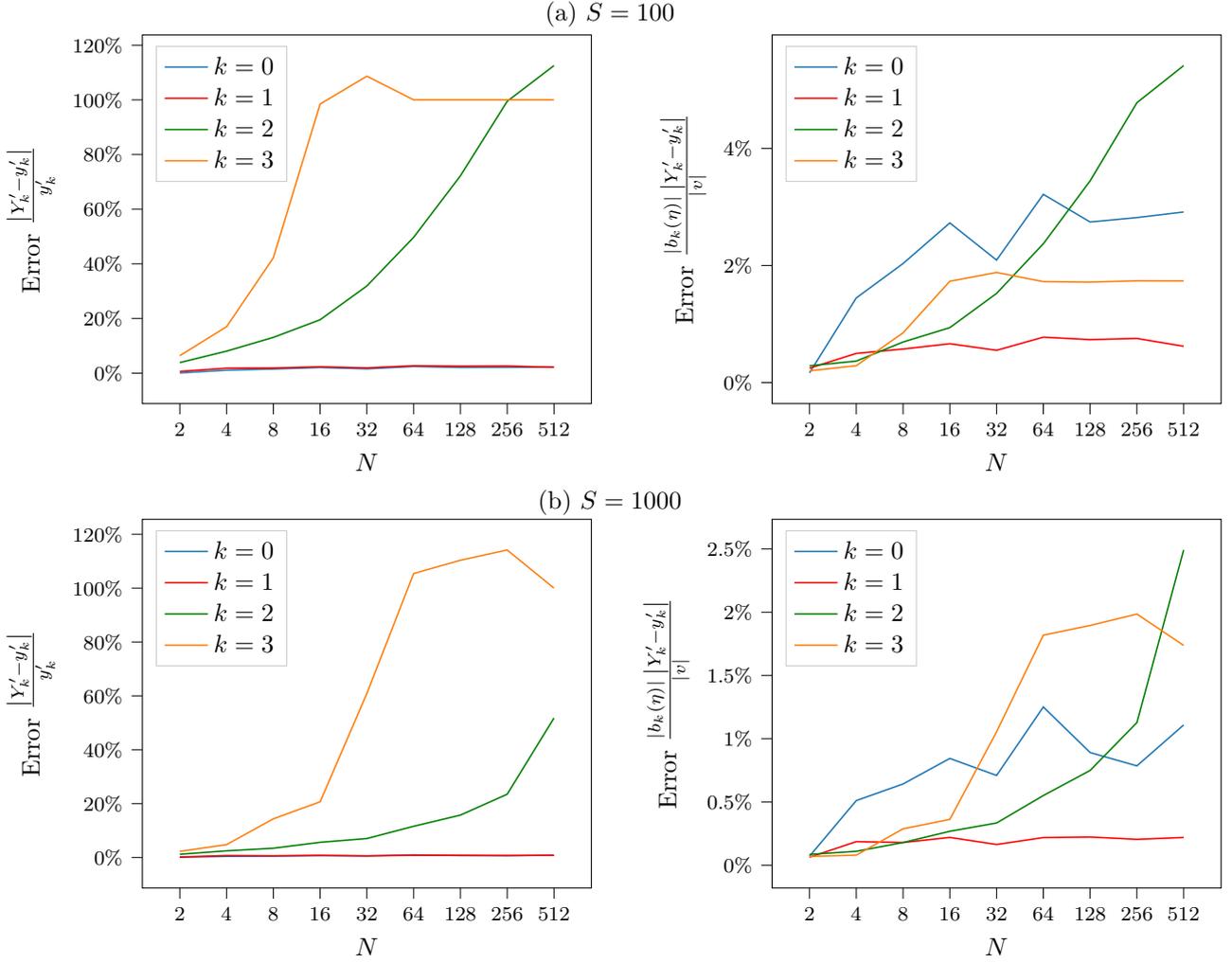

    \centering \small
    \include{Figures/error-scaling}
    \caption{The error of the estimator $Y'_k$, relatively to its own target value $y'_k$ (left) and the error of the power contribution $b_k(\eta) Y'_k$, relatively to the global target $v$ (right), averaged over $S$ samples, for the ancilla-free algorithm without mid-resets. Each point in the plot is the average of 50 independent runs on the qasm simulator. Since the polynomial %
    has a very low coefficient $b_3$, the error for $k=3$ provides a modest contribution to the overall result for small problem sizes.}
    \label{fig:error-scaling-k}
\end{figure}

\begin{example}\label{ex:error-scaling}
Consider a first problem. We are given two four-dimensional inputs $[T'_j]_0^3 = [x_j]_0^3$ and $[E'_j]_0^3 = [y_j]_0^3$, which we assume to be positive and normalized. For simplicity, take the volume function to be $f(x) = x^K$. Then the quantum algorithm is able to estimate $\sum_j (T'_j)^K E'_j = \sum_j x_j^K y_j$.

Now consider a second problem. This time, we are given as inputs the same values as before, but twice: so we have two eight-dimensional vectors $[T'_j]_0^7 = [x_0, ..., x_3, x_0, ..., x_3]$ and $[E'_j]_0^7 = [y_0, ..., y_3, y_0, ..., y_3]$. Obviously this time $\sum_j (T'_j)^K E'_j = 2 \sum_j x_j^K y_j$, and therefore we accept an error that is the double of the one that we would accept in the previous problem. Now, we apply the quantum algorithm: we encode $T_j = 2^{-1/2} T'_j$ and $E_j = 2^{-1/2} E'_j$, obtaining as a result $\sum 2^{-K/2} (T'_j)^K 2^{-1/2} E'_j = 2^{-K/2-1/2} \sum (T'_j)^K E'_j$, which needs to be rescaled by a factor $2^{K/2+1/2}$ to obtain the final result. Unfortunately though this rescaling implies an absolute error propagation that is not 2, but $2^{K/2+1/2}$.
\end{example}

Coherently with Remark~\ref{rem:error-scaling}, the previous examples shows that the relative error scales as $r_k = O(N^{k/2-1/2})$. If we additionally consider that the sampling complexity of the method described in this Appendix scales as $O(\epsilon^2)$, and that we need to add the circuit depth scaling on top to calculate the quantum time, we conclude that we can improve on the classical case $O(N)$ only for $K = 1$. In the next Appendix, we introduce QAE to partially overcome such limitation.

\section{Applying Quantum Amplitude Estimation techniques in amplitude encoding}\label{appendix:qae-ampl}
To outperform known classical results, we exploit the Quantum Amplitude Estimation technique~\cite{brassard_quantum_2002} and its variants, such as Faster QAE~\cite{nakaji_faster_2020}, Iterative QAE (IQAE)~\cite{grinko_iterative_2021}, Chebyshev QAE (ChebQAE)~\cite{rall_amplitude_2022} and Dynamic QAE~\cite{ghosh_energy_2024}.

Let us recall a result for the application of QAE to the estimation of a function expectation, due to Montanaro~\cite{montanaro_quantum_2015}. Suppose to be given a r.v. providing $x$ with probability $\abs{w_x}^2$, with $x \in \{0, \dots, 2^n-1\}$. Let $f$ be a real function defined on the same integer domain. Suppose to have an $n$-qubit loading unitary $\mathcal{A}$ such that $\mathcal{A} \ket{0} = \sum_x w_x \ket{x}$, and an $(n+1)$-qubit unitary $\mathcal{R}$ such that $\mathcal{R}\ket{x} \ket{0} = \ket{x} (\sqrt{1-f(x)}\ket{0} + \sqrt{f(x)}\ket{1})$. The objective is to estimate $\mathbb{E}[f(\mathcal{A})] = \sum_x f(x) \abs{w_x}^2$.

Define $\mathcal{F}:=\mathcal{R} (\mathcal{A} \otimes I)$ and $\ket{\chi} := \mathcal{F} \ket{0}$. Let $\mathcal{Z} := I - 2 \ketbra{0}{0}$ and $\mathcal{U} := I - 2 \ketbra{\chi}{\chi} = \mathcal{F} \mathcal{Z} \mathcal{F}^\dag$. The following holds:

\begin{thm}[QAE scaling for the estimation of the mean of a bounded function]\label{thm:qae-bounded-values}
    Let $f$ and $\mathcal{A}$ as defined above, such that $f(\mathcal{A})$ is valued in $[0,1]$. Let the desired accuracy be $\epsilon$. There exists a quantum algorithm, called QAE, that uses $O (1)$ copies of $\ket{\chi}$ and uses $\mathcal{U}$ for a number of times $O (1/\epsilon)$, and estimates $z := \mathbb{E}[f(\mathcal{A})]$ up to an additive error $\epsilon$ with success probability at least $8/\pi^2 > 0.81$. It suffices to sample from the quantum circuit $S=O(\lg 1/(1-\alpha))$ times and take the median $Z_S$ to obtain an estimate such that $\mathbb{P}(\abs{Z_S-z} \leq \epsilon) \geq \alpha$.
\end{thm}
\begin{proof}
    See \cite[Thm 2.3 and Lemma 2.1]{montanaro_quantum_2015}. Also refer to \cite{rebentrost_quantum_2018}.
\end{proof}
The idea of Montanaro's approach is to connect the desired expectation value to an eigenfrequency of an oscillating quantum system and then use the phase estimation algorithm to obtain the estimation up to a desired accuracy \cite{rebentrost_quantum_2018}. The desired expectation $z$ is linked to the corresponding phase $\theta$ via
\begin{equation}\label{eq:z-theta}
    z = \frac{1}{2} \left(1- \cos \frac{\theta}{2} \right),
\end{equation}
and similarly the estimator for $z$ is defined by
\begin{equation}
    Z = \frac{1}{2} \left(1- \cos \frac{\hat \theta}{2} \right).
\end{equation}
The error in $z$ is then linearly controlled by the error in $\theta$ when $\hat \theta \to \theta$ by
\begin{equation}\label{eq:hatz-error-bound}
    \abs{Z - z} = \frac{1}{2} \sin \frac{\theta}{2} \abs{\hat \theta - \theta} + o \left( \abs{\hat \theta - \theta} \right) = O\left(\abs{\hat \theta - \theta} \right),
\end{equation}
through a Taylor expansion of Eq.~\eqref{eq:z-theta}, as shown for instance in Ref.~\cite[Appendix F]{rebentrost_quantum_2018}.

Now, let us apply the Theorem to our case. Specifically, we start discussing QAE where the unitary is taken from the ancilla-free method without mid-measurements, as depicted in Fig.~\ref{fig:qhp+direct}(c). Set $\mathcal{A}$ to be the full unitary of the ancilla-free method, that loads temperatures, computes QHPs without mid-measurements, and applies the inverse of the price loading. By Eq.~\eqref{eq:U1_inverse}, we get $w_0 = \sum_j T_j^k E_j$, while $w_x$ is garbage for $x \neq 0$. Therefore it is sufficient to define 
\begin{equation}
f(x) := \begin{cases}
          1 \quad &\text{if } x=0, \\
          0 \quad &\text{otherwise,}\\
     \end{cases}
\end{equation}
and the algorithm will estimate the desired inner product. Implementing $f$ through a quantum circuit $\mathcal{R}$ is trivial, as it is simply a multi-controlled NOT gate, testing all qubits in all registers to be 0. Now, $z$ is the squared inner product, so that we can use $Y_S := \sqrt{Z_S}$ to estimate the inner product $y := \sqrt{z}$. It turns out that the error bounds given for $z$ in Prop.~\ref{thm:qae-bounded-values} are also valid for $y=\sqrt{z}$, as stated by the following proposition.

\begin{prop}[Oracle complexity for QHP + ancilla-free method + QAE]\label{prop:errorbound-qae-qhp+direct}
    Let $f$ and $\mathcal{A}$ be those specified right above. By applying QAE, it is possible to estimate $y_k :=  \sum_j T_j^k E_j$ up to an additive error $\epsilon$ with success probability at least $8/\pi^2 > 0.81$, using $O(1)$ copies of $\ket{\chi}$ and using $\mathcal{U}$ for a number of times $O (1/\epsilon)$. It suffices to sample the circuit output $Y$ for $S=O(\lg 1/(1-\alpha))$ times and take the median $Y_S$ of the samples, to obtain an estimate such that $\mathbb{P}(\abs{Y_S-y_k} \leq \epsilon) \geq \alpha$.
\end{prop}
\begin{proof}
    Applying the Taylor expansion to $y = \sqrt{z}$, as done in Eq.~\eqref{eq:hatz-error-bound}, one obtains again
    \begin{equation}
        \abs{Y - y} = O\left(\abs{\hat \theta - \theta} \right)
    \end{equation}
    for $\hat \theta \to \theta$, uniformly in $\theta$. The rest of the proof of Thm.~\ref{thm:qae-bounded-values} flows alike.
\end{proof}

\begin{prop}[Oracle complexity for QHP + ancilla-free method + IQAE]\label{prop:errorbound-iqae-qhp+direct}
    The results of Prop.~\ref{prop:errorbound-qae-qhp+direct} are valid also if the Iterative QAE is applied instead of QAE.
\end{prop}
\begin{proof}
    The scaling of IQAE is grounded on the estimate Eq.~\eqref{eq:hatz-error-bound} too, refer to \cite[Algorithm 1 and Appendix B]{grinko_iterative_2021}. Therefore, the same argument of Prop.~\ref{prop:errorbound-qae-qhp+direct} can be adopted.
\end{proof}

To perform a comparison between the classical and the quantum case, the cost of a single query must be considered. The cost of $\mathcal{U}$ is obviously derived by the cost of $\mathcal{A}$ and of $\mathcal{R}$. Now, the cost of $\mathcal{A}$ was already calculated in Prop.~\ref{prop:width-depth-k}. As far as $\mathcal{R}$ is concerned, Ref.~\cite{barenco1995elementary} shows that an $n$-controlled NOT can either be implemented with $1$ ancilla qubit and $O(2^n)$ gates and depth \cite[Lemma~7.1]{barenco1995elementary}, or with $n-1$ ancillas with $O(n)$ gates and depth \cite[Lemma~7.2]{barenco1995elementary}. Ref.~\cite{maslov_advantages_2016} further improved the ancillas necessary to achieve a linear depth, to $\ceil{\frac{n-3}{2}}$ for $n \geq 5$.

\begin{prop}\label{prop:depth-width-oracle}
    The implementation of the oracle $\mathcal{U}$ required by Prop.~\ref{prop:errorbound-qae-qhp+direct} and~\ref{prop:errorbound-iqae-qhp+direct}, has a width and depth respectively of:
    \begin{equation}\label{eq:width-k-oracle}
        C_{\mathrm{w}}(k) = O(k \lg N),
    \end{equation}
    \begin{equation}\label{eq:depth-k-oracle}
        C_{\mathrm{d}}(k) = 2 C_{\mathrm{d}, \mathrm{load}}(N) + O(k \lg N),
    \end{equation}
where $k$ is the usual monomial degree in the approximating polynomial.
\end{prop}
\begin{proof}
Recall that $\mathcal{U} = \mathcal{F} \mathcal{Z} \mathcal{F}^\dag$, where $\mathcal{F}:=\mathcal{R} (\mathcal{A} \otimes I)$ and $\mathcal{Z} := I - 2 \ketbra{0}{0}$. $\mathcal{Z}$ can be implemented with two 1-qubit gates, plus a $(k \lg N)$-controlled NOT.

Concerning width, $k$ copies of $\ket{\psi_T}$ need to be loaded in parallel for $\mathcal{A}$, leading to $k \lg N$. The $(k \lg N)$-CNOT requires $O(k \lg N)$ ancillas \cite{maslov_advantages_2016}\footnote{For detailed depth and width constants of $n$-CNOT refer to the cited manuscript.}. Finally, $\mathcal{R}$ is a $(k \lg N)$-CNOT as well.

Moving to depth, the data loading is performed in parallel on the different registers for $\ket{\psi_T}$, as well as CNOTs for QHPs. The data loading of $\ket{\psi_E}$ in performed afterwards. The other operations are dominated by the two $(k \lg N)$-CNOTs, which require $O(k \lg N)$ depth.
\end{proof}

\begin{remark}\label{remark:qae-swap}
    For the sake of completeness, let us highlight that the same QAE techniques can be applied to the swap test as well, with a slightly more complex design, without any additional advantage compared to QAE with the ancilla-free method. In this case, indeed, one needs to estimate two quantities: with a first estimation problem, by defining $f_1=1$ when QHPs are successful, one derives the success rate $a_k^{-2}$. Then a second estimation problem is run, by setting $f_2=1$ when both the QHPs are successful and the swap test ancilla provide $1$. Finally, the two quantities are merged into Eq.~\eqref{eq:inner-from-probabilities}, recalling also Eq.~\eqref{eq:successrate}, to estimate the desired inner product.
\end{remark}

We comment the asymptotic performance of this technique, in comparison with the others, in Appendix~\ref{appendix:complexity}.

\section{Bidirectional Orthogonal Encoding of Data}\label{appendix:bidir}
In this section, we define the novel concept of Bidirectional orthogonal encoding (BOE) of data. Then we derive the same results covered in Appendices~\ref{appendix:inner-prod} and~\ref{appendix:qae-ampl}, but this time in the context of the bidir-orth encoding.

\subsection{Towards the BOE}\label{subsec:bidir}
Let $[D_j]$ be a vector of length $N$, a power of $2$. Suppose an angle tree is available (see Fig.~\ref{fig:binary-tree}). In this setting, multiple encodings are possible: (a) the amplitude encoding, (b) the D\&C encoding, (c) the bidirectional encoding, and (d) the D\&C-orth encoding. Let us describe these encodings shortly, thus explaining why we introduced an additional one, that we call Bidirectional Orthogonal Encoding (BOE).
\begin{figure}
    \centering
        \includegraphics[page=1]{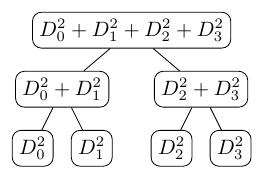}
        \hfill
        \includegraphics[page=2]{Figures/binary-tree.pdf}
    \caption{Two classical binary tree representations of a data set \cite{araujo_configurable_2022}: on the left, the state decomposition representation, and on the right the angle representation. The state decomposition can be built bottom up starting from a classical array, and also applies to unnormalized data sets. The angle representation is specifically suited for quantum data loading, and can be derived travelling the state decomposition tree top-down. Dashed nodes are redundant, since they can inferred from their sibling.}
    \label{fig:binary-tree}
\end{figure}

We are already familiar with the \textit{amplitude encoding}, that produces the state
\begin{equation}\label{eq:psiD-ampl-appendix}
     \ket{\psi_D} := \sum_{j=0}^{N-1} D_j \ket{j}.
\end{equation}
It has the advantage of requiring only $\lg N$ qubits, but unfortunately needs $O(N)$ depth for exact loading.

The \textit{D\&C encoding} is a variant of the analog encoding, recently proposed in Ref.~\cite{araujo2021divide}. We call it divide-and-conquer encoding, or shortly D\&C encoding, after the paper title. In this case, the state produced is
\begin{equation}\label{eq:psiD-dq-appendix}
     \ket{\tilde \psi_D} := \sum_{j=0}^{N-1} D_j \ket{j} \ket{\tilde\phi_j}.
\end{equation}
The idea is to resort to an additional register, that contains auxiliary qubits, entangled with the main register. The advantage of this method is that exact loading can then be performed efficiently, namely in $O\left(\lg^2 N\right)$. The downside is that the required side register is sized $O\left(N\right)$.

This led to the definition of the \textit{bidirectional encoding} \cite{araujo_configurable_2022}, a configurable mixed encoding, that combines the amplitude and the D\&C approaches, and defines a family of encoding techniques parameterized over a so-called split level $s \in \{ 1,...,\lg N\}$, that steers the balance between circuit depth and width. For $s=1$ the encoding coincides with the D\&C, while for $s= \lg N$ the amplitude encoding is retrieved. Finally, for $s=\frac{1}{2} \lg N$, it is possible to achieve a sublinear scaling both in depth and width. The state takes again the form of Eq.~\eqref{eq:psiD-dq-appendix}.

Despite the similarity between Eqs.~\eqref{eq:psiD-ampl-appendix} and~\eqref{eq:psiD-dq-appendix}, and despite the fact that measurement of the primary register provides the same results in both cases, it is essential to remark that D\&C is in fact a different encoding from the amplitude. The algorithms that require the amplitude encoding cannot all be trivially applied to data in the D\&C encoding, and specifically the techniques introduces so far for the calculation of the inner product do not apply in the D\&C encoding. Even more so, they cannot be employed in the bidirectional encoding.

Here the \textit{D\&C-orth encoding} comes to the aid. As the original paper \cite{araujo2021divide} shows, the D\&C encoding can be modified to guarantee that the auxiliary states are orthonormal, i.e. $\braket{\phi_i}{\phi_j} = \delta_{i,j}$, at the expense of an additional side register of small width $\lg N$. The new encoding is relevant for us, since it is compatible with the application of the swap test, that does not provide the same result as in the amplitude encoding (see Prop.~\ref{prop:swap-bidir}), but is still useful for the calculation of the inner product.

Combining these elements, we define the \textit{Bidirectional Orthogonal Encoding} (BOE) in the following way:
\begin{equation}\label{eq:psiD-bidir-appendix}
     \ket{\tilde \psi_D} := \sum_{j=0}^{N-1} D_j \ket{j} \ket{\tilde \phi_j} \ket{j} = \sum_{j=0}^{N-1} D_j \ket{j} \ket{\phi_j}.
\end{equation}
where $\sum_{j=0}^{N-1} D_j \ket{j} \ket{\tilde \phi_j}$ is constructed the bidirectional encoding, and the additional third register obviously guarantees orthonormality of $\left\{ \ket{\phi_j} \right\}_j$.

When $s=1$, the D\&C-orth encoding is retrieved.

\begin{prop}[Circuit depth and width of the BOE]\label{prop:width-depth-bidir}
    Let $[D_j]$ be a vector of length $N$, a power of $2$, and suppose a classical binary tree representation is available. Let $s$ be any integer in $\{1, \dots, \lg N \}$ called split level.
    Then the state in Eq.~\eqref{eq:psiD-bidir-appendix} can be constructed in depth $2^s + \frac{1}{2}\left( \lg^2 N - \lg N - s^2 +s \right) + 1 = O \left( 2^s +\lg^2 N -s^2 \right)$ and width $(s+1) N 2^{-s} -1 + \lg N = O \left( (s+1)N 2^{-s} \right)$.
\end{prop}
\begin{proof}
    It is known \cite{araujo_configurable_2022} that $\sum_{j=0}^{N-1} D_j \ket{j} \ket{\tilde \phi_j}$, namely the associated bidirectional encoding, can be constructed in depth $2^s + \frac{1}{2}\left( \lg^2 N - \lg N - s^2 +s \right)$ and width $(s+1) N 2^{-s} -1$. The conclusion is then trivial.
\end{proof}
For $s=\frac{1}{2} \lg N$, both width and depth are sublinear.

\begin{prop}[Swap test in the BOE]\label{prop:swap-bidir}
Let $\left[D^{(0)}_j\right]$ and $\left[D^{(1)}_j\right]$ be two vectors of length $N$, a power of $2$, represented in the BOE with split level $s$. Apply the swap test between the primary register of the two statevectors (see Fig.~\ref{fig:swaptest}). Then
\begin{enumerate}
    \item 
The swap test qubit is measured in the state $\ket{1}$ with probability $\frac{1}{2}+\frac{1}{2} \sum_j \left|D^{(0)}_j\right|^2 \left|D^{(1)}_j \right|^2$.

    \item
Let $\epsilon>0$ and $\alpha \in (0,1)$. Let $X_i$, for $i=1,...,S$, be a r.v. representing the output of the swap-test measurement after the $i$-th shot of circuit. Call $\bar X_S$ the mean r.v. resulting from the $S$ independent shots. Then $Y_S := 1-2 \bar X_S$ is an estimator for $p=\sum_j \left|D^{(0)}_j\right|^2 \left|D^{(1)}_j \right|^2$ and the error is controlled by
\begin{equation}\label{eq:errorbound-clt-bidirswap}
    \mathbb{P}\left( \abs{ Y_S - p} < \epsilon \right) = \alpha,
\end{equation}
once $S$ is chosen as
\begin{equation}\label{eq:errorbound-clt-shots-bidirswap}
    S = \frac{1-p^2}{\epsilon^2} \left[ \Phi^{-1} \!\left(\frac{1+\alpha}{2}\right) \right]^2
\end{equation}
asymptotically when $S \to \infty$, where $\Phi$ is the CDF of the standard normal distribution.
\end{enumerate}
\end{prop}
\begin{proof}
The proof of the first claim is the same as that of the swap test for the D\&C encoding, see~\cite{araujo2021divide}.
For the second claim, consider that $X_i$ are i.i.d. Bernoullis with mean $\frac{1}{2}-\frac{1}{2}p$ and variance $\sigma^2 = \mu(1-\mu) = \frac{1}{4}-\frac{1}{4}p^2$. By the Central Limit Theorem, $\frac{\bar X_S-\mu}{\sigma/\sqrt{S}}$ is asymptotically a standard normal. Therefore $\frac{Y_S-p}{2 \sigma/\sqrt{S}} = \frac{Y_S-p}{\sqrt{1-p^2} /\sqrt{S}} $ is asymptotically a standard normal as well.
\end{proof}

\subsection{Data encoding and inner product}\label{subsec:bidir-inner}
The previous proposition motivates us to load the normalized versions of $\sqrt{T'-\eta}$ and $\sqrt{E'}$, under the Assumption~\ref{ass:innerprod-positiveterms} that all terms are positive, so that the estimator provides a normalized version of the desired inner product. More specifically, define

\begin{equation}\label{eq:defn-TandE-bidir}
    \tilde T_j = \tilde \rho_T \sqrt{T'_j - \eta}, \qquad
    \tilde E_j = \tilde \rho_E \sqrt{E'_j},
\end{equation}
where
\begin{equation}
    \label{eq:defn-rho-TandE-bidir}
    \tilde \rho_T^{-2} = \sum_{j=0}^{N-1} (T'_j-\eta), \qquad
    \tilde \rho_E^{-2} = \sum_{j=0}^{N-1} E'_j.
\end{equation}

In other words, we start with two BOE states
\begin{equation}\label{eq:psiTandE-bidir}
     \ket{\tilde\psi_{\tilde T}} := \sum_{j=0}^{N-1} \tilde T_j \ket{j} \ket{\phi^T_j}, \qquad
     \ket{\tilde\psi_{\tilde E}} := \sum_{j=0}^{N-1} \tilde E_j \ket{j} \ket{\phi^E_j}.
\end{equation}

It is trivial to verify that the QHP applied to the primary registers of multiple copies of $\ket{\tilde\psi_{\tilde T}}$ still calculates the power and outputs
\begin{equation}\label{eq:psiTpower-bidir}
     \ket{\tilde\psi_{\tilde T}^{(k)}} := \tilde a_k \sum_{j=0}^{N-1} \tilde T_j^k \ket{j} \ket{\phi^T_j}
\end{equation}
when successful, preserving orthonormality of the side register, where $\tilde a_k$ is the appropriate scale factor
\begin{equation}\label{eq:defn_ak-bidir}
    \tilde a_k:=\left(\sum_j \tilde T_j^{2k} \right)^{-\sfrac{1}{2}}.
\end{equation}
and the success rate of the QHP is $\tilde a_k^{-2}$.
\begin{remark}
This setting can be exploited to estimate
\begin{equation}
    \tilde y_k := \sum_{j=0}^{N-1} \tilde E_j^2 \tilde T_j^{2k} = \tilde\rho_E^{2} \tilde\rho_T^{2k} \sum_{j=0}^{N-1} E'_j (T'_j-\eta)^{k} = \sum_{j=0}^{N-1} y'_k,
\end{equation}
as we demonstrate shortly. Before moving to that, it is worth underlying some differences between the newly introduced $\tilde a_k, \tilde y_k$ and the previous $a_k, y_k$. First of all, $\tilde y_k$ scales quadratically with factors applied to $[\tilde T_j^k]$ or $[\tilde E_j]$, while $y_k$ scales linearly with $[T_j^k]$ and $[E_j]$. We will comment later on this fact, that has a major impact on performance (see Prop.~\ref{prop:assemble-k-bidir}). Secondly, and coherently with the former observation, the bounds $0 \leq a_k y_k \leq 1$ are now replaced by $0 \leq \tilde a_k^2 \tilde y_k \leq 1$. Indeed:
\begin{equation}\label{eq:ykak2bound}
    \tilde y_k = \sum \tilde E_j^2 \tilde T_j^{2k} \leq \norm{E_j^2} \ \norm{T_j^{2k}} = \norm{E_j}_4^2 \ \norm{T_j^k}_4^2 \leq \norm{E_j}^2 \ \norm{T_j^k}^2 = \tilde a_k^{-2}.
\end{equation}
Third, in the context of the BOE, Eq.~\eqref{eq:finalsum} is replaced by:
\begin{equation}\label{eq:finalsum-bidir}
    v \approx v^* = \sum_{k=0}^{K} \; b_k(\eta) \; y'_k = \sum_{k=0}^{K} \tilde\rho_T^{-2k} \; \tilde\rho_E^{-2} \; b_k(\eta) \; \tilde y_k.
\end{equation}
Correspondingly, the estimator for $v^*$ is defined as
\begin{equation}\label{eq:estimatorVtilde}
    \tilde V := \sum_{k=0}^K \rho_T^{-2k} \; \rho_E^{-2} \; b_k(\eta) \; \tilde Y_k
\end{equation}
Even though the estimation method is different, $y'_k$ and therefore $v^*$ are the same as in the amplitude encoding.
\end{remark}

\begin{prop}[Algorithm BOE + swap test]\label{prop:errorbound-bidir+qhp+swap}
Let $k$ be a fixed power order. Implement a circuit that produces the state $\ket{\tilde\psi_{\tilde T}^{(k)}}$ defined in Eq.~\eqref{eq:psiTpower-bidir} through QHPs (with or without mid-measurements), then loads $\ket{\tilde\psi_{\tilde E}}$ and applies the swap test between $\ket{\tilde\psi_{\tilde T}^{(k)}}$ and $\ket{\tilde\psi_{\tilde E}}$, as depicted in Fig.~\ref{fig:bidir+qhp+swap}. Call $X \in \{0,1\}$ the output of the measurement of the control qubit in the swap test, and $\mathbf{Z} \in \{0, ..., N-1\}^{k-1}$ the outputs of all the $k-1$ measurements in the QHPs. Define $X_i \sim X$ and $\mathbf{Z}_i \sim \mathbf{Z}$, for $i=1,...,S$, as the outcomes of $S$ independent samples from the circuit. Let
\begin{equation*}
\tilde Y_S := \frac{2 \# \{i : X_i = 0, \mathbf{Z}_i = \mathbf{0} \} - \# \{i: \mathbf{Z}_i = \mathbf{0}\}}{S}; \qquad Y'_S := \tilde\rho_T^{-2k} \tilde\rho_E^{-2} Y_S.
\end{equation*}
Then
\begin{enumerate}
    \item
$\mathbb{E}[\tilde Y_S] \to \sum_j \tilde E_j^2 \tilde T_j^{2k} =: \tilde y_k$ when $S\to\infty$ and $\mathbb{E}[Y'_S] \to \sum_j E'_j (T'_j-\eta)^k =: y'_k$ when $S\to\infty$;

    \item
the absolute error for $\tilde Y_S$ is controlled by $\mathbb{P}\left( \abs{ \tilde Y_S - \tilde y_k} < \epsilon \right) \leq \alpha$ once $S$ is chosen as
\begin{equation}\label{eq:errorbound-bidir+qhp+swap:claim2}
    S \geq \max \left\{ 16 \tilde y_k^2 (\tilde a_k^{2}-1), \ 4 \frac{1-\tilde a_k^4 \tilde y_k^2}{\tilde a_k^2} \right\} \frac{1}{\epsilon^2} \left[ \Phi^{-1} \!\left(\frac{3+\alpha}{4}\right) \right]^2
\end{equation}
asymptotically when $\epsilon \to 0$, where $\Phi$ is the CDF of the standard normal distribution;

    \item
$\mathbb{P}\left( \abs{ \tilde Y_S - \tilde y_k} < \epsilon \right) \leq \alpha$ is also guaranteed by the stronger condition
\begin{equation} \label{eq:errorbound-bidir+qhp+swap:claim3}
    S \geq 16 \frac{1}{\epsilon^2} \left[ \Phi^{-1} \!\left(\frac{3+\alpha}{4}\right) \right]^2
\end{equation}
asymptotically when $\epsilon \to 0$;

    \item
any of the conditions in Eqs.\eqref{eq:errorbound-bidir+qhp+swap:claim2} or~\eqref{eq:errorbound-bidir+qhp+swap:claim3} is also sufficient for the error of the originally scaled problem in the following sense: $\mathbb{P}\left( \abs{ Y'_S - y'_k} < \tilde \rho_T^{-2k} \tilde \rho_E^{-2} \epsilon \right) \leq \alpha$.
\end{enumerate}
\end{prop}

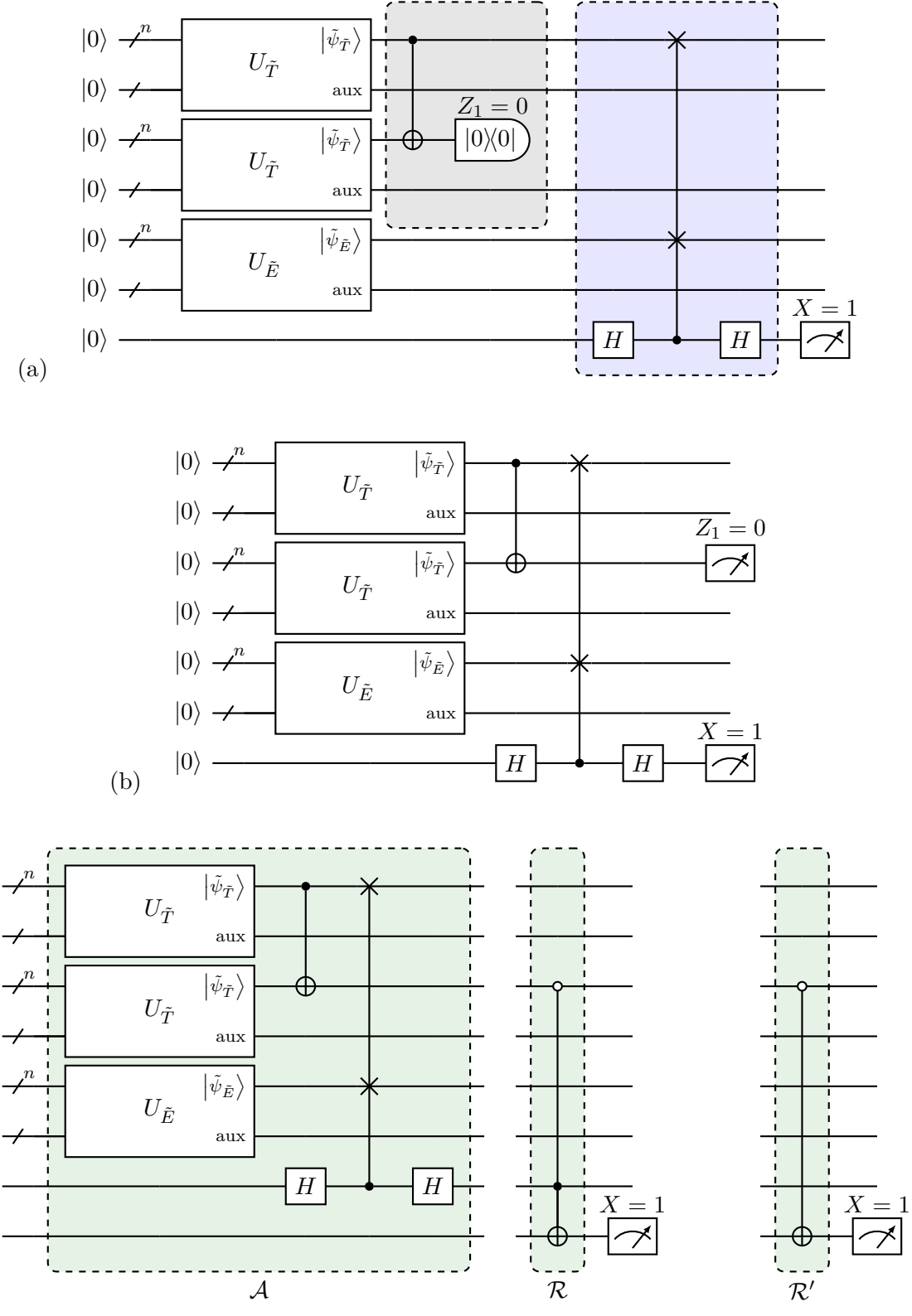
\begin{figure}[p]
    \centering
    (a)
    \tikzsetnextfilename{bidir-qhp-swap1}
    \begin{quantikz}[row sep={0.8cm,between origins}]
    \lstick{$\ket{0}$} & \qwbundle{n} & \gate[wires=2][3cm]{U_{\tilde T} \quad}\gateoutput{$\ket{\tilde\psi_{\tilde T}}$} & \ctrl{2} \gategroup[wires=4,steps=2,style={dashed, rounded corners, fill=black!10}, background]{} & \qw{} & \qw{} & \qw{} \gategroup[wires=7,steps=3,style={dashed, rounded corners, fill=blue!10}, background]{} & \swap{4} & \qw{} & \qw{} \\
    \lstick{$\ket{0}$} & \qwbundle{} & \qw{}\gateoutput{aux} & \qw{} & \qw{} & \qw{} & \qw{} & \qw{} & \qw{} & \qw{} \\
    \lstick{$\ket{0}$} & \qwbundle{n} & \gate[wires=2][3cm]{U_{\tilde T} \quad}\gateoutput{$\ket{\tilde\psi_{\tilde T}}$} & \targ{} & \meterD{\ketbra{0}{0}} \gategroup[wires=1,steps=1,style={draw=none}, label style={label position=above, anchor=north}]{$Z_1=0$} \\
    \lstick{$\ket{0}$} & \qwbundle{} & \qw{}\gateoutput{aux} & \qw{} & \qw{} & \qw{} & \qw{} & \qw{} & \qw{} & \qw{} \\
    \lstick{$\ket{0}$} & \qwbundle{n} & \gate[wires=2][3cm]{U_{\tilde E} \quad}\gateoutput{$\ket{\tilde\psi_{\tilde E}}$} & \qw{} & \qw{} & \qw{} & \qw{} & \swap{2} & \qw{} & \qw{} \\
    \lstick{$\ket{0}$} & \qwbundle{} & \qw{}\gateoutput{aux} & \qw{} & \qw{} & \qw{} & \qw{} & \qw{} & \qw{} & \qw{} \\
    \lstick{$\ket{0}$} & \qw{} & \qw{} & \qw{} & \qw{} & \qw{} & \gate{H} & \ctrl{} & \gate{H} & \meter{$X=1$}
    \end{quantikz}
    \\[8mm]
    (b)
    \tikzsetnextfilename{bidir-qhp-swap2}
    \begin{quantikz}[row sep={0.8cm,between origins}]
    \lstick{$\ket{0}$} & \qwbundle{n} & \gate[wires=2][3cm]{U_{\tilde T} \quad}\gateoutput{$\ket{\tilde\psi_{\tilde T}}$} & \ctrl{2} & \swap{4} & \qw{} & \qw{} \\
    \lstick{$\ket{0}$} & \qwbundle{} & \qw{}\gateoutput{aux} & \qw{} & \qw{} & \qw{} & \qw{} \\
    \lstick{$\ket{0}$} & \qwbundle{n} & \gate[wires=2][3cm]{U_{\tilde T} \quad}\gateoutput{$\ket{\tilde\psi_{\tilde T}}$} & \targ{} & \qw{} & \qw{} & \meter{$Z_1=0$}\\
    \lstick{$\ket{0}$} & \qwbundle{} & \qw{}\gateoutput{aux} & \qw{} & \qw{} & \qw{} & \qw{} \\    
    \lstick{$\ket{0}$} & \qwbundle{n} & \gate[wires=2][3cm]{U_{\tilde E} \quad}\gateoutput{$\ket{\tilde\psi_{\tilde E}}$} & \qw{} & \swap{2} & \qw{} & \qw{} \\
    \lstick{$\ket{0}$} & \qwbundle{} & \qw{}\gateoutput{aux} & \qw{} & \qw{} & \qw{} & \qw{} \\
    \lstick{$\ket{0}$} & \qw{} & \qw{} & \gate{H} & \ctrl{} & \gate{H} & \meter{$X=1$}
    \end{quantikz}
    \\[8mm]
    (c)
    \tikzsetnextfilename{bidir-qhp-swap3}
    \begin{quantikz}[row sep={0.8cm,between origins}]
    \lstick{$\ket{0}$} & \qwbundle{n} & \gate[wires=2][3cm]{U_{\tilde T}}\gateoutput{$\ket{\tilde\psi_{\tilde T}}$} \gategroup[wires=8,steps=4,style={dashed, rounded corners, fill=green!10}, background, label style={label position=below, anchor=north, yshift=-0.2cm}]{$\mathcal{A}$} & \ctrl{2} & \swap{4} & \qw{} & \qw{} & & \qw{} \gategroup[wires=8,steps=1,style={dashed, rounded corners, fill=green!10}, background, label style={label position=below, anchor=north, yshift=-0.2cm}]{$\mathcal{R}$} & \qw{} &[1cm] & \qw{} \gategroup[wires=8,steps=1,style={dashed, rounded corners, fill=green!10}, background, label style={label position=below, anchor=north, yshift=-0.2cm}]{$\mathcal{R'}$} & \qw{} \\
    \lstick{$\ket{0}$} & \qwbundle{} & \qw{}\gateoutput{aux} & \qw{} & \qw{} & \qw{} & \qw{} & & \qw{}& \qw{} & & \qw{} & \qw{} \\
    \lstick{$\ket{0}$} & \qwbundle{n} & \gate[wires=2][3cm]{U_{\tilde T}}\gateoutput{$\ket{\tilde\psi_{\tilde T}}$} & \targ{} & \qw{} & \qw{} & \qw{} & & \octrl{5} & \qw{} & & \octrl{5} & \qw{} \\
    \lstick{$\ket{0}$} & \qwbundle{} & \qw{}\gateoutput{aux} & \qw{} & \qw{} & \qw{} & \qw{} & & \qw{} & \qw{} & & \qw{} & \qw{} \\    
    \lstick{$\ket{0}$} & \qwbundle{n} & \gate[wires=2][3cm]{U_{\tilde E}}\gateoutput{$\ket{\tilde\psi_{\tilde E}}$} & \qw{} & \swap{2} & \qw{} & \qw{} & & \qw{} & \qw{} & & \qw{} & \qw{} \\
    \lstick{$\ket{0}$} & \qwbundle{} & \qw{}\gateoutput{aux} & \qw{} & \qw{} & \qw{} & \qw{} & & \qw{} & \qw{} & & \qw{} & \qw{}  \\
    \lstick{$\ket{0}$} & \qw{} & \qw{} & \gate{H} & \ctrl{} & \gate{H} & \qw{} & & \ctrl{1} & \qw{} & & \qw{} & \qw{} \\
    \lstick{$\ket{0}$} & \qw{} & \qw{} & \qw{} & \qw{} & \qw{} & \qw{} & & \targ{} & \meter{$X=1$} & & \targ{} & \meter{$X=1$}
    \end{quantikz}
    \hspace{1cm}
    \caption{An exemplary demonstration of the algorithm based on BOE data loading, QHP, and the swap test, for the calculation of the inner product between temperatures to the power $k$ and prices, when $k=2$. $U_{\tilde T}$ and $U_{\tilde E}$ are the loading unitaries for $\tilde T$ and $\tilde E$ respectively, in the BOE. In (a), the algorithm in its original formulation, where the gray box highlights the QHP, and the blue box the swap test. In (b), the version without mid-measurements. All measurements are deferred to the end. In (c), a version with single-qubit measurement, suggesting unitaries $\mathcal{A}$ and $\mathcal{R}$ for QAE techniques. To adopt QAE, a separate circuit based on $\mathcal{A}$ and $\mathcal{R'}$ is required to estimate the success rate of QHPs (for this purpose $\mathcal{A}$ can be reduced to a smaller circuit that only computes the QHPs). As usual, we highlighted in gray the QHP, in blue the swap test, and in green the QAE unitaries.}
    \label{fig:bidir+qhp+swap}
\end{figure}
\begin{proof}
The proof follows that in Prop.~\ref{prop:errorbound-qhp+swap}. In particular, the first claim is straight-forward. Concerning the second claim, one defines
$$A_S := 1-2 \frac{ \sum_{i: \mathbf{Z}_i = \mathbf{0}} X_i}{\# \{i: \mathbf{Z}_i = \mathbf{0}\}} \quad \text{and} \quad B_S := \frac{\# \{i: \mathbf{Z}_i = \mathbf{0}\}}{S}.$$
This time $B_S \to \tilde a_k^{-2}$.
Following the proof structure of Prop.~\ref{prop:errorbound-qhp+swap},
\begin{equation*}
    \mathbb{P}\left( \abs{ Y_S - \tilde y_k} < \epsilon \right)
    \geq \mathbb{P}\left( \abs{ A_S B_S - A_S \tilde a_k^{-2}} < \epsilon /2 \right) + \mathbb{P}\left( \abs{A_S \tilde a_k^{-2} - \tilde y_k} < \epsilon /2 \right) -1 \geq 2 \beta -1 = \alpha,
\end{equation*}
once $\beta = \frac{1+\alpha}{2}$ and
\begin{empheq}[left=\empheqlbrace]{align}
    & \mathbb{P}\left( \abs{ A_S B_S - A_S \tilde a_k^{-2}} < \epsilon /2 \right) \geq \beta \label{eq:errorbound-qhp+bidirswap:proof2a} \\
    & \mathbb{P}\left( \abs{A_S - \tilde a_k^2 \tilde y_k} < \tilde a_k^2 \epsilon /2 \right) \geq \beta \label{eq:errorbound-qhp+bidirswap:proof2b}
\end{empheq}

For the latter, apply Prop.~\ref{prop:swap-bidir} to $X_i$ conditioned to $\mathbf{Z}_i = \mathbf{0}$, taking $\tilde a_k^2 \tilde y_k$ as the $p$ in Prop.~\ref{prop:swap-bidir},  $\epsilon \tilde a_k^2 / 2$ as the $\epsilon$ in Prop.~\ref{prop:swap-bidir}, $\# \{i: \mathbf{Z}_i = \mathbf{0}\}$ as the $S$ in Prop.~\ref{prop:swap-bidir}, and $\beta$ as the $\alpha$ in Prop.~\ref{prop:swap-bidir}, and recalling that $\# \{i: \mathbf{Z}_i = \mathbf{0}\}$ is asymptotic to $\tilde a_k^{-2} S$.

For the former instead, since $A_S$ tends to $\tilde a_k^2 \tilde y_k$, it is sufficient to prove
$\abs{B_S - \tilde a_k^{-2}} < \frac{\epsilon}{4 \tilde a_k^2 \tilde y_k}$. $B_S$ is the mean of $S$ i.i.d. Bernoulli variables with $\mu=\tilde a_k^{-2}$ and $\sigma^2 = \mu (1-\mu) = \tilde a_k^{-2} (1- \tilde a_k^{-2})$. Since there are no square roots in this case, the CLT can be applied directly without need for Lemma~\ref{lemma:clt-sqrt},  to verify that the condition in Eq.~\eqref{eq:errorbound-bidir+qhp+swap:claim2} is sufficient. This proves the second claim.

The third one is simple through Eq.~\eqref{eq:ykak2bound}, also remembering $\tilde a_k^{-2} \leq 1$. The fourth is trivial as well.
\end{proof}

\begin{prop}[Circuit width and depth in BOE]\label{prop:width-depth-k-bidir}
    Let $k$ be a fixed power order. Then the algorithm described in Prop.~\ref{prop:errorbound-bidir+qhp+swap} (called BOE + swap test) have the following width and depth:
    $2^s + \frac{1}{2}\left( \lg^2 N - \lg N - s^2 +s \right) + 1$.
    \begin{equation}\label{eq:width-k-bidir}
        C_{\mathrm{w}}(k) = r(k) \left[(s+1) N 2^{-s} -1 + \lg N \right] + 1,
    \end{equation}
    \begin{equation}\label{eq:depth-k-bidir}
        C_{\mathrm{d}}(k) \leq m(k) \left[ 2^s + \frac{1}{2}\left( \lg^2 N - \lg N - s^2 +s \right) \right] + k + (3\lg N + 1) \delta^{(\mathrm{swap})},
    \end{equation}
    where $r(k)$ and $m(k)$ are those defined in Prop.~\ref{prop:width-depth-k} with $\delta^{(\mathrm{swap})}=1$.
\end{prop}
\begin{proof}
    Similar to that of Prop.~\ref{prop:width-depth-k}, exploiting also Prop.~\ref{prop:width-depth-bidir}.
\end{proof}

\begin{prop}[Convergence rate in BOE]\label{prop:assemble-k-bidir}
    Let $w \in [0,1]^K$ such that $\sum_k w_k = 1$. Also, let $\alpha \in [0, 1)^K$ such that $\sum_k (1- \alpha_k) = 1-\beta$ for some $\beta \in (0,1)$. Finally, let $\epsilon > 0$. For instance, one may take $w_k = K^{-1}$ for all $k$ and $\alpha_k = \frac{K-1+\beta}{K}$ for all $k$.

    Then $\tilde V$ defined in Eq.~\eqref{eq:estimatorVtilde} is an estimator for $v^*$ such that
    $\mathbb{P}\left(\abs{V-v^*} \leq \epsilon \right) \geq \beta$
    provided that
    $$\mathbb{P}(\rho_T^{-k} \rho_E^{-1} \abs{b_k(\eta)} \ \abs{Y_k-y_k} \leq w_k \epsilon) \geq \alpha_k \quad \text{for all} k \leq K.$$

    Consequently, a similar estimation holds for the relative error:
    $\mathbb{P}\left(\abs{V-v^*} \leq \epsilon \abs{v} \right) \geq \beta,$ provided that
    $\mathbb{P}\left( \abs{Y_k-y_k} \leq \tilde r_k^{-1} w_k \; \abs{b_k(\eta)}^{-1} \; \epsilon \right)  \geq \alpha_k$ for all $k \leq K$, where
    \begin{equation}\label{eq:defn-rk-bidir}
        \tilde r_k^{-1} := \abs{v} \; \tilde \rho_T^{2k} \; \tilde \rho_E^2.
    \end{equation}

    Finally, let $K$ be fixed. When $N$ grows, $r_K$ dominates the asymptotic behavior of $r_k$ for all other $k \leq K$.
\end{prop}
\begin{proof} Similar to that of Prop.~\ref{prop:assemble-k}.
\end{proof}

\begin{remark}\label{rem:error-scaling-bidir}
    Let us emphasize here that $\tilde r_k$ scales \textit{quadratically} with $\tilde \rho_T^k$, whereas $r_k$ scales linearly with $\rho_T^k$. Under the assumptions of Remark~\ref{rem:error-scaling}, $\tilde \rho_T$ scales as $O(\sqrt{N})$, similarly to $\rho_T$. As a consequence,
    $$\tilde r_k = O(N^{K+1-1}) = O(N^K).$$
    Given the discussion at the end of Subsection~\ref{subsec:assemble-k}, this compares to $r_k = O(N^{K/2-1/2})$ of the amplitude encoding, providing strong limitations to the applicability of the BOE for big values of $N$.
\end{remark}

Complexity of the different techniques is summarized and discussed in Appendix~\ref{appendix:complexity}. For the moment, let us say that it can be improved by resorting to QAE, as done for the amplitude encoding.

\subsection{QAE techniques}

Inspired by Fig.~\ref{fig:bidir+qhp+swap} and similar to Remark~\ref{remark:qae-swap}, one can define two QAE estimation problems, the former to approximate the success rate $\tilde a_k^{-2}$ and the latter to estimate $\frac{1}{2} \left( \tilde y_k + \tilde a_k^{-2} \right)$. As shown in the figure, $\mathcal{A}$ is composed by data loading, QHP, and swap test. Then we define $\mathcal{R'}$ to value an additional qubit as $\ket{1}$ iff all the QHPs were successful, and $\mathcal{R}$ to value an additional qubits $\ket{1}$ iff all the QHPs were successful and the swap test is $\ket{1}$. With these building blocks, we can achieve the following performance.

\begin{prop}[Complexity for BOE + *QHP + swap + QAE]\label{prop:errorbound-xqae-bidir+qhp+direct}
    Let $\mathcal{A}$, $\mathcal{R'}$ and $\mathcal{R}$ as specified right above. Let $\mathcal{U'}$ and $\mathcal{U}$ be the Grover oracles for $\mathcal{A}$ and $\mathcal{R'}$, and $\mathcal{A}$ and $\mathcal{R}$ respectively, according to the usual QAE technique. Let $Y_S^{(1)}$ the median outcome of $S$ executions of a QAE estimation applied to $\mathcal{U'}$, and let $Y_S^{(0)}$ be similarly the median outcome of the QAE technique applied to $\mathcal{A}$ and $\mathcal{R}$.
    Then $Y_S := 2 Y_S^{(1)} - Y_S^{(0)}$ is an estimator for $\tilde y_k := \sum_j \tilde E_j^2 \tilde T_j^{2k}$. 

    Moreover, $\mathbb{P}(\abs{Y_S-\tilde y_k} \leq \epsilon) \geq \alpha$ is obtained by using the Grover oracles $\mathcal{U}$ and $\mathcal{U'}$ for a number of times $O (1/\epsilon)$, and taking $S=O(\lg 1/(1-\alpha))$.

    Additionally, under the same conditions, $\mathbb{P}(\abs{Y'_S-\tilde y'_k} \leq \tilde\rho_T^{-2k} \tilde\rho_E^{-2} \epsilon) \geq \alpha$.
\end{prop}
\begin{proof}
Consider a circuit that performs $\mathcal{U'}$ and measures the last qubit: by definition of $\mathcal{A}$ and $\mathcal{R'}$, the probability to get $1$ is the success rate $\tilde a_k^{-2}$. Similarly, if one executes $\mathcal{U}$ and measures the last qubits, obtains $1$ with probability $\frac{1}{2} \left( \tilde y_k+ \tilde a_k^{-2} \right)$, applying the usual argument of Eq.~\eqref{eq:inner-from-probabilities}.

As a consequence, for Thm.~\ref{thm:qae-bounded-values}, we obtain $\mathbb{P}\left( \abs{Y^{(1)}_S-\tilde a_k^{-2}} \leq \frac{\epsilon}{2} \right) \geq \beta$ by using $\mathcal{U'}$ a number of times $O (1/\epsilon)$, and taking $S=O(\lg 1/(1-\beta))$, where $\beta = \frac{1+\alpha}{2}$. In the same way, $\mathbb{P}\left( \abs{Y^{(0)}_S- \frac{1}{2} \left( \tilde y_k+ \tilde a_k^{-2} \right)} \leq \frac{\epsilon}{2} \right) \geq \beta$ by using $\mathcal{U}$ a number of times $O (1/\epsilon)$, and taking $S=O(\lg 1/(1-\beta))$.

Combining the two estimates, we get $\mathbb{P} \left( \abs{ Y_S - \tilde y_k} \leq \epsilon \right) \geq \alpha$. The proof is complete once we observe that $O(\lg 1/(1-\beta)) = O(\lg 1/(1-\alpha))$.

As far as the originally scaled version of the estimation is concerned, simply substitute $Y'_S$ and $y'_k$ in the previous result.
\end{proof}

\begin{prop}[Oracle depth and width in the BOE]\label{prop:depth-width-oracle-bidir}
    In the setting of the previous Proposition, and assuming a split level $s=s(N) \in \{1, \dots, \lg N \}$, the depth and width of $\mathcal{U}$ are:
    \begin{equation}\label{eq:width-k-oracle-bidir}
        C_{\mathrm{w}}(k) = O \left( (s+1)N 2^{-s} + k \lg N \right),
    \end{equation}
    \begin{equation}\label{eq:depth-k-oracle-bidir}
        C_{\mathrm{d}}(k) = O \left( 2^s +\lg^2 N -s^2 + k \lg N \right),
    \end{equation}
    and they dominate those of $\mathcal{U'}$.
\end{prop}
\begin{proof}
The proof is similar to that of Prop.~\ref{prop:depth-width-oracle}, using the fact that data loading requires resources listed in Prop.~\ref{prop:width-depth-bidir}.
\end{proof}

\subsection{The classical sampling-based algorithm as a benchmark}\label{subsec:tang}
For comparison with the classical case, we recall the following results, that exploit the idea of sample access to design efficient classical quantum-inspired algorithms \cite{tang_quantum-inspired_2019}. Having \textit{sample access} to a real non-null vector $v$, means having a tool that (efficiently) provides the index $j$ of any element from the array, with a probability proportional to its squared value. In other words, it means having access to a random variable $J$ such that $\mathbb{P}(J = j) = v_j^2 / \norm{v}$. Instead, we say that \textit{query access} is given to a vector $v$, if it is possible to obtain $v_j$ from $j$.
\begin{prop}[Sample access {\cite[Prop. 3.2]{tang_quantum-inspired_2019}}]\label{prop:tang-sample}
    The state decomposition tree of a vector of length $N$, described in Fig.~\ref{fig:binary-tree}, also allows for classical sampling from the vector in $O(\lg N)$ time.
\end{prop}
\begin{prop}[Sampling-based inner product {\cite[Prop. 4.2]{tang_quantum-inspired_2019}}]\label{prop:tang-inner}
    Given query access to two real vectors $v, w$, sample access to $v$, and knowledge of $\norm{v}$, the inner product $\sum_j v_j w_j$ can be estimated to additive error $\norm{v} \ \norm{w} \epsilon$ with probability at least $\alpha$ using $O\left( \frac{1}{\epsilon^2} \lg \frac{1}{1-\alpha}\right)$ queries and samples.
\end{prop}

\begin{figure}
    \centering \small
    \include{Figures/iqae_vs_tang_both}
    \caption{IQAE scales quadratically better than the classical sampling-based algorithm described in Prop.~\ref{prop:tang-sample}, in terms of error, when the number of queries grows. Queries are intended as the distribution samples in the classical case, and the oracle calls in the quantum case. Let us remind that D\&C-orth is the same as BOE with $s=1$. The underlying problem is that of calculating the inner product of two fixed vectors of dimension $N=4$, and power $k=1$ in (a), $k=2$ in (b). Quantum runs are executed on a noiseless QASM simulator in Qiskit. Dashed lines are the best-fit lines in the log-log space.}
    \label{fig:iqae-vs-tang}
\end{figure}

The better error scaling of Prop.~\ref{prop:errorbound-iqae-qhp+direct} and Prop.~\ref{prop:errorbound-bidir+qhp+swap}, in contrast with the classical sampling-based version in Prop.~\ref{prop:tang-inner}, is empirically demonstrated in Fig.~\ref{fig:iqae-vs-tang}. 

\begin{remark}\label{rem:error-scaling-tang}
    Let us emphasize that the classical sampling-based algorithm requires sample access to only \textit{one} of the inputs: in our case, we can assume sample access to prices, that do not undergo any transformation. On the contrary, in the quantum case we obviously need to normalize \textit{both} vectors. This translates into the fact that the scaling of the sampling-based algorithm in $N$ does not depend on $k$, since we can assume $\norm{[f(T'_j)]_j} = O(N^{1/2})$ for any $k$, in contrast with the quantum behavior highlighted in Subsection~\ref{subsec:assemble-k}.
\end{remark}

In conclusion, the error scaling of the proposed quantum algorithm is better, but only for fixed vector size $N$.

\section{Comparative complexity analysis}\label{appendix:complexity}
\subsection{Summary of the algorithm variants}\label{subsec:variants}
In the previous appendices we introduced multiple implementations of our approach, generated by different data encodings (amplitude or BOE), different implementations of the QHP (with or without mid-reset), and different techniques for the inner product (ancilla-free method or swap test), as well as by the introduction of *QAE in some cases. Here we collect and compare the main ones, to discuss their complexity. Their key characteristics are also summarized in Fig.~\ref{fig:summary}.

\textbf{(a) Ancilla-free inner product and mid-resets}. The first variant we consider assumes data are available in the amplitude encoding. Each circuit is fed with $k$ copies of the temperature state $\ket{\psi_T}$ in Eq.~\eqref{eq:psiT}, so that the power Eq.~\eqref{eq:psiTpower} of the temperature vector is calculated trough the QHP. Moreover, the same register is reused for the multiple copies of $\ket{\psi_T}$, as shown in Fig.~\ref{fig:qhp-midreset}. At the level of quantum circuit implementation, this variant can benefit from the application of the dynamic stopping technique, as introduced in Subsec.~\ref{subsec:qhp}. Afterwards, under the assumption that a loading unitary $U_E$ for $\ket{\psi_E}$ is known, the inner products $y_k$ can be calculated with the direct method. Finally the volume $v$ is reconstructed classically via Eq.~\eqref{eq:finalsum}. We refer to Subsection~\ref{subsec:qhp+inner} for the detailed implementation and analysis.

\textbf{(b) Ancilla-free inner product and no mid-resets.} In a second variant, we resort to the QHP implementation without mid-circuit resets, as depicted in Fig.~\ref{fig:qhp-nomidreset}, while leaving all other steps unchanged. This results in a shallower circuit, at the cost of more qubits. Details are again in Subsec.~\ref{subsec:qhp+inner}.

In both techniques introduced so far, the ancilla-free method for the estimation of inner products can be replaced with the swap test. We do not extensively discuss these variants here, since, as previously mentioned, the swap test comes with the relevant disadvantage that the sampling complexity is in general unbounded (Fig.~\ref{fig:inner-p}). On the other hand, the swap test does not require knowledge of a loading unitary $U_E$, as it is enough to have access to a copy of the state $\ket{\psi_E}$ at each iteration. This does not appear to be a significant limitation in our use case, where we assume data are classically available and loaded via circuits. The algorithm complexity is also not improved by using the swap test, see Subsec.~\ref{subsec:qhp+inner}.

\textbf{(c) Ancilla-free inner product and *QAE.} The version above without mid-circuit resets inspires another interesting variant: under the assumption that loading unitaries $U_T$ and $U_E$ are available for both vectors, the overall circuit before any measurement is made can be seen as a unitary, that is fed to a Quantum Amplitude Estimation (QAE) technique to obtain a quadratic speedup in the precision $\epsilon$, in line with the general result of QAE. Details are presented in Appendix~\ref{appendix:qae-ampl}. This version (c) is the main variant described in the body of the manuscript.

\textbf{(d) BOE and *QAE.} Additional algorithms can be obtained by resorting to the BOE, as discussed in Appendix~\ref{appendix:bidir}. In this context, the swap test is the only viable option. We focus on the variant that loads data in BOE, applies the QHP for the power calculation, the swap test for the inner product, and boosts precision via QAE techniques.

\subsection{Classical benchmarking algorithms}\label{subsec:benchmarks}
For comparison, we consider three classical algorithms:

\textbf{Exact algorithm} By this term, we refer to the direct calculation of $f(T'_j)$ for all $j$, and of the inner product $\sum_j f(T'_j) E'_j$, without resorting to a polynomial approximation.

\textbf{Classical polynomial approximation} This algorithm is the evaluation of the summation in Eq.~\eqref{eq:finalsum}, where $y_k$ are computed classically.

\textbf{Sampling-based polynomial approximation} In her recent work \cite{tang_quantum-inspired_2019}, Tang showed that inner products can be approximated efficiently on classical computers as well, under the assumption of availability of a binary tree data structure similar to the one needed for the BOE loading for at least one of the two inputs, as further described in Subsec.~\ref{subsec:tang}. With the term `Sampling-based polynomial approximation' we refer to an algorithm made of the efficient calculation of $y_k$ by means of such technique, and then of the usual summation in Eq.~\eqref{eq:finalsum}.

\subsection{Complexity measures}\label{subsec:complex-measures}
We use the following metrics for complexity. The \textit{classical time} $C_{\mathrm{c}}$ is the pre- and post-processing time computed in terms of elementary operations. In case of purely classical benchmarking algorithms, it accounts for the all processing. Classical data is assumed to be given in a specific format, and the cost of data preparation is not included in the pre-processing phase. The \textit{circuit depth} $C_{\mathrm{d}}$ is the depth of a single, independently instantiated run of a circuit shot, measured in terms of layers of 1-qubit gates and CNOTs (in terms of big-O notation, this is equivalent to measure it in terms of 1- and 2-qubit gates).
Notice that $C_{\mathrm{d}}$ may even depend on the single shot, if dynamic stopping is applied: in this case, we can treat it as a random variable. The maximal $C_{\mathrm{d}}$ is an important metric for an approximate estimation of the impact of circuit noise. %
The \textit{sampling complexity} $C_{\mathrm{S}}$ is the number of repeated executions (shots) of each circuit required for the statistical estimation of outputs, and is again indexed on the circuit parameters $k$ and $l$. The \textit{oracle complexity} $C_{\mathrm{o}}$, which appears in the discussion of the classical sampling-based algorithm or when QAE techniques are involved, represents the number of calls made to a black-box operation by a general-purpose method. We also define the \textit{quantum time} $C_{\mathrm{q}}$ as the sum of circuit depths multiplied by the respective sampling complexities, which gives an approximate measure of the overall quantum circuit execution cost under the assumption that all circuit layers have similar duration.
We summarize the sum of quantum time and classical time as \textit{time scaling}, and for this specific metric, we add the additional assumption that all norms scale as $O(\sqrt{N})$ (see Subsection~\ref{subsec:assemble-k}), instead of keeping the explicit contribution of norms, so that the scaling in $N$ becomes more clearly readable.
Finally, the \textit{circuit width} $C_{\mathrm{w}}$ denotes the spatial requirements of an algorithm, as measured by the total number of qubits necessary for its implementation.

\subsection{Space and time complexity}\label{subsec:complexity}
\newcounter{tablenotesComparison}
\renewcommand{\thetablenotesComparison}{\alph{tablenotesComparison}}
\newcommand{\tablenoteItem}[1]{\refstepcounter{tablenotesComparison}\label{#1} \item[\thetablenotesComparison]}
\newcommand{\tnoteRef}[1]{\tnote{\ref{#1}}}
\begin{table*}[p]
    \newcolumntype{Y}{>{\centering\arraybackslash}X}
    \centering
    \small
    \begin{threeparttable}[b]
    \begin{tabularx}{\textwidth}{p{11mm}p{1mm}cYYY} %
    \toprule
	\multicolumn{3}{c}{\textit{Classical algorithms}} & Classical exact & Classical polynomial approximation & Sampling-based polynomial approximation\\
 
	\midrule
    \multicolumn{3}{c}{Classical Encoding} &
    Array &
    Array &
    Binary Tree
    \\
    \cline{4-6}

    \multicolumn{3}{c}{Oracle complexity $C_{\mathrm{o}}(k)$} &
    N/A &
    N/A &
    $O_\alpha(\rho_E^2 \rho_T^2 \abs{v} \epsilon^{-2})$ \tnoteRef{tablenotes:tang}
    \\
    \cline{4-6}

    \multirow{3}{*}{Summary} &
    \multirow{3}{*}{$\left\lbrace\begin{array}{l} \! \\ \! \\ \! \end{array}\right.$} &
    Classical time $C_{\mathrm{c}}$ &
    $O(N)$ \tnoteRef{tablenotes:classical} &
    $O(KN)$ \tnoteRef{tablenotes:classical} &
    $O(\lg N) \; O_\beta(r_1^2 \epsilon^{-2}) + O_K(1)$ \tnoteRef{tablenotes:tang}\tnoteRef{tablenotes:bintree-assumption}\tnoteRef{tablenotes:kprepost}
    \\
    \cline{4-6}

    & &
    Time scaling $C_{\mathrm{c}}$ \tnoteRef{tablenotes:timescaling} &
    $O(N)$ &
    $O(N)$ &
    $O_{\beta, K}(\epsilon^{-2} \lg N)$
    \\
	\bottomrule
    \end{tabularx}
    
    \begin{tabularx}{\textwidth}{p{11mm}p{1mm}cYY} %
    \toprule
	\multicolumn{3}{c}{\textit{QAE-free algorithms}} & (a) Ancilla-free and mid-resets & (b) Ancilla-free, no mid-resets \\
 
	\midrule
    \multicolumn{3}{c}{Classical Encoding} &
    Various &
    Various
    \\
    \cline{4-5}

    \multicolumn{3}{c}{Quantum Encoding} &
    Amplitude &
    Amplitude
    \\
    \cline{4-5}

    \multirow{3}{*}{Circuit} &
    \multirow{3}{*}{$\left\lbrace\begin{array}{l} \! \\ \! \\ \! \end{array}\right.$} &
    Depth $C_{\mathrm{d}}(k)$ &
    $\leq C_{\mathrm{d}, \mathrm{load}}(N) k + k$ \tnoteRef{tablenotes:width-depth-k} &
    $2 C_{\mathrm{d}, \mathrm{load}}(N) + k$ \tnoteRef{tablenotes:width-depth-k}
    \\
    
    & &
    Samples $C_{\mathrm{S}}(k)$ &
    $O_\alpha \left( \rho_E^2 \rho_T^{2k} (y'_k)^{-2} \epsilon^{-2} \right)$ \tnoteRef{tablenotes:direct-scaling}&
    $O_\alpha \left( \rho_E^2 \rho_T^{2k} (y'_k)^{-2} \epsilon^{-2} \right)$ \tnoteRef{tablenotes:direct-scaling}
    \\
    
    & &
    Width $C_{\mathrm{w}}(k)$ &
    $2 \lg N$ \tnoteRef{tablenotes:width-depth-k}&
    $k \lg N$ \tnoteRef{tablenotes:width-depth-k}
    \\
    \cline{4-5}

    \multirow{4}{*}{Summary} &
    \multirow{4}{*}{$\left\lbrace\begin{array}{l} \! \\ \! \\ \! \\ \! \end{array}\right.$} &
    Classical time $C_{\mathrm{c}}$ &
    $2 C_{\mathrm{c}, \mathrm{load}}(N) +O_K(1)$ \tnoteRef{tablenotes:classicaltime}\tnoteRef{tablenotes:kprepost}&
    $2 C_{\mathrm{c}, \mathrm{load}}(N) +O_K(1)$ \tnoteRef{tablenotes:classicaltime}\tnoteRef{tablenotes:kprepost}
    \\

    & &
    Quantum time $C_{\mathrm{q}}$ &
    $C_{\mathrm{d}, \mathrm{load}}(N) \; O_{\beta, K, b}(r_K^2 \epsilon^{-2})$ \tnoteRef{tablenotes:quantumtime} &
    $C_{\mathrm{d}, \mathrm{load}}(N) \; O_{\beta, K, b}(r_K^2 \epsilon^{-2})$ \tnoteRef{tablenotes:quantumtime}
    \\

    & &
    Time scaling $C_{\mathrm{c}} + C_{\mathrm{q}}$ \tnoteRef{tablenotes:timescaling} &
    $O_{\beta, K, b}(C_{\mathrm{c}, \mathrm{load}}(N) + C_{\mathrm{d}, \mathrm{load}}(N) N^{K-1} \epsilon^{-2})$ \tnoteRef{tablenotes:time}&
    $O_{\beta, K, b}(C_{\mathrm{c}, \mathrm{load}}(N) + C_{\mathrm{d}, \mathrm{load}}(N) N^{K-1} \epsilon^{-2})$ \tnoteRef{tablenotes:time}
    \\
	\bottomrule
    \end{tabularx}
    
    \begin{tabularx}{\textwidth}{p{11mm}p{1mm}cYY} %
    \toprule
	\multicolumn{3}{c}{\textit{QAE-based algorithms}} & (c) Ancilla-free and *QAE & (d) BOE and *QAE\\
 
	\midrule
    \multicolumn{3}{c}{Classical Encoding} &
    Various &
    Sqrt Binary Tree
    \\
    \cline{4-5}

    \multicolumn{3}{c}{Quantum Encoding} &
    Amplitude &
    BOE
    \\
    \cline{4-5}
    
    \multirow{3}{*}{Oracle} &
    \multirow{3}{*}{$\left\lbrace\begin{array}{l} \! \\ \! \\ \! \end{array}\right.$} &
    Depth $C_{\mathrm{d}}(k)$ &
    $2 C_{\mathrm{d}, \mathrm{load}}(N) + O_k(\lg N)$ \tnoteRef{tablenotes:width-depth-oracle} &
    $O \left( 2^s +\lg^2 N -s^2 + k \lg N \right)$ \tnoteRef{tablenotes:width-depth-oracle-bidir}
    \\

    & &
    Width $C_{\mathrm{w}}(k)$ &
    $O_k(\lg N)$ \tnoteRef{tablenotes:width-depth-oracle}&
    $O \left( (s+1)N 2^{-s} + k \lg N \right)$ \tnoteRef{tablenotes:width-depth-oracle-bidir}
    \\
    
    & &
    Complexity $C_{\mathrm{o}}(k)$ &
    $O \left( \rho_E \rho_T^k (y'_k)^{-1} \epsilon^{-1} \right)$ \tnoteRef{tablenotes:qae}&
    $O \left( \tilde \rho_E \tilde \rho_T^k (\tilde y'_k)^{-1} \epsilon^{-1} \right)$ \tnoteRef{tablenotes:qae-bidir}
    \\
    \cline{4-5}

    \multirow{3}{*}{Circuit} &
    \multirow{3}{*}{$\left\lbrace\begin{array}{l} \! \\ \! \\ \! \end{array}\right.$} &
    Depth $C_{\mathrm{d}}(k)$ &
    Various \tnoteRef{tablenotes:qae-depth} &
    Various \tnoteRef{tablenotes:qae-depth} 
    \\

    & &
    Width $C_{\mathrm{w}}(k)$ &
    $O_k(\lg N)$ \tnoteRef{tablenotes:width-depth-oracle}&
    $O \left( (s+1)N 2^{-s} + k \lg N \right)$ \tnoteRef{tablenotes:width-depth-oracle-bidir}
    \\
    
    & &
    Samples $C_{\mathrm{S}}(k)$ &
    $O_\alpha(1)$ \tnoteRef{tablenotes:qae}&
    $O_\alpha(1)$ \tnoteRef{tablenotes:qae-bidir}
    \\
    
    \cline{4-5}

    \multirow{5}{*}{Summary} &
    \multirow{5}{*}{$\left\lbrace\begin{array}{l} \! \\ \! \\ \! \\ \! \\ \! \end{array}\right.$} &
    Classical time $C_{\mathrm{c}}$ &
    $2 C_{\mathrm{c}, \mathrm{load}}(N) +O_K(1)$ \tnoteRef{tablenotes:classicaltime}\tnoteRef{tablenotes:kprepost} &
    $O_K(1)$ \tnoteRef{tablenotes:bintree-assumption}\tnoteRef{tablenotes:kprepost}
    \\

    & &
    Quantum time $C_{\mathrm{q}}$ &
    $C_{\mathrm{d}, \mathrm{load}}(N) \; O_{\beta, K, b}(r_K \epsilon^{-1})+$ $O_{\beta, K, b}(r_K \epsilon^{-1} \lg N)$ \tnoteRef{tablenotes:quantumtime-qae} &
    $O_{\beta, K, b} \left( \tilde r_K \epsilon^{-1} \left(2^s +\lg^2 N -s^2 \right) \right)$ \tnoteRef{tablenotes:quantumtime-bidir}\tnoteRef{tablenotes:width-depth-oracle-bidir}
    \\

    & &
    Time scaling $C_{\mathrm{c}} + C_{\mathrm{q}}$ \tnoteRef{tablenotes:timescaling} &
    $O_{\beta, K, b}(C_{\mathrm{c}, \mathrm{load}}(N) + \epsilon^{-1} \left[C_{\mathrm{d}, \mathrm{load}}(N) + \lg N \right] N^{\frac{K-1}{2}})$ \tnoteRef{tablenotes:time} &
    $O_{\beta, K, b} \left( \epsilon^{-1} N^K \left(2^s +\lg^2 N -s^2 \right) \right)$
    \tnoteRef{tablenotes:time-bidir} \\
	\bottomrule
    \end{tabularx}
    \begin{tablenotes}
    \tablenoteItem{tablenotes:timescaling} Under the additional assumptions of Remark~\ref{rem:error-scaling}.
    \tablenoteItem{tablenotes:classical} Under the simplification that floating point operations are performed in $O(1)$ time and without precision loss.
    \tablenoteItem{tablenotes:tang} Prop.~\ref{prop:tang-sample}, Prop.~\ref{prop:tang-inner} and Remark~\ref{rem:error-scaling-tang}.
    \tablenoteItem{tablenotes:bintree-assumption} Assuming that the input binary tree is available; otherwise, its preparation requires $O(N)$ time \cite{araujo2021divide}.
    \tablenoteItem{tablenotes:kprepost} The term $O_K(1)$ encompasses the polynomial reconstruction of Eq.~\eqref{eq:finalsum} or~\eqref{eq:finalsum-bidir}.
    \tablenoteItem{tablenotes:width-depth-k} Prop.~\ref{prop:width-depth-k}
    \tablenoteItem{tablenotes:direct-scaling} Prop.~\ref{prop:errorbound-qhp+direct}.
    \tablenoteItem{tablenotes:classicaltime} $C_{\mathrm{c}, \mathrm{load}}$ is the classical pre-processing needed to load a copy of one input vector.
    \tablenoteItem{tablenotes:overalldepth} Easily derived from the circuit depth.
    \tablenoteItem{tablenotes:quantumtime} Use Prop.~\ref{prop:assemble-k} and the fact that quantum time is defined as the product of circuit depth and samples.
    \tablenoteItem{tablenotes:time} Remark~\ref{rem:error-scaling}.
    \tablenoteItem{tablenotes:width-depth-oracle} Prop.~\ref{prop:depth-width-oracle}.
    \tablenoteItem{tablenotes:qae} Prop.~\ref{prop:errorbound-qae-qhp+direct} and Prop.~\ref{prop:errorbound-iqae-qhp+direct}.
    \tablenoteItem{tablenotes:qae-depth} The circuit depth depends on the specific QAE technique, but all the techniques considered share the same overall depth.
    \tablenoteItem{tablenotes:quantumtime-qae} Prop.~\ref{prop:assemble-k} and Prop.~\ref{prop:errorbound-qae-qhp+direct} (Prop.~\ref{prop:errorbound-iqae-qhp+direct}).
    \tablenoteItem{tablenotes:width-depth-oracle-bidir} Prop.~\ref{prop:depth-width-oracle-bidir}.
    \tablenoteItem{tablenotes:qae-bidir} Prop.~\ref{prop:errorbound-xqae-bidir+qhp+direct}.
    \tablenoteItem{tablenotes:quantumtime-bidir} Prop.~\ref{prop:assemble-k-bidir} and Prop.~\ref{prop:errorbound-xqae-bidir+qhp+direct}.
    \tablenoteItem{tablenotes:time-bidir} Remark~\ref{rem:error-scaling-bidir}.
    \end{tablenotes}
    \end{threeparttable}
    \caption{Complexity analysis of algorithm variants proposed in Subsec.~\ref{subsec:variants}, in comparison with the classical benchmarks listed in Subsec.~\ref{subsec:benchmarks}. Refer to Subsec.~\ref{subsec:complex-measures} for a description of the complexity measures. In the case of circuit complexity and oracle complexity measures, the parameters in the analysis are the data set size $N$, the monomial degree $k$ ($k \geq 1$), the target precision $\epsilon$ and confidence level $\alpha$ such that $\mathbb{P}\left( \abs{ Y'_k - y'_k} \leq \epsilon y'_k \right) \geq \alpha$, the success rates $a_k^{-2}$, the target value $y_k$, the scaling factors $\rho_E, \rho_T$, the `tilde version' of the variables above and the split level $s=s(N) \in \{1, ..., \lg N\}$ for the BOE. In the case of summary complexity measures, we use $N$ as before, the polynomial degree $K$ ($K \geq 1$), the target precision $\epsilon$ and confidence level $\beta$ such that $\mathbb{P}\left( \abs{ V - v^*} \leq \epsilon \abs{v} \right) \geq \beta$, the coefficients $b_k(\eta)$ in Eq.~\eqref{eq:contractvalue-approx}, the scaling factors $\rho_E, \rho_T$, the convergence ratios $r_k$ defined in Eq.~\eqref{eq:defn-rk} or $\tilde r_k$ in Eq.~\eqref{eq:defn-rk-bidir}. Remarkably, the error $\epsilon$ is measured against the exact polynomial evaluation, under the assumption that the polynomial itself is a good approximation of the target volume function. Asymptotic estimations are provided for $\epsilon \to 0$ or $N \to \infty$. Constants affecting asymptotic estimates are marked in the subscript of the big $O$ notation: for instance the notation $O_\alpha(\epsilon^{-2})$ is intended for $\epsilon \to 0$ uniformly in $N$, with factors depending only on $\alpha$. For readability in the subscripts we use $b$ for $b_k(\eta)$ and $a$ for $a_k$. It is worth emphasising that all estimations in the table are independent of $a_k$.}
    \label{tab:algo-compare}
\end{table*}
Table~\ref{tab:algo-compare} collects the results obtained so far on the performance of the techniques listed in Subsec.~\ref{subsec:variants}, and for all the metrics defined above. In terms of time scaling in $N$, the best method is variant (c), which is $O(C_{\mathrm{c}, \mathrm{load}}(N) + C_{\mathrm{d}, \mathrm{load}}(N) N^{K/2-1/2} \lg N)$. If $C_{\mathrm{d}, \mathrm{load}}(N) \ll N$, $C_{\mathrm{d}, \mathrm{load}}(N) \ll N^{1/2}$ and $K \leq 2$, then we have a quantum speedup against the classical case (notice that the sampling-based method is not applicable since we are not supposing that a binary tree structure is available). If instead $C_{\mathrm{d}, \mathrm{load}}(N) = O(1)$ and $K = 3$, the quantum complexity is almost comparable to the classical one, in the sense that it is worse only by logarithmic factors.

The following main takeaways can also be observed.
Methods (a) and (b) have the same asymptotic time complexity of the classical techniques when $N$ grows, if $K=2$ and data loading is performed in $O(1)$ depth. Furthermore, the method (d) is not competitive against the classical counterparts in terms on $N$, while keeping better than the sampling-based algorithm in terms of $\epsilon$ for fixed $N$, as shown in Fig.~\ref{fig:iqae-vs-tang}. Finally, note that the space complexity, indicated in the Table as circuit depth, is $O(k \lg N)$ for algorithms (a), (b), (c), with better constants in the QAE-free cases.

We also want to verify our scaling results for variant (c) empirically. For simplicity, instead of showing the quantum time complexity, we deal with the oracle complexity, which constitutes the key factor of the quantum time complexity. Thence, we want to demonstrate that the oracle complexity of (c) scales linearly with $N^{K/2-1/2}$, coherently with Table~\ref{tab:algo-compare} and Remark~\ref{rem:error-scaling}. It is known by Theorem~\ref{thm:qae-bounded-values} that the error tolerance of *QAE scales linearly with its inverse oracle complexity. Therefore, to verify that the oracle complexity of (c) scales linearly with $N^{K/2-1/2}$, we need to show that the tolerance we have to require to the corresponding *QAE scales as the inverse of $N^{K/2-1/2}$, in order to obtain an (approximately) constant relative error $\epsilon$ in output. This property is shown in Fig.~\ref{fig:error-scaling-k-qae} (left pane).

\begin{figure}[t]
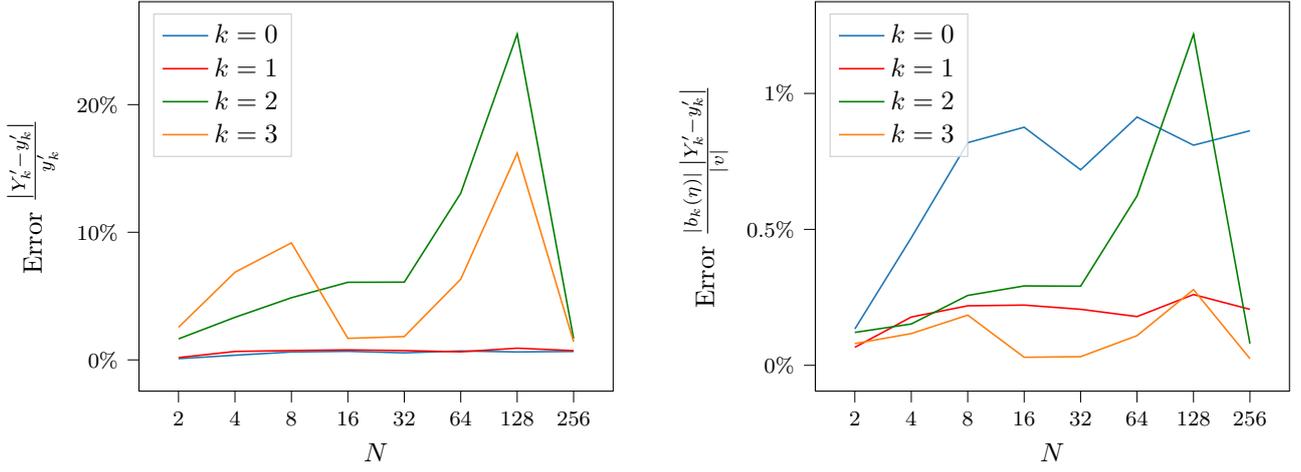

    \centering \small
    \include{Figures/error-scaling-qae}
    \caption{The error of the estimator $Y'_k$, relatively to its own target value $y'_k$ (left), and the error of the power contribution $b_k(\eta) Y'_k$, relatively to the global target $v$ (right), for IQAE with the ancilla-free method.  The IQAE is set to obtain a precision for $Y_k$ of $0.5 \, N^{-k/2+1/2}$, where $0.5$ is an arbitrary constant and the scaling in $N$ is designed to compensate with the error scaling. Each point in the plot is the average of 20 independent runs on the qasm simulator. There is no increasing or decreasing trend in the lines of the plot, consistently with the theory, stating that the error keeps approximately constant if the required precision scales as $r_k = O(N^{k/2-1/2})$. As a comparison, in Fig.~\ref{fig:error-scaling-k} the samples $S$ are kept constant when $N$ grows, and so the error increases, whereas by contrast, here the precision is appropriately scaled with $N$ to keep the relative error constant.}
    \label{fig:error-scaling-k-qae}
\end{figure}

\end{document}